&pdflatex
\documentclass[5p]{elsarticle}

\usepackage{url}
\usepackage{amssymb}
\usepackage{amsmath}
\usepackage{amsthm}
\usepackage{nicefrac}
\usepackage{tabu}
\usepackage{algpseudocode}
\usepackage[section]{algorithm}

\journal{Computer-Aided Design}

\begin{document}

\let\originalleft\left
\let\originalright\right
\renewcommand{\left}{\mathopen{}\mathclose\bgroup\originalleft}
\renewcommand{\right}{\aftergroup\egroup\originalright}

\newcommand{\tablestrutsize}{2.80cm}

\newcommand{\eqnspace}{2ex}

\newcommand{\cpp}{{\nolinebreak C\texttt{++} }}

\newcommand{\BigO}[1]{\mathop{}\!O{\left(#1\right)}}

\newcommand{\Del}[1]{\operatorname{Del}\left(#1\right)}
\newcommand{\Vor}[1]{\operatorname{Vor}\left(#1\right)}
\newcommand{\Conv}[1]{\operatorname{Conv}\left(#1\right)}
\newcommand{\DelS}[1]{\operatorname{Del}|_{\Sigma}\left(#1\right)}
\newcommand{\DelV}[1]{\operatorname{Del}|_{\Omega}\left(#1\right)}
\newcommand{\TS}{\mathcal{T}|_{\Sigma}}
\newcommand{\TV}{\mathcal{T}|_{\Omega}}
\newcommand{\reledge}{h_{r}}
\newcommand{\MAD}{\operatorname{MAD}\left(\theta_{f}\right)}

\newtheorem{proposition}{Proposition}
\newtheorem{corollary}{Corollary}
\newtheorem{lemma}{Lemma}
\newdefinition{definition}{Definition}
\newproof{pf}{Proof}

\begin{frontmatter}

\title{Off-centre Steiner points for Delaunay-refinement on curved surfaces\tnoteref{tnote1}}
\tnotetext[tnote1]{A short version of this paper appears in the proceedings of the 23rd International Meshing Roundtable \cite{engwirda2014face}.}

\author[a]{Darren Engwirda\corref{cor1}\corref{cor2}}
\ead{d.engwirda@maths.usyd.edu.au; engwirda@mit.edu}

\author[a]{David Ivers}
\ead{david.ivers@sydney.edu.au}

\address[a]{School of Mathematics and Statistics F07, University of Sydney, NSW 2006, Australia}

\cortext[cor1]{Corresponding author. Tel.: +1-212-678-5573.}
\cortext[cor2]{Current address: Department of Earth, Atmospheric and Planetary Sciences, Room 54-1517, Massachusetts Institute of Technology, 77 Massachusetts Avenue, Cambridge, MA 02139-4307.}

\begin{abstract}
An extension of the restricted Delaunay-refinement algorithm for surface mesh generation is described, where a new point-placement scheme is introduced to improve element quality in the presence of mesh size constraints. Specifically, it is shown that the use of \textit{off-centre} Steiner points, positioned on the faces of the associated Voronoi diagram, typically leads to significant improvements in the shape- and size-quality of the resulting surface tessellations. The new algorithm can be viewed as a Frontal-Delaunay approach -- a hybridisation of conventional Delaunay-refinement and advancing-front techniques in which new vertices are positioned to satisfy both element size and shape constraints. The performance of the new scheme is investigated experimentally via a series of comparative studies that contrast its performance with that of a typical Delaunay-refinement technique. It is shown that the new method inherits many of the best features of classical Delaunay-refinement and advancing-front type methods, leading to the construction of smooth, high quality surface triangulations with bounded radius-edge ratios and convergence guarantees. Experiments are conducted using a range of complex benchmarks, verifying the robustness and practical performance of the proposed scheme.
\end{abstract}

\end{frontmatter}

\section{Introduction}
\label{section_introduction}

Surface mesh generation is a key component in a variety of computational modelling and simulation tasks, including many forms of computational engineering and numerical modelling, a range of problems in computer graphics and animation, and various data visualisation applications. Given a geometric domain described by a bounding surface $\Sigma$ embedded in $\mathbb{R}^{3}$, the surface \textit{triangulation} problem consists of tessellating $\Sigma$ into a \textit{mesh} of non-overlapping triangular elements, such that all geometrical, topological and user-defined constraints are satisfied. While each use-case contributes its own set of specific considerations, it is typical to require that surface tessellations: (i) consist of elements of high shape-quality, (ii) provide good geometrical and topological approximations to the underlying surface $\Sigma$, and (iii) satisfy a set of user-specified element sizing constraints. While various strategies have been developed to solve the surface meshing problem in the past, a new algorithm is developed in this study with the aim of improving the \textit{quality} of the resulting triangulations.

\subsection{Related Work}

Mesh generation is a broad and evolving area of research. A number of algorithms have been developed to tackle various meshing tasks, including \textit{grid-overlay} methods \cite{Bern90ProvablyGood,Mitchell92Octree} in which a mesh is generated by warping and refining a semi-structured background grid, \textit{advancing-front} techniques \cite{Peraire99Afront,NETGEN,Rypl03HabThesis,Schreiner06Afront} which grow a mesh incrementally in a layer-wise fashion, and \textit{Delaunay-based} strategies \cite{Chew89Provable,Ruppert93Provable,Ruppert95Provable,Shewchuk97PhD,Shewchuk98Tetra,
Cheng03WeightedDelaunay,Cheng10PiecewiseSmoothMeshing,boissonnat03ProvablyGoodSurface,
boissonnat05ProvablyGoodMeshing,jamin2013cgalmesh} which incrementally refine an initially coarse Delaunay triangulation. 

The Delaunay triangulation \cite{Delaunay34Sphere} is imbued with a number of attractive theoretical properties. In the context of mesh generation, for triangulations in $\mathbb{R}^{2}$ and surface tessellations embedded in $\mathbb{R}^{3}$, it is known that such structures are topologically optimal \cite{ChengDeyShewchuk}, maximising the worst-case element quality in the resulting mesh. The reader is referred to \cite{ChengDeyShewchuk} for additional details. As the name suggests, Delaunay-refinement schemes are based on the incremental \textit{refinement} of an initially coarse bounding Delaunay tessellation. At each step of the algorithm, elements that violate a set of geometrical, topological or user-defined constraints are identified and the worst offending element is \textit{eliminated}. This is achieved through the insertion of an additional \textit{Steiner} vertex at the refinement point of the element -- either the circumcentre of the element itself, or a point on an adjacent segment of the bounding geometry. The original planar Delaunay-refinement methods of Chew \cite{Chew89Provable} and Ruppert \cite{Ruppert93Provable,Ruppert95Provable} have since been extended to support volumetric tessellations \cite{Shewchuk98Tetra} and surface triangulations \cite{boissonnat03ProvablyGoodSurface,boissonnat05ProvablyGoodMeshing,
Cheng10PiecewiseSmoothMeshing}.

Delaunay-based methods have previously been applied to the surface meshing problem, via the so-called \textit{restricted} Delaunay paradigm originally introduced by Edelsbrunner and Shah \cite{Edelsbrunner97Restricted}. Given a surface $\Sigma$ embedded in $\mathbb{R}^{3}$ and a set of points $X\in\Sigma$, the restricted Delaunay surface triangulation $\DelS{X}$ is a sub-complex of the full-dimensional Delaunay tessellation $\Del{X}$, containing the set of 2-faces $f\in\Del{X}$ that approximate the underlying surface. Boissonnat and Oudot \cite{boissonnat03ProvablyGoodSurface,boissonnat05ProvablyGoodMeshing} have shown that when the point-wise sampling $X$ is sufficiently \textit{dense} and \textit{well-distributed}, the restricted Delaunay triangulation $\DelS{X}$ is an accurate piecewise approximation of $\Sigma$, incorporating several theoretical guarantees of fidelity. Given such behaviour, a Delaunay-refinement algorithm for surface mesh generation proceeds as per the planar case, with an initially coarse restricted Delaunay triangulation, induced by a sparse sampling $X\in\Sigma$, being incrementally refined via the introduction of new Steiner vertices positioned on the surface $\Sigma$. This process is discussed in detail in Section~\ref{section_delaunay_refinement}. 

Advancing-front techniques are an alternative approach to mesh generation, in which elements are created one-by-one, starting from a set of \textit{frontal} facets initialised on the bounding geometry. The mesh is \textit{marched} inwards layer-by-layer, with new elements created by carefully positioning vertices adjacent to facets in the frontal set. Following the placement of new elements, the set of frontal facets is updated to incorporate the new mesh topology. Elements are added incrementally in this fashion until a complete tessellation of the domain is obtained. Such methods are applicable to surface meshing problems \cite{NETGEN,Rypl03HabThesis,Schreiner06Afront}, although several non-trivial issues of robustness are known to exist \cite{Schreiner06Afront}. Fundamentally, advancing-front techniques are \textit{heuristic} in nature, and do not typically offer theoretical guarantees of convergence or bounds on element quality. Nonetheless, when such methods are successful, they are known to produce very high-quality tessellations.

Due to their theoretical robustness, Delaunay-refine\-ment schemes are a popular choice for practical meshing algorithms. On the other hand, it is well known that optimised advancing-front type methods can often produce meshes with superior element quality characteristics, though their heuristics can sometimes break-down in practice. In this study, a new \textit{Frontal-Delaunay} algorithm is presented which seeks to combine the benefits of classical Delaunay-refine\-ment and advancing-front type approaches. While such strategies have been pursued previously in the case of planar and/or volumetric refinement algorithms, it is believed that the present study is the first application of such ideas to the restricted surface meshing problem directly. The new algorithm is designed to produce smooth, high-quality triangulations consistent with an advancing-front scheme, whilst also inheriting the theoretical robustness of Delaunay-based techniques. It is expected that this new algorithm may be of interest to users who place a high premium on mesh quality, including those operating in the areas of computational engineering and numerical simulation.

A brief review of the traditional surface Delaunay-refine\-ment scheme is presented in Section~\ref{section_delaunay_refinement}, along with a discussion of the underlying theoretical constructs, including the \textit{restricted} Delaunay tessellation. The new Frontal-Delaunay algorithm is presented in Section~\ref{section_frontal_delaunay}, focusing in-detail on the \textit{off-centre} point-placement scheme proposed in this work. In Section~\ref{section_results}, a detailed comparison between the conventional Delaunay-refinement and proposed Frontal-Delaunay schemes is presented, contrasting output quality and computational performance.

\section{Restricted Delaunay-refinement Techniques}
\label{section_delaunay_refinement}

Delaunay-refinement algorithms operate by incrementally adding new Steiner vertices to an initially coarse restricted Delaunay triangulation of the underlying surface. Contrary to typical planar algorithms \cite{Chew89Provable,Ruppert93Provable,Ruppert95Provable}, refinement schemes for surface meshing are designed not only to ensure that the resulting mesh satisfies element shape and size constraints, but that the geometry and topology of the mesh itself is an accurate piecewise approximation of the underlying surface. 

\begin{figure*}[t]
\centering
\caption{Restricted tessellations for a curved domain in $\mathbb{R}^{2}$, showing (i) the curve $\Sigma$ and enclosed area $\Omega$, (ii) the Delaunay tessellation $\Del{X}$ and Voronoi diagram $\Vor{X}$, and (iii) the restricted curve and area tessellations $\DelS{X}$ and $\DelV{X}$.}

\label{figure_restricted_delaunay}

{
\footnotesize
\tabulinesep=0pt

\medskip

\begin{tabu}{ccc}

\begin{minipage}[c]{0.25\textwidth}
\centering
\includegraphics[width=4.25cm]{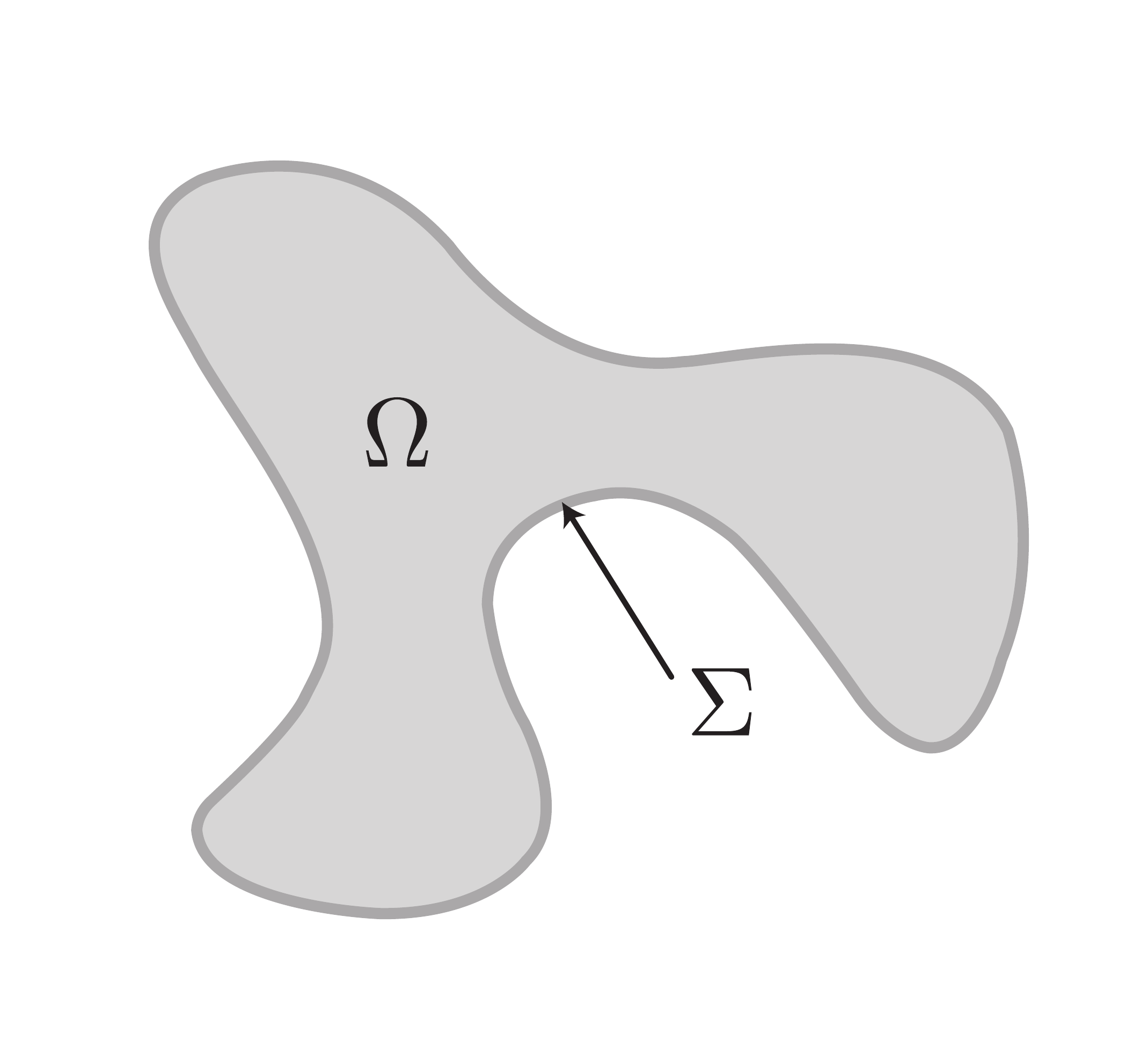} 
\end{minipage} &
\begin{minipage}[c]{0.25\textwidth}
\centering
\includegraphics[width=4.25cm]{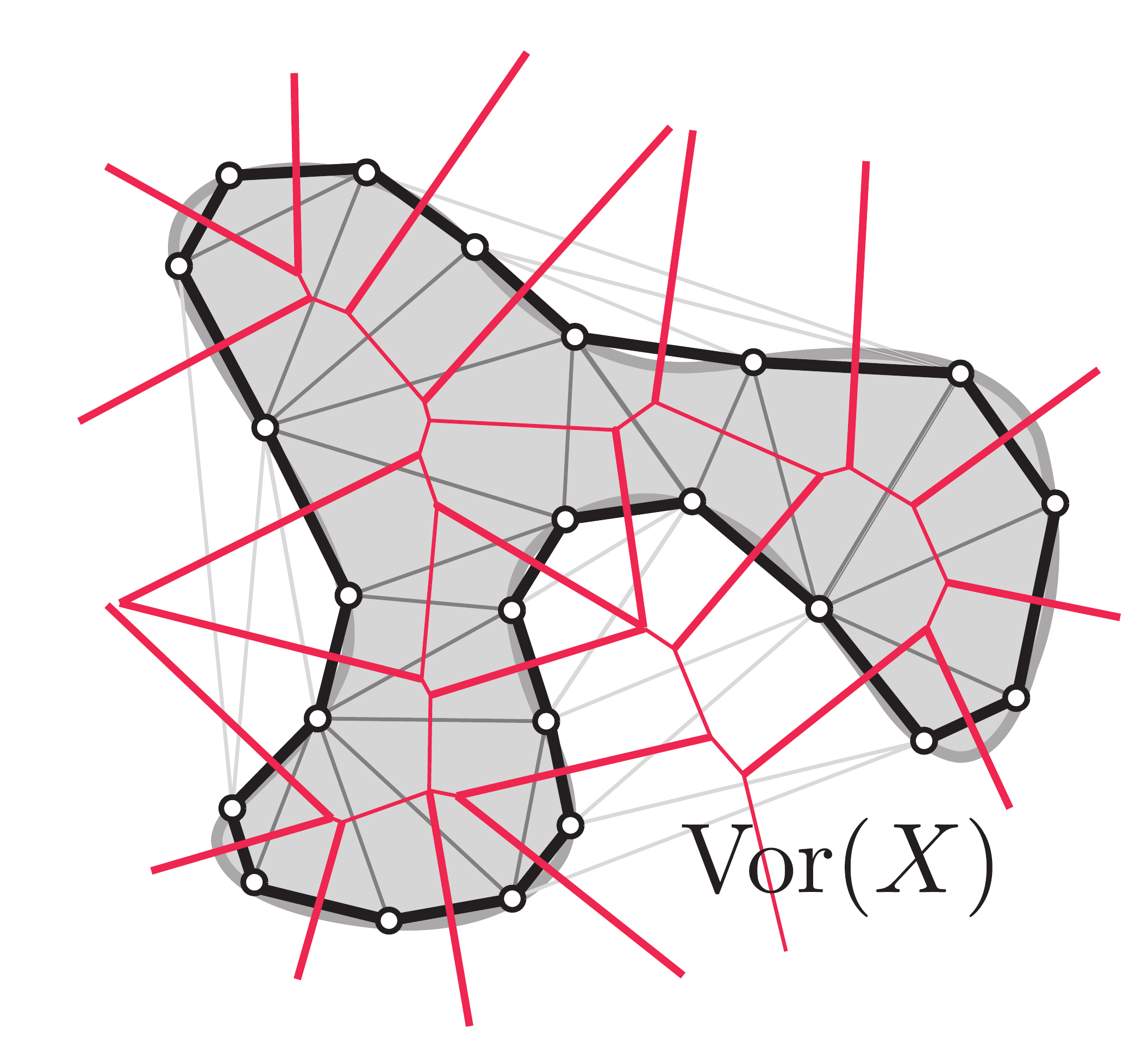} 
\end{minipage} &
\begin{minipage}[c]{0.25\textwidth}
\centering
\includegraphics[width=4.25cm]{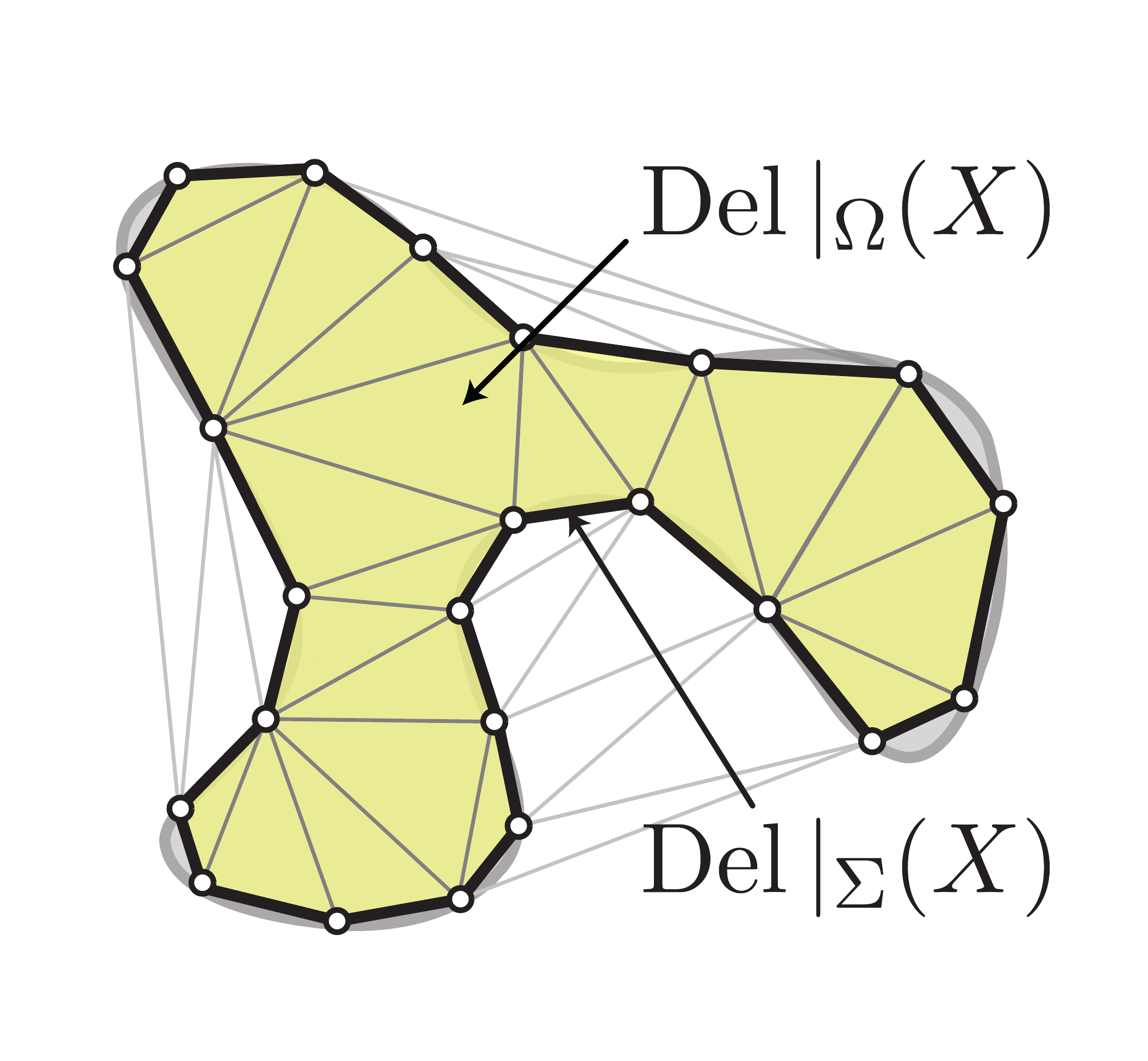} 
\end{minipage} \\

(i) & (ii) & (iii) 

\end{tabu}}
\end{figure*}

Before discussing the details of Delaunay-refinement algorithms for surface meshing problems, a number of important theoretical concepts are introduced:

\begin{definition}
\label{definition_restricted_tessellation}
Let $\Sigma$ be a smooth surface embedded in $\mathbb{R}^3$, enclosing a volume $\Omega\subset\mathbb{R}^3$. Let $\Del{X}$ be a full-dimensional Delaunay tessellation of a point-wise sample $X\subseteq\Sigma$ and $\Vor{X}$ be the Voronoi diagram associated with $\Del{X}$. The \textit{restricted Delaunay surface triangulation} $\DelS{X}$ is a sub-complex of $\Del{X}$ including any 2-face $f\in\Del{X}$ associated with an edge $\mathbf{v}_{f}\in\Vor{X}$ such that $\mathbf{v}_{f}\cap\Sigma\neq\emptyset$. The \textit{restricted Delaunay volume triangulation} $\DelV{X}$ is a sub-complex of $\Del{X}$ including any 3-simplex $\tau\in\Del{X}$ associated with an \textit{internal} circumcentre $\mathbf{c}\in\Omega$.
\end{definition}

It has been shown by a number of authors, including Amenta and Bern \cite{amenta1999surface}, Cheng, Dey, Levine \cite{Cheng08PiecewiseSmoothMeshing},  Cheng, Dey and Ramos \cite{Cheng10PiecewiseSmoothMeshing}, Cheng, Dey, Ramos and Ray \cite{cheng2007sampling}, Boissonnat and Oudot \cite{boissonnat03ProvablyGoodSurface,boissonnat05ProvablyGoodMeshing} and Cheng, Dey and Shewchuk \cite{ChengDeyShewchuk}, that given a \textit{sufficiently dense} point-wise sampling of the surface, $X\in\Sigma$, the restricted Delaunay tessellation $\DelS{X}$ is guaranteed to be both geometrically and topologically representative of the underlying surface $\Sigma$. Specifically, it is known that, when the restricted tessellation is a so-called \textit{loose $\epsilon$-sample} \cite{boissonnat03ProvablyGoodSurface,boissonnat05ProvablyGoodMeshing}, $\DelS{X}$ is homeomorphic to the surface $\Sigma$, the Hausdorff distance $\operatorname{H}\left(\DelS{X},\Sigma\right)$ is small and that $\DelS{X}$ provides a good piecewise approximation of the geometrical properties of $\Sigma$, including its normals, area and curvature. Similar theoretical guarantees extend to the associated restricted volume tessellation $\DelV{X}$. The properties of restricted Delaunay tessellations are well documented in the literature, and a full account is not given here. The reader is referred to the detailed expositions presented in \cite{boissonnat03ProvablyGoodSurface,boissonnat05ProvablyGoodMeshing} and \cite{ChengDeyShewchuk} for further details and proofs.

\begin{definition} 
Let $\DelS{X}$ be a restricted Delaunay triangulation of a smooth surface $\Sigma$. Given a 2-simplex $f\in\DelS{X}$, any circumball of $f$ centred on the surface $\Sigma$ is known as a \textit{surface Delaunay ball} of $f$, denoted $\operatorname{SDB}\left(f\right)$.
\end{definition}

Surface Delaunay balls are centred at the intersection of the Voronoi diagram $\Vor{X}$ with the surface $\Sigma$. Specifically, each 2-face $f\in\DelS{X}$ is associated with an orthogonal edge in the Voronoi diagram $\mathbf{v}_{f}\in\Vor{X}$, which is guaranteed, by definition, to intersect with the surface $\Sigma$ \textit{at least} once. In some cases, especially when the sampling $X\in\Sigma$ is relatively coarse, there may be multiple surface intersections associated with a given 2-face $f\in\DelS{X}$. In such cases, it is typical to consider only the largest associated surface ball. See Figure~\ref{figure_surface_delaunay_ball} for an illustration of the surface Delaunay ball for a general facet $f$.

\begin{definition} 
Let $\DelS{X}$ be a restricted Delaunay surface triangulation of a smooth surface $\Sigma$. Given a 2-simplex $f\in\DelS{X}$, the \textit{surface discretisation error} $\epsilon(f)$ is the Euclidean distance between the centres of the largest surface Delaunay ball of $f$ and its diametric ball.
\end{definition}

The surface discretisation error is a measure of how well the restricted triangulation $\DelS{X}$ approximates the underlying surface $\Sigma$ geometrically. Considering a 2-face $f\in\DelS{X}$, the surface discretisation error $\epsilon(f)$ is a measure of the distance between the face $f$ and the furthest adjacent point on the surface $\Sigma$. This measure can be thought of as a one-sided discrete Hausdorff distance, $\epsilon\left(f\right)=\operatorname{H}\left(\DelS{X},\Sigma\right)$, defined from a set of representative points on the triangulation $\DelS{X}$ to the surface $\Sigma$. Clearly, if the triangulation $\DelS{X}$ is a good piecewise representation of the surface $\Sigma$, the surface discretisation error $\epsilon\left(f\right)$ should be small. See Figure~\ref{figure_surface_delaunay_ball} for a description of the surface discretisation error for a facet $f\in\DelS{X}$.

\begin{definition}
Given a $d$-simplex $\tau$, its \textit{radius-edge} ratio, $\rho\left(\tau\right)$, is given by:
\begin{gather}
\rho\left(\tau\right) = \nicefrac{R}{\|\mathbf{e}_{0}\|},
\end{gather} 
where $R$ is the radius of the diametric ball of $\tau$ and $\|\mathbf{e}_{0}\|$ is the length of its shortest edge.
\end{definition}

The radius-edge ratio is a measure of element shape-quality. It achieves a minimum, $\rho\left(\tau\right)=\nicefrac{1}{\sqrt{3}}$, for equilateral triangles and increases toward $+\infty$ as elements tend toward degeneracy. For 2-simplexes, the radius-edge ratio is \textit{robust} and can be related to the minimum plane angle $\theta_{\text{min}}$ between adjacent edges, such that $\rho\left(\tau\right)={1\over 2}\left(\sin\left(\theta_{\text{min}}\right)\right)^{{-1}}$. Due to the summation of angles in a triangle, given a minimum angle $\theta_{\text{min}}$ the largest angle $\theta_{\text{max}}$ is clearly bounded, such that $\theta_{\text{max}}\leq 180^\circ - 2\theta_{\text{min}}$.

\begin{figure*}[t]
\centering
\caption{The surface Delaunay ball for a restricted 2-face $f\in\DelS{X}$, showing (i) placement of the surface ball at the intersection of the dual edge $\mathbf{v}_{f}\in\Vor{X}$ and the surface $\Sigma$, and (ii) associated radius $r(f)$ and surface discretisation error $\epsilon(f)$.}

\label{figure_surface_delaunay_ball}

{
\footnotesize
\tabulinesep=2pt

\medskip

\begin{tabu} {cc}

\begin{minipage}[c]{0.375\textwidth}
\centering
\includegraphics[width=3.75cm]{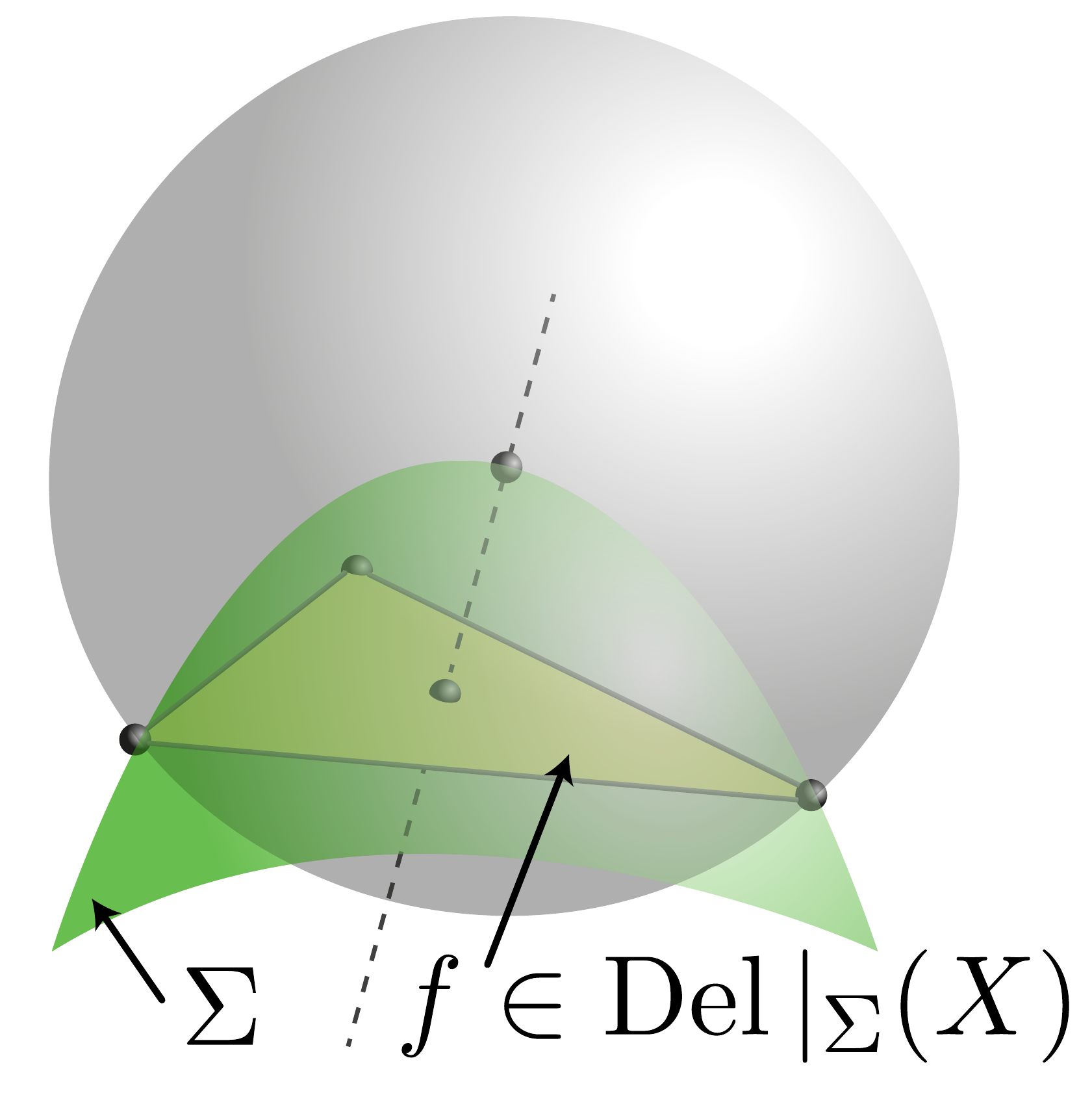} 
\end{minipage} &
\begin{minipage}[c]{0.375\textwidth}
\centering
\includegraphics[width=3.75cm]{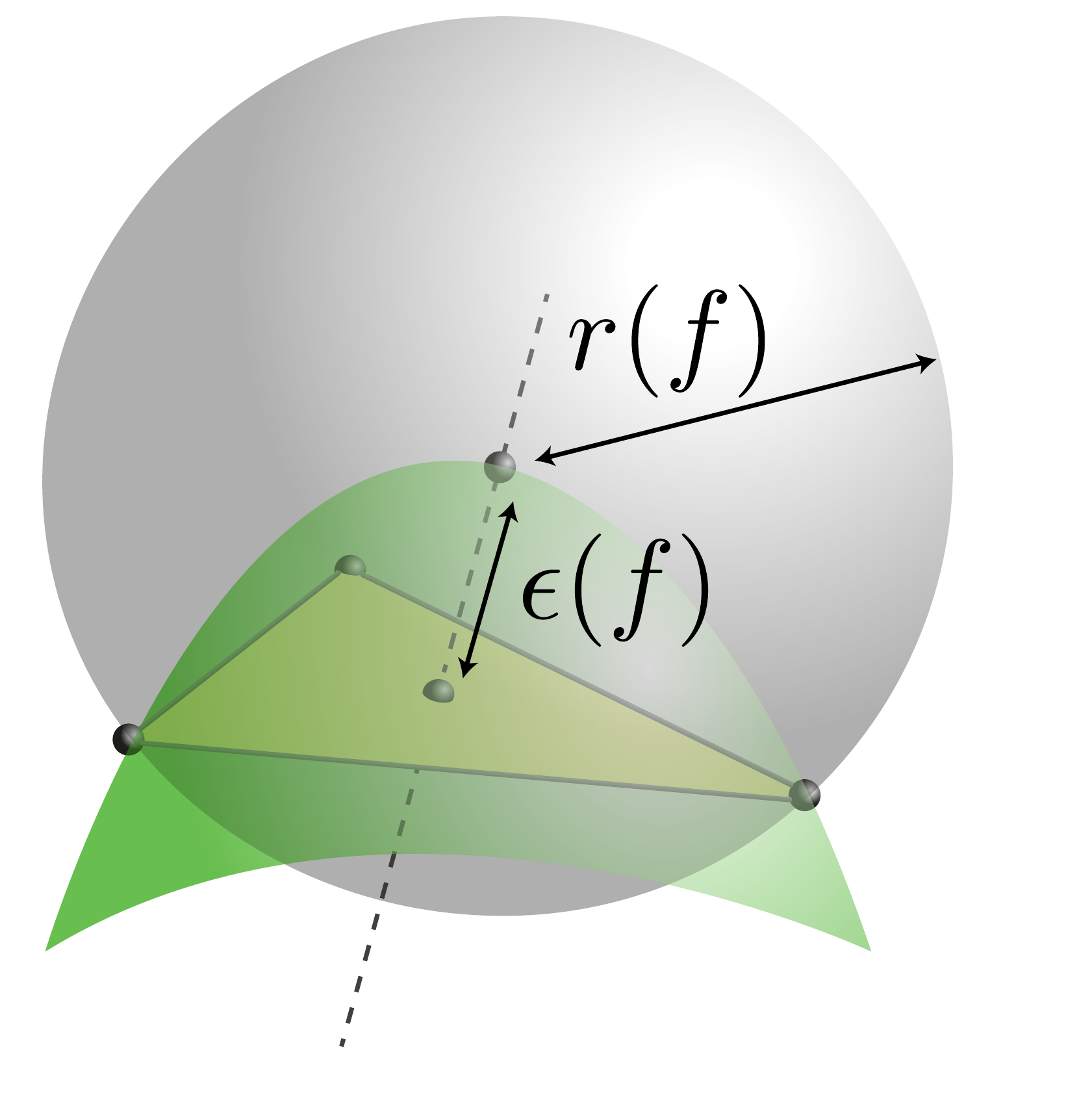} 
\end{minipage} \\

(i) & (ii) 

\end{tabu}}
\end{figure*}

\subsection{An Existing Delaunay-refinement Algorithm}

The development of \textit{provably-good} Delaunay-refine\-ment schemes for surface-based mesh generation is an ongoing area of research. An algorithm for the meshing of closed 2-manifolds embedded in $\mathbb{R}^3$, adapted largely from the methods presented by Boissonnat and Oudot in \cite{boissonnat03ProvablyGoodSurface,boissonnat05ProvablyGoodMeshing} is presented here. This method is largely equivalent to the \textsc{cgalmesh} algorithm, available as part of the \textsc{cgal} package, and summarised by Jamin, Alliez, Yvinec and Boissonnat in \cite{jamin2013cgalmesh}. A similar algorithm is also described by Cheng, Dey and Shewchuk in \cite{ChengDeyShewchuk} and Dey and Levine in \cite{dey2009delaunay}, while a variant developed for the reconstruction of medical images is presented by Foteinos, Chernikov and Chrisochoides in \cite{foteinos2014tetmed}. The algorithm presented in this section is referred to as the `conventional' Delaunay-refinement approach throughout, due to its direct use of circumcentre-based Steiner vertices. 

As per Jamin et~al.~\cite{jamin2013cgalmesh}, the development of the conventional Del\-aunay-refinement algorithm is \textit{geometry-agnostic}, being independent of the specific description used for the underlying surface $\Sigma$. It is required only that the geometry support a so-called \textit{oracle} predicate that can be used to compute the intersection of a given line segment with the surface $\Sigma$. The Frontal-Delaunay algorithm presented in Section~\ref{section_frontal_delaunay}, additionally requires that the oracle compute intersections between a disk of a given radius and the surface $\Sigma$. Such constructions will be discussed in further detail in Section~\ref{section_frontal_delaunay}. While a very general class of surface description is supported at the theoretical level, in this study, attention is restricted to the development of so-called \textit{remeshing} operations, in which the underlying surfaces $\Sigma$ are specified as existing 2-manifold triangular complexes $\mathcal{P}$. This restriction is made in the present study for convenience only, and future work is intended to focus on more general surface descriptions, including implicit and analytic functions, in addition to those that contain sharp creases. 

Following Jamin et~al.~\cite{jamin2013cgalmesh}, the Delaunay-refinement algorithm takes as input a surface domain, described by a smooth 2-manifold $\Sigma$, an upper bound on the allowable element radius-edge ratio $\bar{\rho}$, a mesh size function $\bar{h}\left(\mathbf{x}\right)$ defined at all points on the surface $\Sigma$ and an upper bound on the allowable surface discretisation error $\bar{\epsilon}\left(\mathbf{x}\right)$. The input surface $\Sigma$ encloses a bounded volume $\Omega$. The algorithm returns a triangulation $\TS$ of the surface $\Sigma$, where $\TS$ is a restricted Delaunay surface triangulation of a point-wise sampling $X\in\Sigma$, such that $\TS=\DelS{X}$. As a by-product, the algorithm also returns a coarse triangulation $\TV$ of the enclosed volume $\Omega$, where $\TV$ is a restricted Delaunay volume triangulation $\TV=\DelV{X}$. Both $\DelS{X}$ and $\DelV{X}$ are sub-complexes of the full-dimensional Delaunay tessellation $\Del{X}$. Note that $\DelS{X}$ is a triangular complex, while $\DelV{X}$ and $\Del{X}$ are tetrahedral complexes. The method is summarised in Algorithm~\ref{algorithm_restricted_delaunay_surface}.

The Delaunay-refinement algorithm guarantees, firstly, that all elements in the output surface triangulation $f\in\TS$ satisfy element shape constraints, $\rho\left(f\right)\leq\bar{\rho}$, element size constraints $h\left(f\right)\leq \bar{h}(\mathbf{x}_{f})$ and surface discretisation bounds $\epsilon\left(f\right)\leq\bar{\epsilon}(\mathbf{x}_{f})$. Furthermore, for sufficiently small mesh size functions $\bar{h}\left(\mathbf{x}\right)$, the surface triangulation $\TS$ is guaranteed to be a \textit{good} piecewise approximation of the underlying surface $\Sigma$, exhibiting both geometrical and topological convergence as $\bar{h}\left(\mathbf{x}\right)\rightarrow 0$, consistent with the characteristics of restricted Delaunay tessellations outlined in  Definition~\ref{definition_restricted_tessellation}.

The Delaunay-refinement algorithm begins by creating an initial point-wise sampling of the surface $X_{0}\in\Sigma$. Exploiting the discrete representation available for $\Sigma$, in this study the initial sampling is obtained as a \textit{well-distributed} subset of the existing vertices $Y\in\mathcal{P}$, where $\mathcal{P}$ is the polyhedral representation of the surface $\Sigma$. Care is taken to ensure that each connected component in $\mathcal{P}$ is sampled. Alternative initialisation methods making use of robust \textit{persistent-triangles} techniques are described in \cite{boissonnat05ProvablyGoodMeshing,ChengDeyShewchuk}.

In the next step, the initial triangulation objects are formed. In the present work, the full-dimensional Delaunay tessellation, $\Del{X}$, is built using an incremental Delaunay triangulation algorithm, based on the Bowyer-Watson technique \cite{bowyer81algorithm}. The restricted surface and volumetric triangulations, $\DelS{X}$ and $\DelV{X}$, are derived from $\Del{X}$ by explicitly testing for intersections between edges in the associated Voronoi diagram $\Vor{X}$ and the surface $\Sigma$. These queries are computed efficiently by storing the surface definition $\mathcal{P}$ in a spatial-tree. 

The main loop of the algorithm proceeds to incrementally refine any 2-faces $f\in\DelS{X}$ that violate either the radius-edge, element size or surface discretisation requirements. The refinement process is priority scheduled, with triangles $f\in\DelS{X}$ ordered according to their radius-edge ratios $\rho\left(f\right)$, ensuring that the element with the \textit{worst} ratio is refined at each iteration. Individual elements are refined based on their surface Delaunay balls, with a triangle $f\in\Del{X}$ eliminated by inserting the centre of the largest ball $\operatorname{B}\left(\mathbf{c},r\right)_{\text{max}}=\operatorname{SDB}\left(f\right)$ into the tessellation $\Del{X}$. This process is a direct generalisation of the circumcentre-based insertion method associated with Ruppert's algorithm for planar domains \cite{Ruppert93Provable,Ruppert95Provable}. 

As a consequence of changes to the full-dimensional tessellation following the insertion of a new Steiner vertex, corresponding updates to the restricted triangulations $\DelS{X}$ and $\DelV{X}$ are instigated, ensuring that all tessellation objects remain valid throughout the refinement process. The Delaunay-refinement algorithm terminates when all 2-faces $f\in\DelS{X}$ satisfy all radius-edge, size and surface discretisation thresholds, such that\footnote{The coefficient $\alpha = \nicefrac{4}{3}$, ensuring that the mean element size does not, on average,  undershoot the desired target size $\bar{h}(\mathbf{x}_{f})$.} $\rho\left(f\right)\leq\bar{\rho}$, $h\left(f\right)\leq \alpha \bar{h}(\mathbf{x}_{f})$ and $\epsilon\left(f\right)\leq\bar{\epsilon}(\mathbf{x}_{f})$, respectively, where the element size $h\left(f\right)$ is proportional to the radius of the associated surface ball $\operatorname{B}\left(\mathbf{c},r\right)_{\text{max}}=\operatorname{SDB}\left(f\right)$, such that\footnote{The coefficient $\sqrt{3}$ represents the mapping between the edge length and diametric ball radius for an equilateral element. Such scaling ensures that size constraints are applied with respect to mean edge length.} $h\left(f\right) = \sqrt{3}\,r$, and the target size $\bar{h}(\mathbf{x}_{f})$ is sampled at the centre of the surface ball $\mathbf{x}_{f}=\mathbf{c}$.

\begin{algorithm*}[t]
\caption{Surface mesh generation by restricted Delaunay-refinement}
\label{algorithm_restricted_delaunay_surface}
\centering

{
\footnotesize
\tabulinesep=0pt

\smallskip

\begin{tabu} {c|c}

\begin{minipage}[c]{.45\textwidth}
\begin{algorithmic}[1]

\Function{DelaunaySurface}{$\Sigma,\Omega,\bar{\rho},\bar{\epsilon}\left(\mathbf{x}\right),\bar{h}\left(\mathbf{x}\right),\TS,\TV$}

 \State \parbox[t]{.70\textwidth}{Form an initial sampling $X\in\Sigma$ such that $X$ is \textit{well-distributed} on $\Sigma$.\strut}
 \State Form Delaunay tessellation $\Del{X}$.
 \State Form restricted objects $\DelS{X}$ and $\DelV{X}$.

 \State \parbox[t]{.70\textwidth}{Enqueue all restricted 2-simplexes $Q|_{\Sigma}\gets f\in\DelS{X}$. A 2-simplex $f$ is enqueued if \Call{BadSimplex}{$f$} returns \textsc{true}.\strut}

 \While{$\left(Q|_{\Sigma}\neq\emptyset\right)$}\Comment{\{\textit{main refinement sweeps}\}}
  \State Call \Call{RefineSimplex}{$f\gets Q|_{\Sigma}$}
  \State \parbox[t]{.625\textwidth}{Update the restricted Delaunay tessellations $\DelS{X}$ and $\DelV{X}$.\strut}
  \State Update $Q|_{\Sigma}$ to reflect changes to $\DelS{X}$.

 \EndWhile
 \State\Return $\:\TS=\DelS{X}$ and $\TV=\DelV{X}$
\EndFunction
\end{algorithmic}

\end{minipage} &
\begin{minipage}[c]{.45\textwidth}

\begin{algorithmic}[1]
\Function{RefineSimplex}{$f$}\Comment{\{\textit{surface refinement}\}}
 \State Call \Call{SurfaceDelaunayBall}{$f,\operatorname{B}\left(\mathbf{c},r\right)_{\text{max}}$}. 
 \State Form new Steiner vertex $\mathbf{p}$ about $\operatorname{B}\left(\mathbf{c},r\right)_{\text{max}}$.
 \State Insert Steiner vertex $X\gets\mathbf{p}$, update $\Del{X}\gets X$.
\EndFunction
\end{algorithmic}

\begin{algorithmic}[1]
\Function{SurfaceDelaunayBall}{$f$}\Comment{\{\textit{surface ball}\}}
 \State Form Voronoi edge $\mathbf{v}_{f}$ orthogonal to 2-simplex $f$.
 \State \parbox[t]{.70\textwidth}{Form the set of associated surface Delaunay balls $\operatorname{B}\left(\mathbf{c},r\right)_{i}$ for the restricted 2-simplex $f\in\DelS{X}$. Balls are centred about the set of surface intersections $\mathbf{v}_{f}\cap\Sigma\neq\emptyset$.}
 \State\Return $\operatorname{B}\left(\mathbf{c},r\right)_{\text{max}}$, where $r_{\text{max}}$ is maximal.
\EndFunction
\end{algorithmic}

\begin{algorithmic}[1]
\Function{BadSimplex}{$f$}\Comment{\{\textit{termination criteria}\}}
 \State\Return $\:\rho\left(f\right)>\bar{\rho}$ or 
               $\epsilon\left(f\right)>\bar{\epsilon}(\mathbf{x}_{f})$ or
               $h\left(f\right)>\bar{h}(\mathbf{x}_{f})$
\EndFunction

\end{algorithmic}
\end{minipage}

\end{tabu}}

\smallskip

\end{algorithm*}

\section{Restricted Frontal-Delaunay Methods}
\label{section_frontal_delaunay}

Frontal-Delaunay algorithms are a hybridisation of adv\-ancing-front and Delaunay-refinement techniques, in which a Delaunay triangulation is used to define the topology of a mesh while
new Steiner vertices are inserted in a manner consistent with advancing-front methodologies. In practice, such techniques have been observed to produce very high-quality meshes, inheriting the smooth, semi-structured vertex placement of pure advancing-front methods and the optimal mesh topology of Delaunay-based approaches. While Frontal-Delaunay methods have previously been used by a range of authors in the context of planar, volumetric and parametric surface meshing, including, for example studies by \"Ung\"or and Erten \cite{Ungor09OptSteiner}, Rebay \cite{Rebay93FrontalDelaunay}, Mavriplis \cite{Mavriplis95FrontalDelaunay}, Frey, Borouchaki and George \cite{Frey98FrontalDelaunay}, and Remacle, Henrotte, Carrier-Baudouin, B\'echet, Marchandise,  Geuzaine and Mouton \cite{Remacle13Quads}, the authors are not aware of any previous investigations describing the application of such techniques to the surface meshing problem directly. The conventional advancing-front method, on the other hand, has been generalised to support surface meshing, as, for example, outlined in studies by Rypl \cite{Rypl98PhD,Rypl03HabThesis} and Schreiner, Scheidegger, Fleishman and Silva \cite{Schreiner06Afront,Schreiner06Afront2}. 

In these previous studies, the conventional planar adva\-ncing-front methodology is extended directly to facilitate surface operations. Meshing proceeds via the incremental introduction of a well-distributed set of vertices $X$ positioned on the surface $X\in\Sigma$. It should be noted that, contrary to the restricted Delaunay techniques introduced in previous sections, advancing-front methods typically maintain a partial, and possibly non-Delaunay manifold triangular complex $\TS$ only -- they do not construct a full-dimensional tessellation $\TV$. It should also be noted that, in addition to the usual limitations associated with advancing-front strategies, serious issues of robustness often afflict these techniques in practice due to the difficulties associated with the reliable evaluation of the requisite geometric intersection and overlap predicates for sets of non-planar elements. Additionally, for highly curved and/or poorly separated surface definitions it is known to be difficult to ensure the topological correctness of the output mesh $\TS$. Schreiner et~al.~\cite{Schreiner06Afront,Schreiner06Afront2} introduce a number of heuristic techniques in an effort to overcome these difficulties.

\subsection{Off-centres}

The new Frontal-Delaunay algorithm is based primarily on a generalisation of ideas introduced by Rebay, who, in \cite{Rebay93FrontalDelaunay}, developed a \textit{planar} Frontal-Delaunay algorithm in which new vertices are positioned along edges of the associated Voronoi diagram. Rebay showed that new vertices can be positioned on $\Vor{X}$ according to an \textit{a priori} mesh size function $\bar{h}\left(\mathbf{x}\right)$, a strategy broadly consistent with conventional advancing-front techniques. While his algorithm maintains a Delaunay triangulation $\mathcal{T}=\Del{X}$ of the current vertex set $X$, it is still fundamentally an advancing-front scheme -- bereft of guarantees on the element shape quality $\rho\left(\tau\right)$. Rebay reported that his scheme produced very high-quality output, typically outperforming conventional Delaunay-refinement in practice.

\"Ung\"or and Erten have approached the problem from the flip-side, introducing the notion of \textit{generalised} Steiner vertices for planar Delaunay-refinement methods. Their `off-centre' vertices are points lying along edges in the associated Voronoi diagram, as per Rebay. In \cite{Ungor09OffCenters}, \"Ungor showed that when refining a triangle $\tau$, its off-centre should be positioned, if possible, such that the new triangle $\sigma$ adjacent to the shortest edge in $\tau$ satisfies the desired shape-quality target, such that $\rho\left(\sigma\right)\leq\bar{\rho}$. Importantly, \"Ung\"or demonstrated that such a strategy typically leads to an improvement in the performance of Delaunay-refinement in practice, reducing the size of the output $|\mathcal{T}|$. \"Ung\"or extended the guarantees derived for Ruppert's Delaunay-refinement algorithm to his off-centre technique and showed that such bounds are typically improved upon in practice. Off-centre refinement is based on Ruppert's Delaunay-refine\-ment framework directly, involving modifications to the procedure used to refine low-quality triangles only. 

The use of \textit{generalised} Delaunay-refinement strategies has also been explored by Chernikov, Chrisochoides and Foteinos in \cite{chernikov2012generalized,foteinos2010fully}, in which Steiner points are positioned within a set of so-called \textit{selection-balls} adjacent to element circumcentres. In \cite{chernikov2012generalized}, Chernikov and Chrisochoides show that a family of `provably-good' generalised two- and three-dimensional Delaunay-refinement schemes exist, and can be realised via the specification of appropriate selection-ball radii parameters. It is unclear whether this selection-ball formalism incorporates the full set of Voronoi-type off-centres currently in use, including those introduced by Rebay \cite{Rebay93FrontalDelaunay}, and Erten and \"Ung\"or \cite{Ungor09OptSteiner}, in addition to the strategies defined in the present study. Such a determination may form the basis for future work.

\subsection{Point-placement Preliminaries}

In this study, a generalisation of the ideas introduced by Rebay and \"Ung\"or is formulated for the surface meshing problem -- using off-centre Steiner vertices to simulate the vertex placement strategy of a conventional advancing-front approach, while also preserving the framework of a restricted Delaunay-refinement technique. The aim of such a strategy is to recover the high element qualities and smooth, semi-structured point-placement generated by frontal methods, while inheriting the theoretical guarantees of Delaunay-refinement methods. Advancing-front algorithms typically incorporate a mesh size function $\bar{h}\left(\mathbf{x}\right)$, a function $f : \mathbb{R}^3\rightarrow\mathbb{R}^{+}$ defined over the domain to be meshed, where $\bar{h}\left(\mathbf{x}\right)$ represents the desired edge length $\|\mathbf{e}\|$ at any point $\mathbf{x}\in\Sigma$. This mesh size function typically incorporates size constraints dictated by both the user and the geometry of the domain to be meshed. The construction of appropriate mesh size functions will be discussed in subsequent sections, but for now, it suffices to note that $\bar{h}\left(\mathbf{x}\right)$ is a $g$-Lipschitz function defined at all points on $\Sigma$ with $0<g<1$.

The proposed Frontal-Delaunay algorithm is an extension of the restricted Delaunay-refinement algorithm presented in Section~\ref{section_delaunay_refinement}, modified to use off-centre rather than circumcentre-based refinement strategies. The basic framework is consistent with the Delaunay-refinement algorithm described previously, in which an initially coarse restricted Delaunay triangulation of a surface $\Sigma$ is refined through the introduction of additional Steiner vertices $X\in\Sigma$ until all constraints are satisfied. A restricted surface triangulation $\TS=\DelS{X}$ is constructed as a sub-complex of the full-dimensional tessellation $\Del{X}$ and a coarse restricted volumetric tessellation $\TV=\DelV{X}$ is also available as a by-product. The constraints satisfied by the Frontal-Delaunay algorithm are identical to those described previously for the restricted Delaunay-refinement scheme, with upper bounds on the radius-edge ratio $\bar{\rho}$, surface discretisation error $\bar{\epsilon}(\mathbf{x}_f)$ and element size $\bar{h}(\mathbf{x}_f)$ all required to be satisfied for convergence. See Algorithm~\ref{algorithm_restricted_delaunay_surface} for a detailed outline of the method.

\subsection{Point-placement Strategy}
\label{section_point_placement}

Given a surface facet $f\in\DelS{X}$ marked for refinement, the new Steiner vertex introduced to eliminate $f$ is an off-centre, constructed to adhere to local size and shape constraints. The off-centres introduced in this study involve the placement of two distinct kinds of Steiner vertices. Type~I vertices, $\mathbf{c}^{(1)}$, are equivalent to conventional element circumcentres, and are used to satisfy constraints on the element radius-edge ratios. Type~II vertices, $\mathbf{c}^{(2)}$, are \textit{size-optimal} points, designed to satisfy element sizing constraints in a locally optimal fashion. Adopting the \textit{generalised} off-centre framework introduced by \"Ung\"or, the `ideal' location of the size-optimal off-centre $\mathbf{c}^{(2)}$ is based on a consideration of the isosceles triangle $\sigma$ formed about the short `frontal' edge $\mathbf{e}_{0}\in f$. The point $\mathbf{c}^{(2)}$ is positioned to ensure that $\sigma$ satisfies local size constraints.

\begin{figure*}[t]
\centering
\caption{Placement of off-centre vertices, showing (i) the surface ball associated with a facet $f\in\DelS{X}$, (ii) the plane $\mathcal{V}$ aligned with the local facet of $\Vor{X}$  associated with the frontal edge $\mathbf{e}_{0}\in f$, (iii) placement of the size-optimal vertex $\mathbf{c}^{(2)}$.}

\label{figure_offcenter_surf}

{
\footnotesize
\tabulinesep=2pt

\bigskip

\begin{tabu} {ccc}

\begin{minipage}[c]{0.275\textwidth}
\centering
\includegraphics[width=3.70cm]{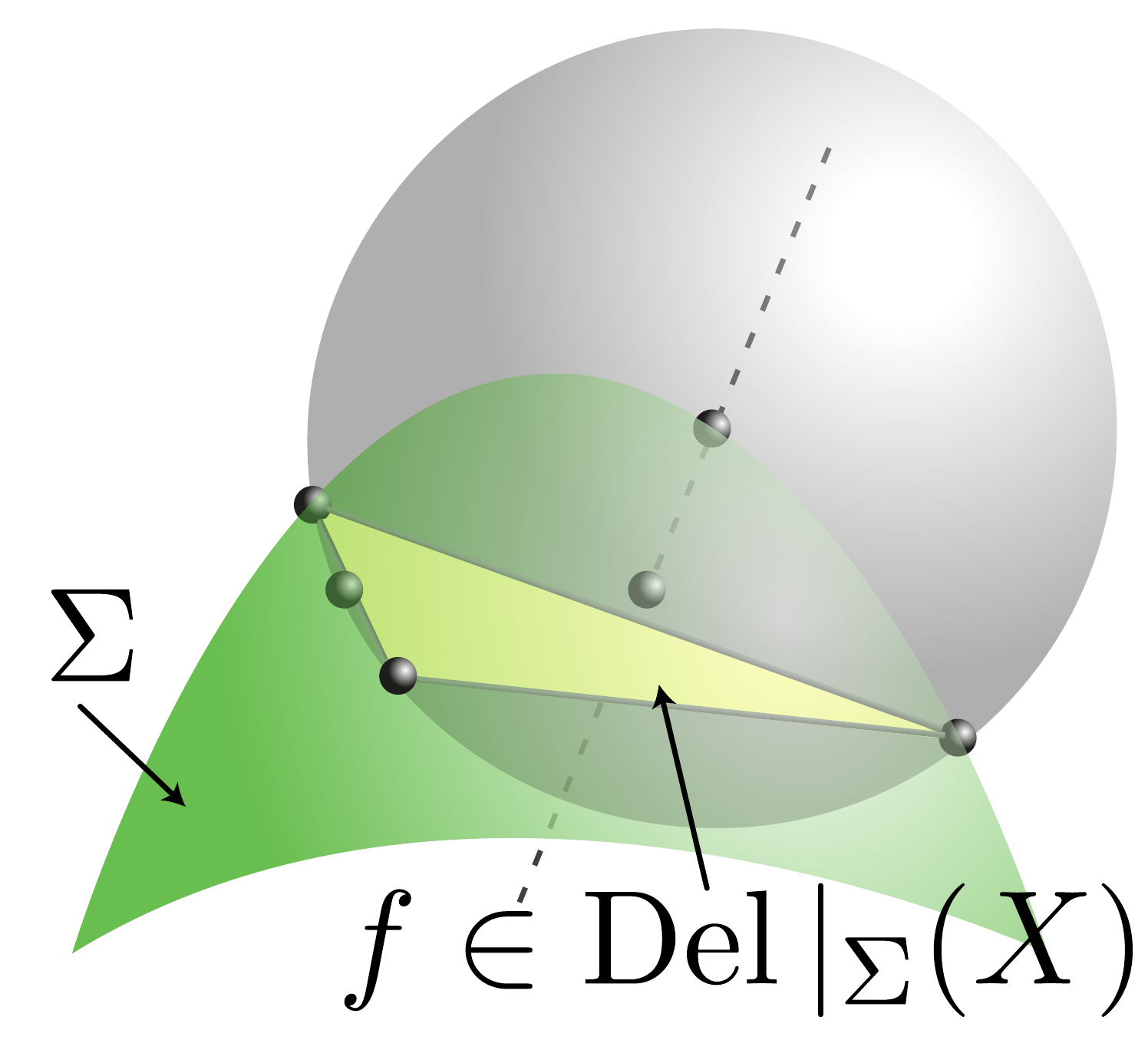} 
\end{minipage} &
\begin{minipage}[c]{0.275\textwidth}
\centering
\includegraphics[width=3.70cm]{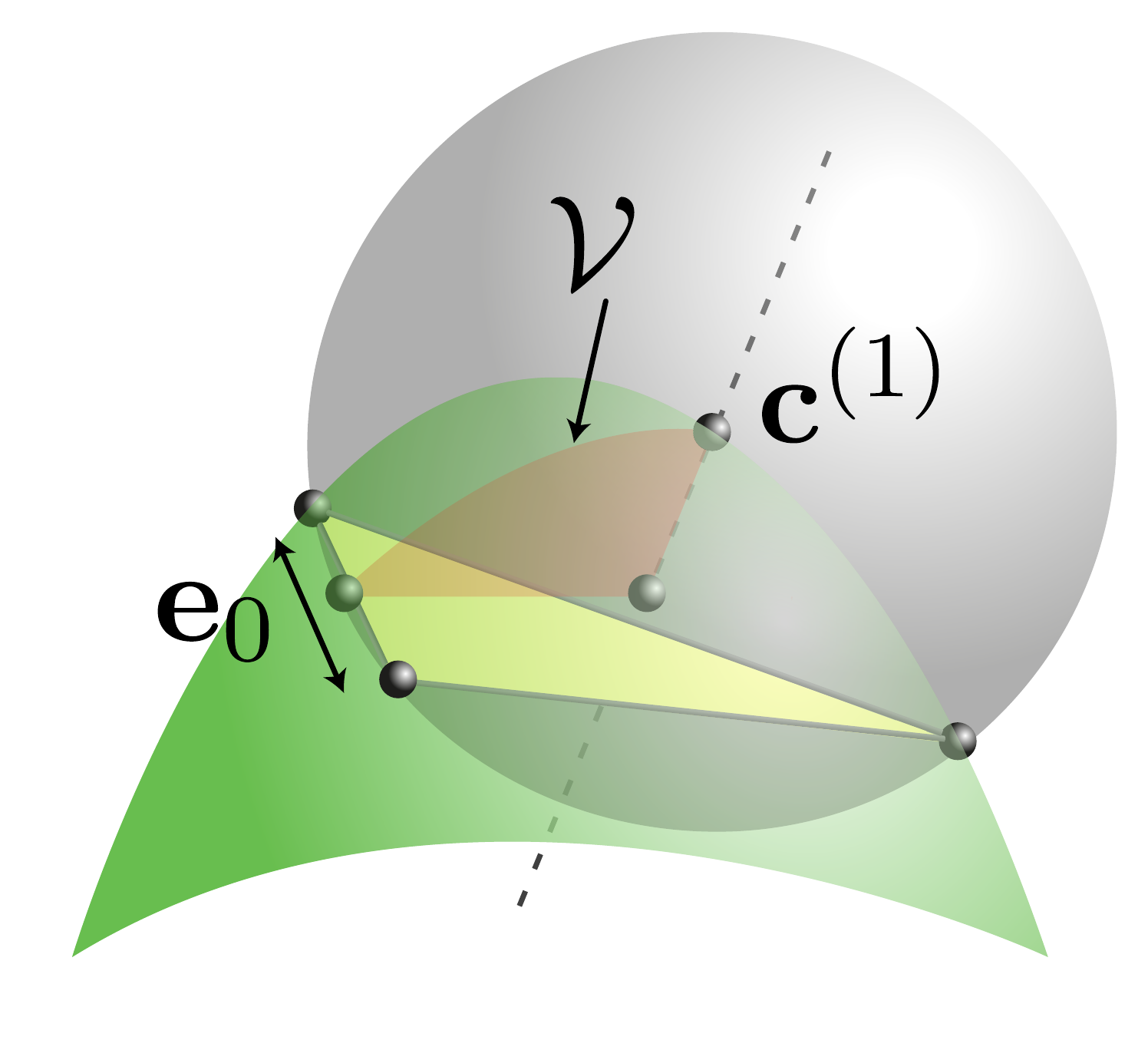} 
\end{minipage} &
\begin{minipage}[c]{0.275\textwidth}
\centering
\includegraphics[width=3.70cm]{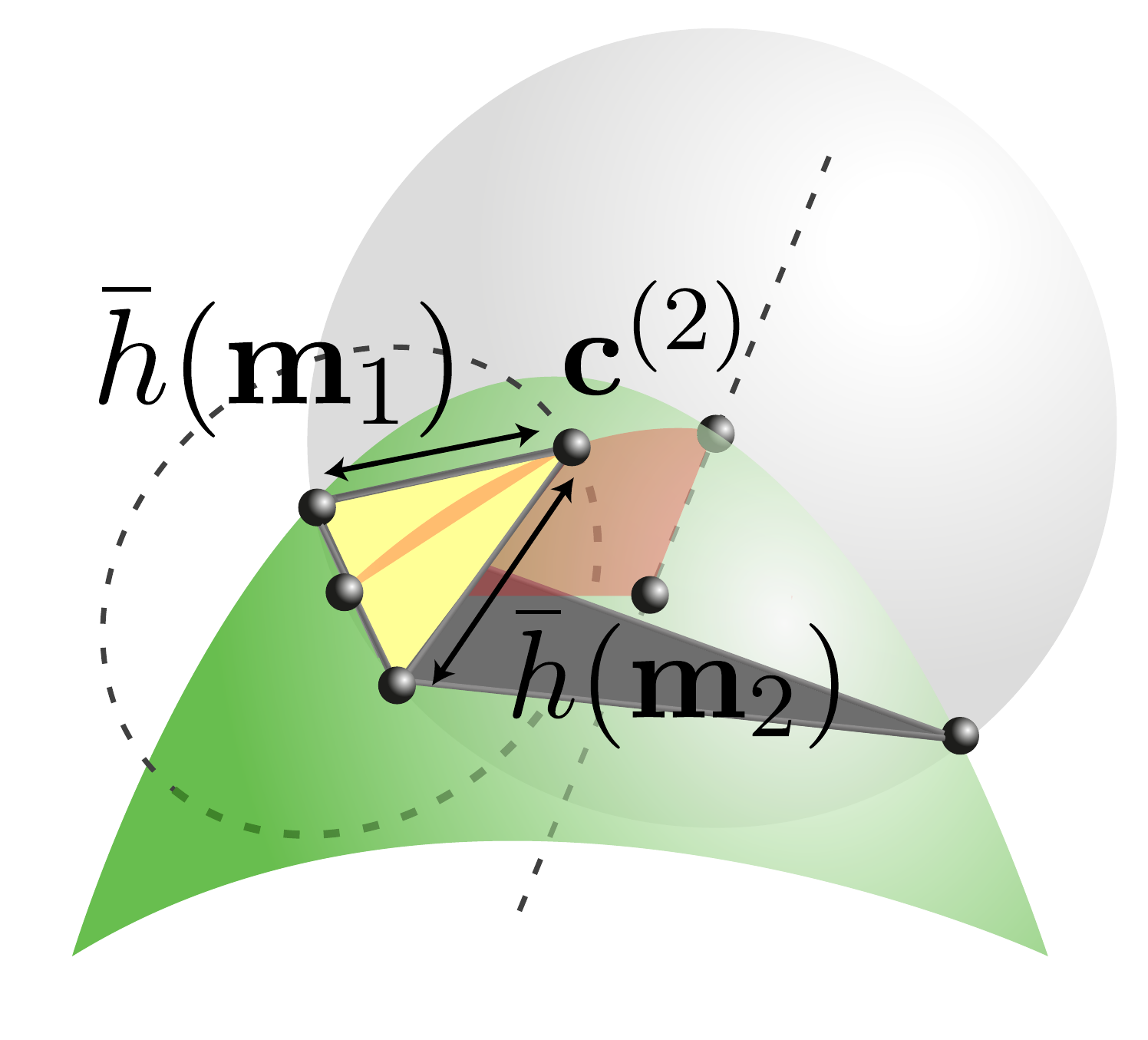} 
\end{minipage} \\

(i) & (ii) & (iii) \\

\end{tabu}}
\end{figure*}

Given a \textit{refinable} 2-simplex $f\in\DelS{X}$, the size-optimal Type~II vertex $\mathbf{c}^{(2)}$ is positioned at an intersection of the surface $\Sigma$ and a plane $\mathcal{V}$, where $\mathcal{V}$ is aligned with the local face of the Voronoi diagram $\Vor{X}$ associated with the frontal edge $\mathbf{e}_{0}\in f$. The plane is positioned such that it passes through three local points on $\Vor{X}$: the midpoint of the frontal edge $\mathbf{e}_{0}\in f$, the centre of the diametric ball of $f$ and the centre of the surface Delaunay ball $\operatorname{B}\left(\mathbf{c},r\right)_{\text{max}}=\operatorname{SDB}\left(f\right)$. The vertex $\mathbf{c}^{(2)}$ is positioned such that it forms an isosceles triangle candidate $\sigma$ about the frontal edge $\mathbf{e}_{0}$, such that its size $h\left(\sigma\right)$ satisfies local constraints. Specifically, the altitude of $\sigma$ is computed from local mesh-size information, such that
\begin{gather}
\label{equation_frontal_length1}
{a_{\sigma}} = \min\left( \left(\bar{h}_{\sigma}^{2} - \|\tfrac{1}{2}\mathbf{e}_{0}\|^2\right)^{_{1\over 2}}, \nicefrac{\sqrt{3}}{2}\,\bar{h}_{\sigma} \right) 
\\[1ex]
\label{equation_frontal_length2}
\bar{h}_{\sigma} = \tfrac{1}{2}\left({\bar{h}(\mathbf{m}_{1})}+{\bar{h}(\mathbf{m}_{2})}\right)
\end{gather}
where the $\mathbf{m}_{i}$'s are the edge midpoints and $\nicefrac{\sqrt{3}}{2}\,\bar{h}_{\sigma}$ is the altitude for an `ideal' element. 

The position of the point $\mathbf{c}^{(2)}$ is calculated by computing the intersection of the surface $\Sigma$ with a circle of radius $a_{\sigma}$, centred at the midpoint of the frontal edge $\mathbf{e}_{0}\in f$ and inscribed on the plane $\mathcal{V}$. In the case of multiple intersections, the candidate point ${\mathbf{c}_{i}}^{(2)}$ of closest alignment to the \textit{frontal} direction vector $\mathbf{d}_{f}$ is selected. Specifically, the point ${\mathbf{c}_{i}}^{(2)}$ that maximises the scalar product $\big({\mathbf{c}_{i}}^{(2)}-\mathbf{m}_{0}\big) \cdot \mathbf{d}_{f}$ is chosen, where $\mathbf{m}_{0}$ is the midpoint of the short edge $\mathbf{e}_{0}\in f$ and the frontal direction vector $\mathbf{d}_{f}$ is taken from the midpoint $\mathbf{m}_{0}$ to the centre of the surface ball $\operatorname{B}\left(\mathbf{c},r\right)_{\text{max}}=\operatorname{SDB}\left(f\right)$.

For non-uniform $\bar{h}\left(\mathbf{x}\right)$, expressions for the position of the point $\mathbf{c}^{(2)}$ are weakly non-linear, with the altitude $a_{\sigma}$ depending on an evaluation of the mesh size function at the edge midpoints $\bar{h}\left(\mathbf{m}_{i}\right)$ and visa-versa. In practice, since $\bar{h}\left(\mathbf{x}\right)$ is guaranteed to be Lipschitz smooth, a simple iterative predictor-corrector procedure is sufficient to solve these expressions approximately. The positioning of size-optimal Type~II Steiner vertices is illustrated in Figure~\ref{figure_offcenter_surf}.

Using the size-optimal Type~II point $\mathbf{c}^{(2)}$ and the Type~I point $\mathbf{c}^{(1)}$, the final position of the refinement point $\mathbf{c}$ for the facet $f$ is calculated. The point $\mathbf{c}$ is selected to satisfy the \textit{limiting} local constraints, setting
\begin{gather}
\label{equation_frontal_selection}
\mathbf{c}=\left\{
\begin{array}{ll}
\mathbf{c}^{(2)}, & 
\text{if }\big(d^{(2)} \leq d^{(1)}\big) \text{ and } \big(d^{(2)} \geq \tfrac{1}{2}\|\mathbf{e}_{0}\|\big),
\\[1ex]
\mathbf{c}^{(1)}, &
\text{otherwise}
\end{array}
\right.
\end{gather}
where $d^{(1)}=\|\mathbf{c}^{(1)}-\mathbf{m}_{0}\|$ and $d^{(2)}=\|\mathbf{c}^{(2)}-\mathbf{m}_{0}\|$ are distances from the midpoint of the frontal edge $\mathbf{e}_{0}$ to the Type~I and Type~II vertices, respectively. 

The selection criteria is designed to ensure that the refinement scheme smoothly degenerates to that of a conventional circumcentre-based Delaunay-refinement strategy in limiting cases, while using a locally shape-optimal approach where possible.

Specifically, the condition $d^{(2)} \leq d^{(1)}$ guarantees that $\mathbf{c}$ lies no further from the frontal edge $\mathbf{e}_{0}$ than the centre of the surface Delaunay ball $\operatorname{SDB}\left(f\right)$. Such behaviour ensures that a conventional circumcentre-based scheme is selected when the element size becomes sufficiently small. Additionally, the condition $d^{(2)} \geq \tfrac{1}{2}\|\mathbf{e}_{0}\|$ ensures that the size-optimal scheme is only chosen when the edge $\mathbf{e}_{0}$ is sufficiently small compared to the local mesh size function. This behaviour ensures that $\mathbf{e}_{0}$ is a good `frontal' edge candidate.

\subsection{Refinement Order}
 
In addition to the use of `off-centre' point-placement schemes, the Frontal-Delaunay algorithm also introduces changes to the order in which elements are refined. To better mimic the behaviour of an advancing-front type method, elements are refined only if they are adjacent to an existing `frontal' entity. In the case of surface facets $f\in\DelS{X}$, the short frontal edge $\mathbf{e}_{0}\in f$ must be shared by at least one adjacent facet $f_{j}\in\DelS{X}$ that has `converged' -- satisfying its associated radius-edge, mesh-size and surface-error constraints. The idea of defining `frontal' entities as a dynamic boundary between converged and un-converged elements in the mesh is a common feature of Frontal-Delaunay algorithms, with similar approaches used in \cite{Rebay93FrontalDelaunay,Mavriplis95FrontalDelaunay,
Frey98FrontalDelaunay,Remacle13Quads}.

\subsection{Theoretical Guarantees}
\label{section_discussion}

Boissonnat and Oudot \cite{boissonnat03ProvablyGoodSurface,boissonnat05ProvablyGoodMeshing} have previously shown that the conventional restricted Delaunay-refinement algorithim presented in Section~\ref{section_delaunay_refinement} is guaranteed to: (i) terminate in a finite number of steps, (ii) produce elements of bounded radius-edge ratios $\rho(f)\leq\bar{\rho}$, (iii) satisfy non-uniform user-defined mesh-sizing and surface error constraints $h(f)\leq \bar{h}(\mathbf{x}_{f})$, $\epsilon(f)\leq \bar{\epsilon}(\mathbf{x}_{f})$, and (iv) provide a \textit{good} geometrical and topological approximation to the underlying surface $\Sigma$ when the mesh size function $\bar{h}(\mathbf{x})$ is sufficiently small.

The behaviour of the new Frontal-Delaunay algorithm can be assessed using a similar framework.

\subsection{Termination \& Convergence}

Firstly, the termination and convergence of the algorithm is analysed:

\begin{proposition}
\label{proposition_surface_termination}
Given a smooth surface $\Sigma$, a radius-edge threshold $\bar{\rho}\geq 1$ and positive mesh-size and surface-error functions $\bar{h}\left(\mathbf{x}\right)\geq\bar{h}_{0}$, $\bar{\epsilon}\left(\mathbf{x}\right)\geq\bar{\epsilon}_{0}$, $\bar{h}_{0},\bar{\epsilon}_{0}\in\mathbb{R}^{+}$ the Frontal-Delaunay algorithm is guaranteed to terminate in a finite number of steps.
\end{proposition}

Before seeking to prove Proposition~\ref{proposition_surface_termination} itself, a number of intermediate results are first established.

\begin{lemma}
\label{lemma_packing}
Let $D\subset\mathbb{R}^d$ be a bounded domain. Given a subset $X\subset D$, satisfying $\|\mathbf{u}-\mathbf{v}\|\geq\gamma$ for all pairs $\mathbf{u},\mathbf{v}\in X$ and some scalar $\gamma\in\mathbb{R}^{+}$, the size of $X$ is bounded, such that $|X| \leq\mu$ for some constant $\mu\in\mathbb{Z}^{+}$.
\end{lemma}

This so-called \textit{packing-lemma} states that any point-set $X$ in a bounded domain $\mathbb{R}^{d}$ for which there exists a positive minimum separation distance between the points in $X$ must, consequently, be finite. As such, if it can be shown that a meshing algorithm preserves such a minimum separation length between its vertices, Lemma~\ref{lemma_packing} can be invoked to prove that such a vertex-set is finite, and, as a result, that termination of the algorithm is guaranteed. Lemma~\ref{lemma_packing} is stated here without proof -- interested readers are referred to \cite{ChengDeyShewchuk} for full details.

Before examining the behaviour of the Frontal-Delaunay algorithm as a whole, the impact of a single refinement operation is analysed. Firstly, the behaviour of the various off-centre point-placement schemes introduced in Section~\ref{section_point_placement} is examined:

\begin{proposition}
\label{proposition_surface_shape_quality}
Let the threshold on radius-edge ratios be $\bar{\rho}\geq 1$ and $\DelS{X_{k}}$ be the restricted surface triangulation induced after $k$ refinement operations. Given a low-quality surface facet $f\in\DelS{X_{k}}$, for which $\rho\left(f\right)\geq\bar{\rho}$, the minimum edge length $\|\mathbf{e}_{0}\|\in\DelS{X_{k}}$ is not decreased following the insertion of a new Type~I Steiner vertex at the centre of $\operatorname{SDB}(f)$.
\end{proposition}

\begin{proof}
Recalling that the surface facet $f$ is locally Delaunay, the insertion of a new Type~I vertex $\mathbf{c}_{k}$ at the centre of the associated surface Delaunay ball $\operatorname{B}(\mathbf{c}_{k},r_{k})=\operatorname{SDB}(f)$ is guaranteed to create edges no shorter than the radius of the ball of $f$:
\begin{gather}
\|\mathbf{e}_{0}\|_{k+1} \geq r_{k}.
\end{gather}
Here $\|\mathbf{e}_{0}\|_{k+1}$ is the length of the shortest edge created by the insertion of the new vertex $\mathbf{c}_{k}$. Dividing by the current minimum edge length, and using the definition of the radius-edge ratio:
\begin{gather}
\label{equation_typeI_refinement}
\frac{\|\mathbf{e}_{0}\|_{k+1}}{\|\mathbf{e}_{0}\|_{k}\hfill} \geq \frac{r_{k}}{\|\mathbf{e}_{0}\|_{k}} \geq \bar{\rho}.
\end{gather}
Clearly, when $\bar{\rho}\geq 1$ the existing minimum length $\|\mathbf{e}_{0}\|_{k}$ is preserved by the insertion of the new vertex $\mathbf{c}_{k}$.
\end{proof}

\begin{proposition}
\label{proposition_offcentre_deferral}
Let $\bar{h}(\mathbf{x})\geq\bar{h}_{0},\, \bar{h}_{0}\in\mathbb{R}^{+}$ be a positive mesh-size function and $\DelS{X_{k}}$ be the restricted surface triangulation induced after $k$ refinement operations. Given a refinable surface facet $f\in\DelS{X_{k}}$ use of the Type~II point-placement scheme is `declined' when $f$ is sufficiently small. Specifically, selection of the Type~II point-placement scheme requires that $r_{k} \geq H$, where $r_{k}$ is the radius of the surface ball associated with the facet $f$ and $H$ is the length of the new frontal edge candidates defined via expressions~(\ref{equation_frontal_length1})--(\ref{equation_frontal_length2}).
\end{proposition}

\begin{proof}
Given a surface facet $f$, the off-centre selection criteria expressed in (\ref{equation_frontal_selection}) requires that the Type~II off-centre $\mathbf{c}^{(2)}$ lies along a `safe' region of the Voronoi diagram -- bounded by the centre of the ball $\operatorname{B}(\mathbf{c}_{k},r_{k})=\operatorname{SDB}(f)$ such that:
\begin{align}
\label{equation_frontal_selection2}
\|\mathbf{c}^{(2)}-\mathbf{m}_{0}\| & \leq \|\mathbf{c}_{k}-\mathbf{m}_{0}\|
\end{align}
where $\mathbf{m}_{0}$ is the midpoint of the base edge $\|\mathbf{e}_{0}\|$. Expressing (\ref{equation_frontal_selection2}) in terms of the target edge lengths and surface ball radii:
\begin{gather}
\sqrt{H^{2}-\|\tfrac{1}{2}\mathbf{e}_{0}\|^{2}} \leq \sqrt{r_{k}^{2}-\|\tfrac{1}{2}\mathbf{e}_{0}\|^{2}}
\end{gather}
where $H$ is the desired edge-length, computed from local mesh-size constraints (\ref{equation_frontal_length1})--(\ref{equation_frontal_length2}). This expressions can be further simplified, leading to a simple inequality for the selection of Type~II vertices:
\begin{gather}
H \leq r_{k}
\end{gather}
It can be seen that the Type~II point-placement scheme is \textit{declined} when the surface ball associated with a given facet $f$ is sufficiently small compared to the local target edge-length $H$.
\end{proof}

Using Propositions~\ref{proposition_surface_shape_quality} and \ref{proposition_offcentre_deferral}, the off-centre point-place\-ment strategies utilised in the Frontal-Delaunay algorithm are shown to constitute a `safe' set of refinement operations, with the insertion of Type~I vertices leading to a non-decreasing minimum edge-length and use of the Type~II scheme declined once the mesh is sufficiently well refined. Using these results, termination of the full Frontal-Del\-aunay algorithm can be re-visited:

\newproof{pop}{Proof of Proposition \ref{proposition_surface_termination}}

\begin{pop}
The Frontal-Delaunay algorithm refines any surface facet $f\in\DelS{X}$ if: (i) it is of poor shape quality, such that $\rho\left(f\right)\geq\bar{\rho}$, (ii) it is too large, such that $h\left(f\right)\geq \alpha \bar{h}(\mathbf{x}_{f})$, or (iii) it violates the local surface discretisation error threshold, such that $\epsilon\left(f\right)\geq\bar{\epsilon}(\mathbf{x}_{f})$.

\begin{enumerate}[(i)]\itemsep +6pt
\item Using Proposition~\ref{proposition_offcentre_deferral}, it is known that the introduction of Type~II off-centres is declined once the mesh is sufficiently well refined. Type~I points alone can therefore be relied upon to satisfy element shape-constraints.

\medskip

Given a radius-edge threshold $\bar{\rho}\geq 1$, Proposition~\ref{proposition_surface_shape_quality} states that refinement does not lead to a reduction in minimum edge length. Shape-based refinement therefore preserves the existing minimal edge length $\|\mathbf{e}_{0}\|_{k}$, associated with either the initial sampling $X_{0}\in\Sigma$, or the insertion of some previous vertex $\mathbf{x}_{k}\in X$ due to size- or surface-error-driven refinement.

\item Noting that $\bar{h}(\mathbf{x})$ is positive, with $\bar{h}\left(\mathbf{x}\right)\geq \bar{h}_{0}$ for some $\bar{h}_{0}\in\mathbb{R}^{+}$, size-driven refinement is clearly \textit{declined} once the local element size is sufficiently small. Recalling that $h\left(f\right) = \sqrt{3}\,r$, where $r$ is the radius of the surface ball associated with a given facet $f$, the mesh-sizing constraints $h\left(f\right) \leq \alpha\bar{h}(\mathbf{x}_{f})$ are satisfied once all $r \leq \left(\nicefrac{\alpha}{\sqrt{3}}\right) \bar{h}(\mathbf{x}_{f})$.

\item Noting that $\bar{\epsilon}(\mathbf{x})$ is positive, with $\bar{\epsilon}(\mathbf{x})\geq\bar{\epsilon}_{0}$ for some $\epsilon_{0}\in\mathbb{R}^{+}$, it is clear that surface-error driven refinement is \textit{declined} once the local element size is sufficiently small. Recalling that $\epsilon(f)$ is the orthogonal distance between a given facet $f$ and the centre of its associated surface ball $\operatorname{B}(\mathbf{c},r)=\operatorname{SDB}(f)$, it is clear that $\epsilon(f)\leq r$, as $\operatorname{SDB}(f)$ must circumscribe the facet. In the worst-case, surface-error driven refinement is declined once all $r \leq \epsilon(\mathbf{x}_{f})$.
\end{enumerate}

Using (i)--(iii), the progression of the Frontal-Delaunay algorithm can be examined. Starting with an initial vertex distribution $X_{0}\in\Sigma$, both Type~I and Type~II point-placement schemes are used to satisfy mesh-size and surface-error constraints -- with refinement continuing until all surface balls are sufficiently small, as per (ii)--(iii). Any remaining low-quality elements in violation of the radius-edge constraints are subsequently refined through the insertion of Type~I vertices alone. This final phase is guaranteed to preserve minimum edge-length.

Using Lemma~\ref{lemma_packing}, it is clear that the Frontal-Delaunay algorithm is guaranteed to produce a point-wise sampling $X\in\Sigma$ of finite size and must terminate in a finite number of steps as a consequence.
\end{pop}

Finite termination leads directly to a number of useful auxiliary guarantees on both the nature and quality of the output tessellation. Adopting the conventional terminology, surface meshes generated using the Frontal-Delaunay algorithm can be considered to be `provably-good':

\begin{corollary}
\label{corollary_surface_convergence}
Given a smooth surface $\Sigma$, a radius-edge threshold $\bar{\rho}\geq 1$ and positive mesh-size and surface-error functions $\bar{h}\left(\mathbf{x}\right)>0$, $\bar{\epsilon}\left(\mathbf{x}\right)>0$, all surface facets $f\in\TS$ in the restricted surface tessellation $\TS=\DelS{X}$ generated by the Frontal-Delaunay algorithm satisfy: (i) $\rho\left(f\right)<\bar{\rho}$, (ii) $h\left(f\right)<\alpha \bar{h}(\mathbf{x}_{f})$, and (iii) $\epsilon\left(f\right)<\bar{\epsilon}(\mathbf{x}_{f})$.
\end{corollary}

\begin{proof}
The Frontal-Delaunay maintains a queue of `bad' elements, containing any surface facets $f\in\DelS{X}$ that are in violation of one or more local constraints. Specifically, a given 2-face $f\in\DelS{X}$ is enqueued if: (i) $\rho\left(f\right)\geq\bar{\rho}$, (ii) $h\left(f\right)\geq \alpha \bar{h}(\mathbf{x}_{f})$, or (iii) $\epsilon\left(f\right)\geq \bar{\epsilon}(\mathbf{x}_{f})$. Given that termination is assured, it is clear that the refinement queue must become empty after a finite number of steps and that all facets in the resulting triangulation $\DelS{X}$ satisfy the requisite radius-edge ratio, element size and surface error constraints as a consequence.
\end{proof}

In addition to proofs of termination and convergence, it is also necessary to ensure the point-set $X\in\Sigma$ remains \textit{well-distributed} throughout the refinement process. Such behaviour requires that no spurious short edges be introduced within the initial refinement phase as the algorithm seeks to satisfy mesh-size and surface-error constraints.

\begin{proposition}
\label{proposition_surface_refinement}
Let $\DelS{X_{k}}$ be the current restricted surface triangulation. Given a large \textit{refinable} surface facet $f\in\DelS{X_{k}}$, the minimum edge length $\|\mathbf{e}_{0}\|_{k}$ is decreased by an $\BigO{1}$ factor of $\nicefrac{1}{\sqrt{3}}$ in the worst-case following the insertion of a new Type~I or Type~II Steiner vertex.
\end{proposition}

\begin{proof} 
Considering the insertion of new Type~I and Type~II vertices separately:

\begin{enumerate}[(i)]\itemsep +6pt
\item Recalling that the surface facet $f$ is locally Delaunay, the insertion of a new Type~I vertex $\mathbf{c}_{k}$ at the centre of the associated surface Delaunay ball $\operatorname{B}(\mathbf{c}_{k},r_{k})=\operatorname{SDB}(f)$ is guaranteed to create edges no shorter than the radius of the ball of $f$:
\begin{gather}
\|\mathbf{e}_{0}\|_{k+1} \geq r_{k}.
\end{gather}
Here $\|\mathbf{e}_{0}\|_{k+1}$ is the length of the shortest edge created by the insertion of the new vertex $\mathbf{c}_{k}$. Dividing by the current minimum edge length, and using the definition of the radius-edge ratio:
\begin{gather}
\label{equation_typeI_refinement}
\frac{\|\mathbf{e}_{0}\|_{k+1}}{\|\mathbf{e}_{0}\|_{k}\hfill} \geq \frac{r_{k}}{\|\mathbf{e}_{0}\|_{k}} = \rho(f).
\end{gather}
Clearly, the minimum edge length is decreased when any element with $\rho(f)<1$ is refined. Noting that $\rho(f)$ achieves a minimum of $\nicefrac{1}{\sqrt{3}}$ when $f$ is equilateral, it can be seen that insertion of a Type~I Steiner vertex leads to a reduction in the minimum edge length by a factor of $\nicefrac{1}{\sqrt{3}}$ in the worst case. 

\item A Type~II vertex $\mathbf{c}_{k}$ positioned on a facet $\mathcal{V}\subset\Vor{X}$ associated with the short edge $e_{0}$, is equidistant to the vertices $\mathbf{x}_{i},\mathbf{x}_{j}\in e_{0}$. Since $\mathcal{V}$ is a local facet of the Voronoi complex, no other vertex $\mathbf{x}_{m}\in X$ lies closer to the point $\mathbf{c}_{k}$. 

\medskip

Recalling that the Type~II refinement strategy requires the diametric ball of the edge $e_{0}$ remain empty, a worst case reduction in $\|\mathbf{e}_{0}\|$ occurs when $\mathbf{c}_{k}$ is positioned at the intersection of $\mathcal{V}$ and the diametric ball of $e_{0}$, forming a right triangle. Noting that $\mathcal{V}$ intersects $e_{0}$ at its midpoint, the minimum edge length is decreased by a factor of $\nicefrac{1}{\sqrt{2}}$ in such cases.
\end{enumerate}

Overall, Type~I point-placement is a limiting case, reducing minimum edge length by an $\BigO{1}$ factor of $\nicefrac{1}{\sqrt{3}}$ in the worst case.
\end{proof}

\subsection{Quality of Approximation}

While the preceding results address the behaviour and performance of the underlying Delaunay-refinement algorithm, the approximation quality of the resulting mesh is also a key consideration. Using the framework introduced by Boissonnat and Oudot \cite{boissonnat03ProvablyGoodSurface,boissonnat05ProvablyGoodMeshing}, a variety of geometrical and topological guarantees are sought. Firstly, a number of results from \cite{boissonnat03ProvablyGoodSurface,boissonnat05ProvablyGoodMeshing} are summarised:

\begin{definition}
\label{definition_medial_axis}
Given a bounded domain $\Omega$ enclosed by a surface $\Sigma$, the \textit{medial-axis} of $\Omega$ is the topological closure of the set of centres of \textit{empty} balls that touch the surface $\Sigma$ at more than one point.
\end{definition}

Noting that the medial-axis lies at the `centre' of any geometrical features in the domain induced by $\Sigma$, it is clear that the distance to the medial axis, denoted $d_{M}\left(\mathbf{x}\right)$ throughout, is a measure of the local `thickness' of the domain. Using such considerations, Boissonnat and Oudot introduce the notion of \textit{loose $\epsilon$-samples} to describe point-wise samplings $X\in\Sigma$ that induce restricted surface triangulations $\DelS{X}$ exhibiting favourable geometrical and topological guarantees.

\begin{definition}
\label{definition_epsilon_sample}
Let $X\in\Sigma$ be a point-wise sampling of a smooth surface $\Sigma$. Let $\mathcal{E}_{\text{V}}\left(X\right)\subseteq\Vor{X}$ be the set of edges in the associated Voronoi complex. The sampling $X$ is called a is called a \textit{loose $\epsilon$-sample} of $\Sigma$ if: (i) for all intersections $\mathbf{x}\in \left\{\mathcal{E}_{\text{V}}\left(X\right)\cap\Sigma\right\}$ the associated \textit{$\epsilon$-balls}, $\operatorname{B}(\mathbf{c},r)$, are non-empty, such that $X\cap\operatorname{B}\left(\mathbf{x},\epsilon d_{M}\left(\mathbf{x}\right)\right) \neq \emptyset$ for some $\epsilon\in\mathbb{R}^{+}$, and (ii) the restricted triangulation $\DelS{X}$ includes at least one vertex on all connected components of the surface $\Sigma$.
\end{definition}

Boissonnat and Oudot additionally establish the following geometrical and topological guarantees for restricted Delaunay tessellations $\DelS{X}$ induced by such loose $\epsilon$-samples:

\begin{proposition}
\label{prop_loose_gamma_sample}
Let $X\in\Sigma$ be a loose $\epsilon$-sampling of a smooth surface $\Sigma$, with $\epsilon\leq 0.08$. The associated restricted surface triangulation $\DelS{X}$ exhibits the following properties: (i) the triangulation $\DelS{X}$ is homeomorphic to the surface $\Sigma$, (ii) for any 2-face $f\in\DelS{X}$ the angle between the facet normal $\mathbf{\hat{n}}_{f}$ and the surface normal $\mathbf{\hat{n}}_{\Sigma}$, sampled at the vertices $\mathbf{x}\in f$, is $\BigO{\epsilon}$, and (iii) the one-sided Hausdorff distance $\operatorname{H}\left(\DelS{X},\Sigma\right)$, sampled at the centre of the surface Delaunay balls $\operatorname{B}\left(\mathbf{c},r\right)_{\text{max}}=\operatorname{SDB}\left(f\right)$ is $\BigO{\epsilon^{2}}$.
\end{proposition}

The reader is referred to \cite{boissonnat03ProvablyGoodSurface,boissonnat05ProvablyGoodMeshing} or \cite{ChengDeyShewchuk} for additional details and proofs. Boissonnat and Oudot go on to show that restricted Delaunay refinement, consistent with the algorithm presented in Section~\ref{section_delaunay_refinement}, is a `provably-good' technique, being guaranteed to produce a loose $\epsilon$-sample of a surface $\Sigma$ in a finite number of steps, and to generate a high-quality surface tessellation $\DelS{X}$ as a result. 

The Frontal-Delaunay algorithm presented in this study derives from the same Delaunay-refinement framework developed in \cite{boissonnat03ProvablyGoodSurface,boissonnat05ProvablyGoodMeshing}, and inherits much of its theoretical robustness as a result. Using the properties of loose $\epsilon$-samples, the surface triangulations generated using the Frontal-Delaunay algorithm are guaranteed to be a good geometrical and topological approximation to the underlying surface -- provided the mesh size function is sufficiently small:

\begin{proposition}
\label{proposition_surface_correctness}
Given a smooth surface $\Sigma$ and an $\epsilon$-con\-forming size function $\left(\nicefrac{\alpha}{\sqrt{3}}\right)\bar{h}\left(\mathbf{x}\right) \leq \epsilon d_{M}\left(\mathbf{x}\right)$, the point-wise sampling $X\in\Sigma$ generated by the Frontal-Delaunay algorithm is a loose $\epsilon$-sample. When $\epsilon\leq 0.08$ the associated triangulation $\DelS{X}$ is a tight topological and geometrical approximation of the surface $\Sigma$. 
\end{proposition}

\begin{proof}
The Frontal-Delaunay algorithm refines any large surface facets, ensuring that $h\left(f\right)\leq \alpha \bar{h}(\mathbf{x}_{f})$ for all $f\in\DelS{X}$. Rearranging the previous expression, and substituting for the given mesh size constraints leads to $h(f) = \sqrt{3}r \leq \alpha \left(\nicefrac{\sqrt{3}}{\alpha}\right) \epsilon d_{M}\left(\mathbf{x}\right)$, which simplifies to $r\leq 0.08\, d_{M}\left(\mathbf{x}\right)$ when $\epsilon\leq 0.08$. Recalling Definition~\ref{definition_epsilon_sample}, it is clear that $X$ is a loose $\epsilon$-sample of $\Sigma$, and, via Proposition~\ref{prop_loose_gamma_sample}, $\DelS{X}$ is a tight geometrical and topological approximation of $\Sigma$ as a result.
\end{proof}

\section{Mesh Size Functions}
\label{section_size_functions}

The construction of high-quality mesh-size functions is an important aspect of both the restricted Delaunay-refinement and Frontal-Delaunay algorithms presented in Sections~\ref{section_delaunay_refinement} and \ref{section_frontal_delaunay}. A good mesh size function $\bar{h}\left(\mathbf{x}\right)$ incorporates sizing constraints imposed by both the user and the geometry of the domain to be meshed. These contributions can be considered via two separate size functions, where $\bar{h}_{u}\left(\mathbf{x}\right)$ represents user-defined sizing information and $\bar{h}_{g}\left(\mathbf{x}\right)$ encapsulates geometry driven constraints. In this study, it is required that both $\bar{h}_{u}\left(\mathbf{x}\right)$ and $\bar{h}_{g}\left(\mathbf{x}\right)$ be piecewise linear functions defined on a supporting tetrahedral complex $\mathcal{S}$. Construction of appropriate user-defined functions $\bar{h}_{u}\left(\mathbf{x}\right)$ is highly problem dependent and detailed formulations are not considered here.

As per the discussions presented in Section~\ref{section_discussion}, it is known that the geometric mesh size function $\bar{h}_{g}\left(\mathbf{x}\right)$ must be sufficiently small relative to the \textit{local-feature-size} of the domain, denoted $\operatorname{lfs}\left(\mathbf{x}\right)$, to ensure that the resulting surface triangulation $\DelS{X}$ is both topologically and geometrically correct. Specifically, to guarantee that the resulting tessellation is a \textit{loose $\epsilon$-sample}, it is required that $h_{g}\left(\mathbf{x}\right)\leq 0.08\operatorname{lfs}\left(\mathbf{x}\right)$, where $\operatorname{lfs}\left(\mathbf{x}\right)$ is the distance to the \textit{medial-axis} of the bounded domain $\Omega$. In practice, such constraints are typically found to be overly restrictive \cite{ChengDeyShewchuk}, allowing for a relaxation of the coefficient $\epsilon$. In this study, $\epsilon = \nicefrac{1}{2}$ is used throughout. The local-feature-size can be expressed as the distance from any point $\mathbf{x}_{i}\in\Sigma$ to the closest point on the \textit{medial-axis} of the domain. Given that $\Sigma$ is represented as an underlying triangulation $\mathcal{P}$ in this work, a discrete approximation to $\operatorname{lfs}\left(\mathbf{x}\right)$ can be calculated directly via the method of Amenta and Bern \cite{amenta1999surface,amenta2001power,dey2004approximating}. For the sake of brevity, we do not describe these methods in full here, but simply note that they provide an efficient means to estimate $\operatorname{lfs}\left(\mathbf{x}\right)$ at the vertices of the supporting complex $\mathcal{S}$.

Given an estimate of $\operatorname{lfs}\left(\mathbf{x}\right)$ on the boundary of $\mathcal{S}$, it is possible to exact a degree of user-defined control on the resulting size function via a Lipschitz smoothing process. Following the work of Persson \cite{Persson06SizeFunc}, a \textit{gradient-limited} function $\tilde{h}\left(\mathbf{x}\right)$ can be constructed by limiting the variation in $\bar{h}\left(\mathbf{x}\right)$ over the elements of $\mathcal{S}$. In this study, a scalar smoothing parameter $g\in\mathbb{R}^{+}$ is used to globally, and isotropically limit variation, such that 
\begin{gather}
\tilde{h}(\mathbf{x}_{i}) \leq \tilde{h}(\mathbf{x}_{j}) + g\,\|\mathbf{x}_{i}-\mathbf{x}_{j}\| 
\end{gather}
for all vertex pairs $\mathbf{x}_{i},\mathbf{x}_{j}\in\mathcal{S}$. The gradient-limited size function $\tilde{h}\left(\mathbf{x}\right)$ becomes more uniform as $g\rightarrow 0$.

\section{Results \& Discussions}
\label{section_results}

\begin{figure*}[p]
\centering
\caption{Triangulations for the \textsc{elephant}, \textsc{hip} and \textsc{bunny} problems, showing output for the \textsc{jgsw--fd} algorithm (upper) and the \textsc{cgal--dr} algorithm (lower). Meshes were built using: (i) uniform mesh size functions $\bar{h}\left(\mathbf{x}\right)=\alpha$, with $\alpha\in\mathbb{R}^{+}$, (ii) tight constraints on element shape-quality, such that $\theta_{\text{min}}\geq 30^\circ$, and (iii) uniform surface discretisation thresholds $\bar{\epsilon} = \left(\nicefrac{1}{4}\right)\, \bar{h}\left(\mathbf{x}\right)$. Element counts $|\TS|$, and algorithm run-times $(\mathrm{t})$ are included for each case. Normalised histograms of element area-length ratio $a\left(f\right)$, plane-angle $\theta\left(f\right)$ and relative-length $\reledge$ are illustrated.}

\label{figure_surf_cgal_uniform}

{
\footnotesize
\tabulinesep=1pt

\medskip

\begin{tabu} {c|c|c}

\parbox[b][1em][b]{.300\textwidth}{\center (\textsc{jgsw--fd}): $|\TS|=16,202$, $\mathrm{t}=0.98\mathrm{s}$} &
\parbox[b][1em][b]{.300\textwidth}{\center (\textsc{jgsw--fd}): $|\TS|=23,291$, $\mathrm{t}=1.42\mathrm{s}$} &
\parbox[b][1em][b]{.300\textwidth}{\center (\textsc{jgsw--fd}): $|\TS|=19,120$, $\mathrm{t}=1.22\mathrm{s}$} \\

\begin{minipage}[c]{.300\textwidth}
\begin{center}
\includegraphics[height=5.50cm]{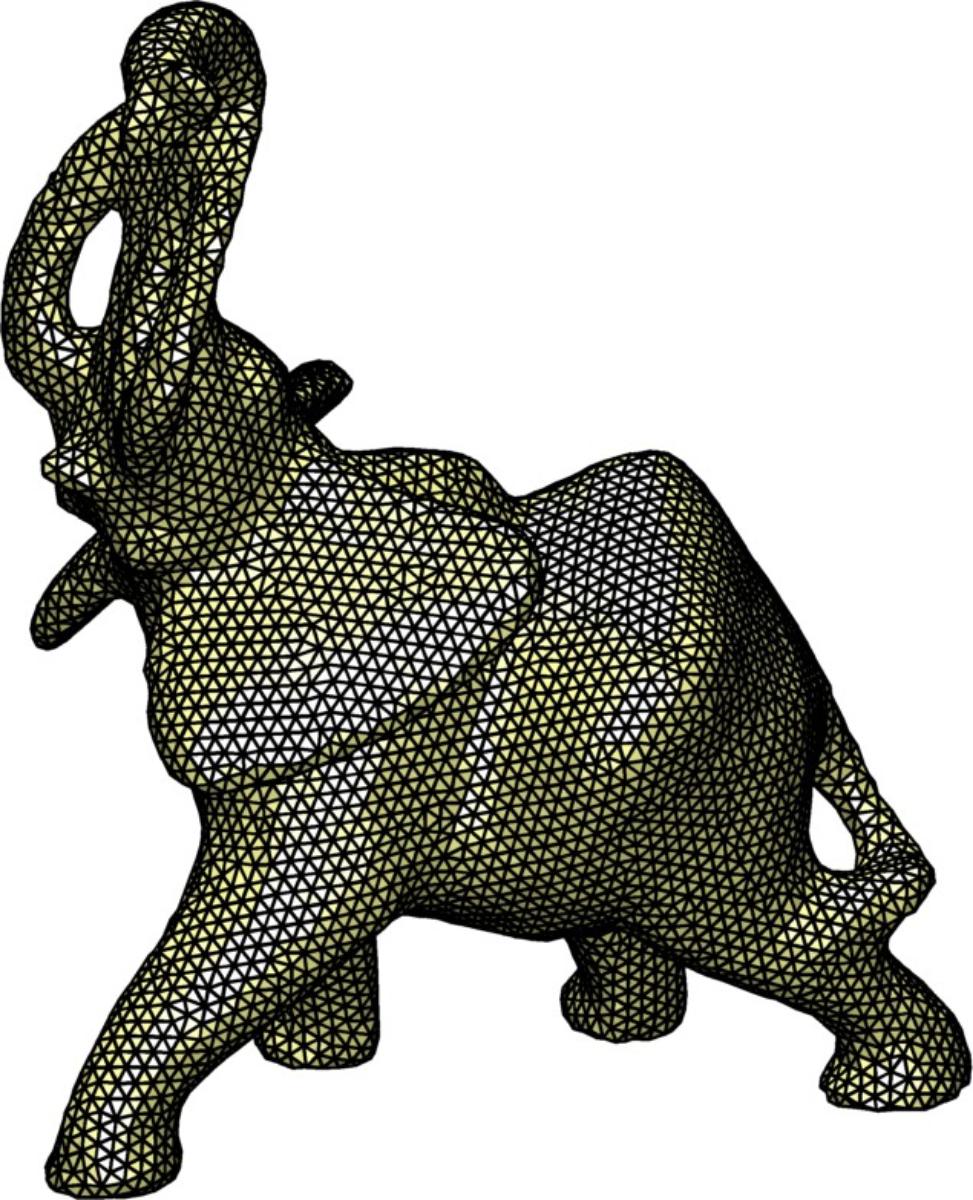} 
\end{center}
\end{minipage} &
\begin{minipage}[c]{.300\textwidth}
\begin{center}
\includegraphics[height=5.50cm]{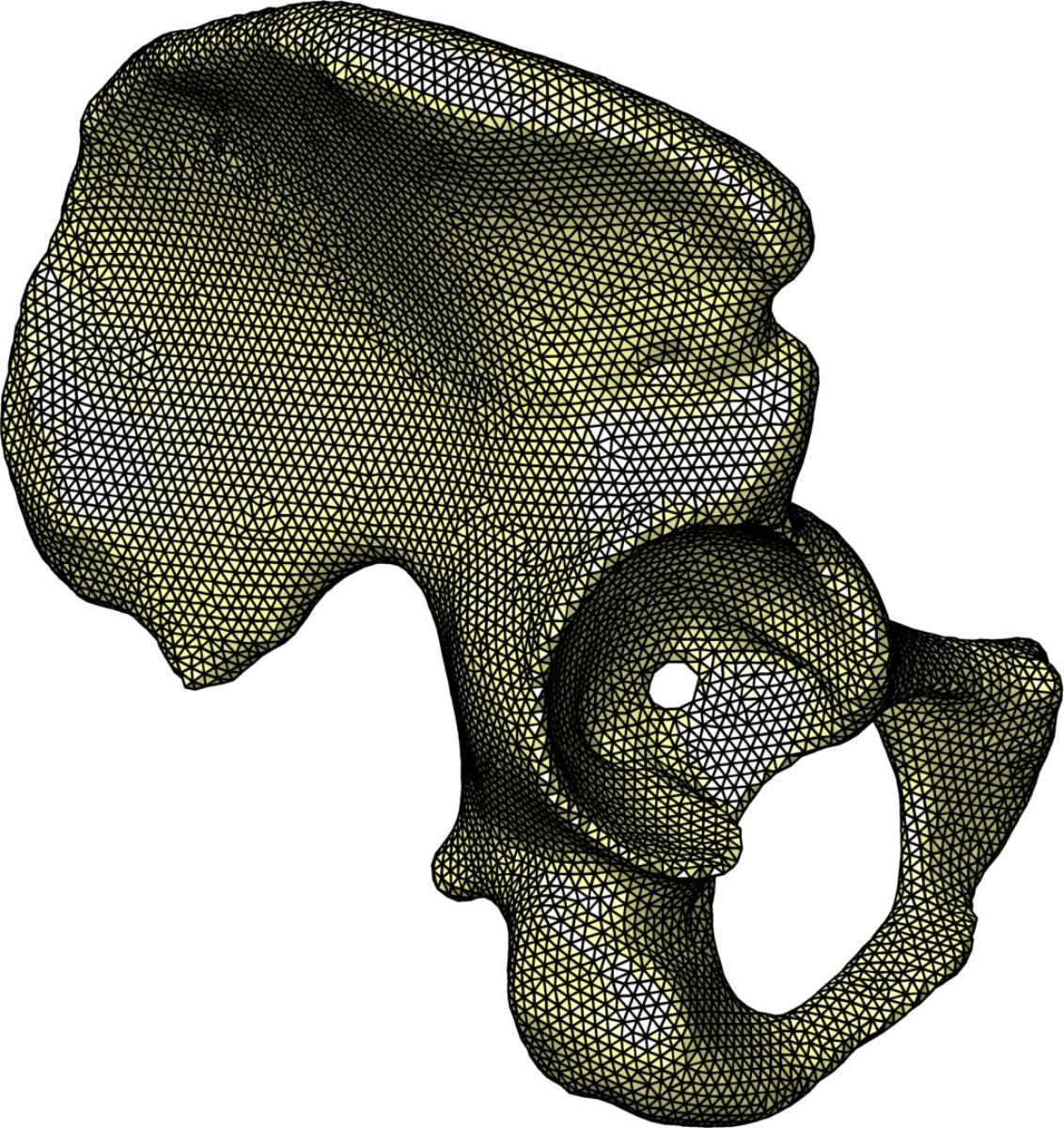} 
\end{center}
\end{minipage} &
\begin{minipage}[c]{.300\textwidth}
\begin{center}
\includegraphics[height=5.50cm]{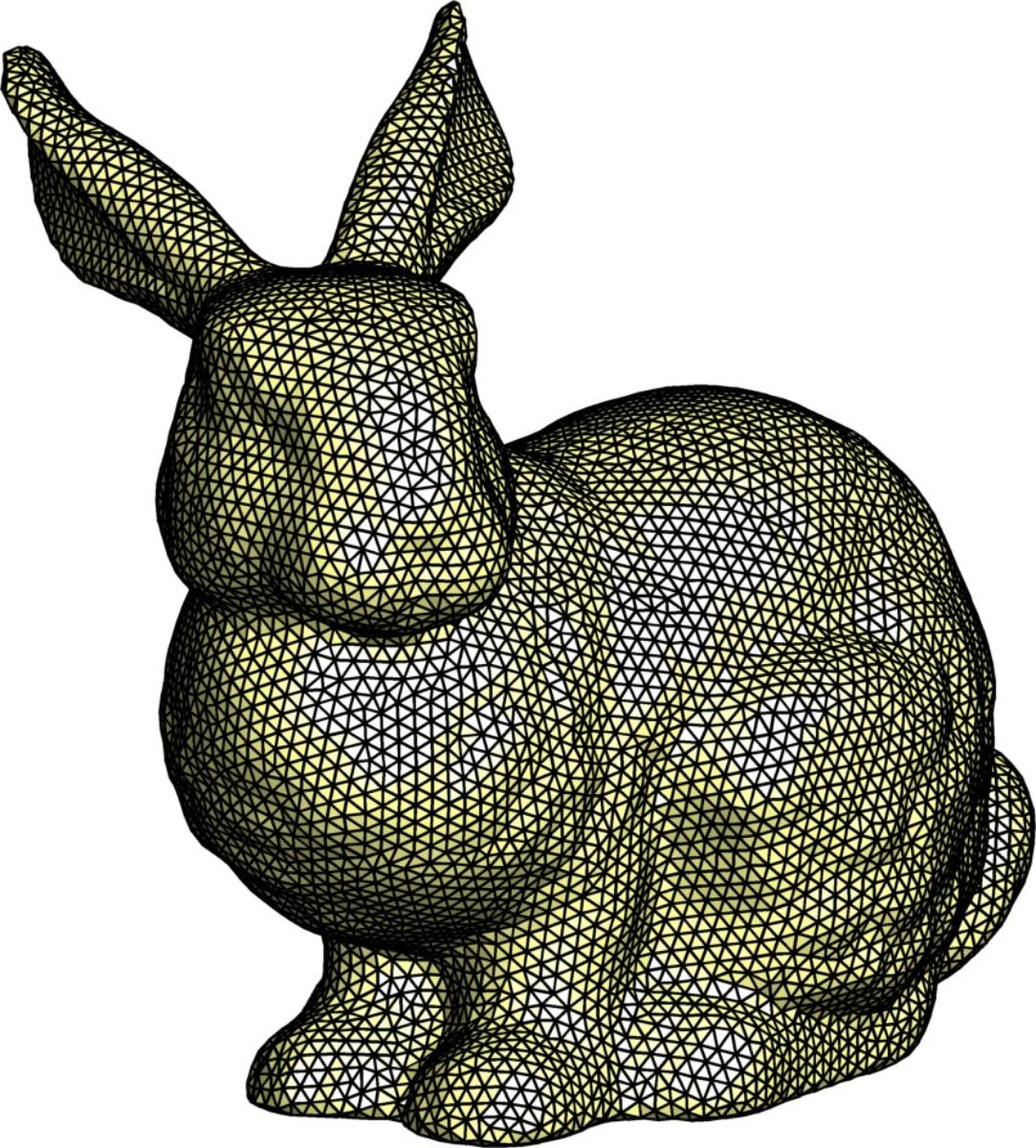} 
\end{center}
\end{minipage}
\rule{0pt}{\tablestrutsize}\rule[-\tablestrutsize]{0pt}{\tablestrutsize} \\

\includegraphics[width=5.25cm]{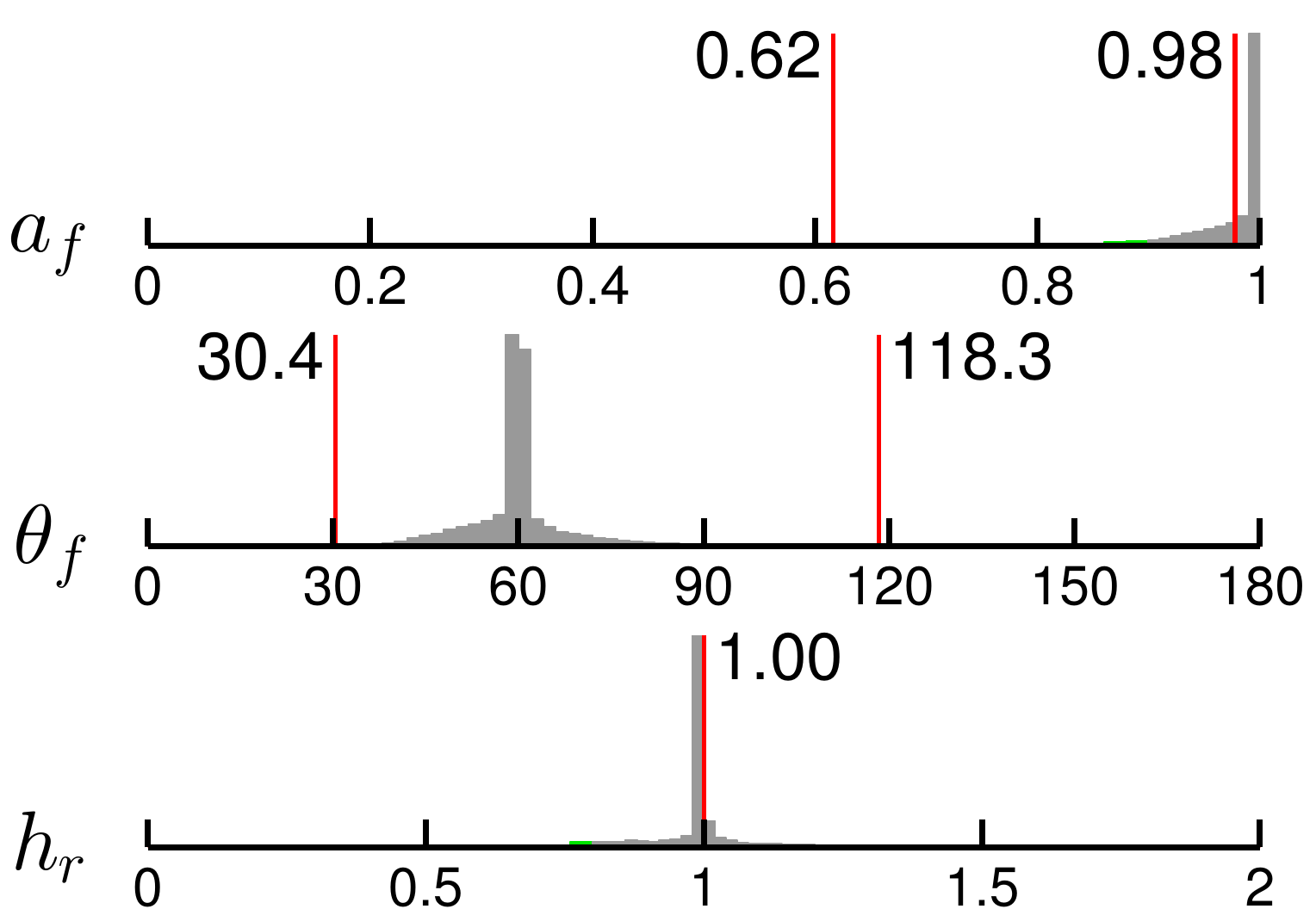} &
\includegraphics[width=5.25cm]{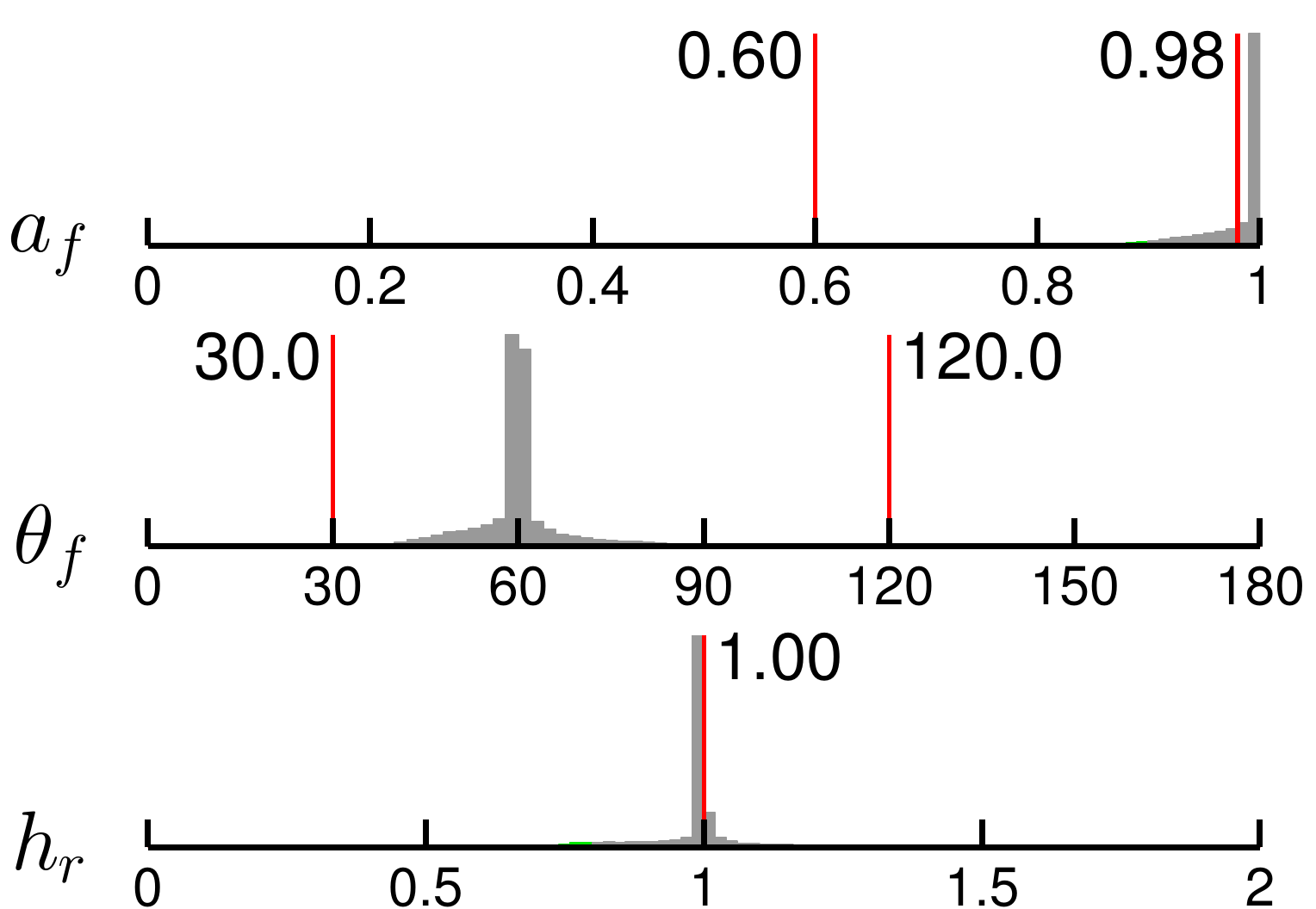} &
\includegraphics[width=5.25cm]{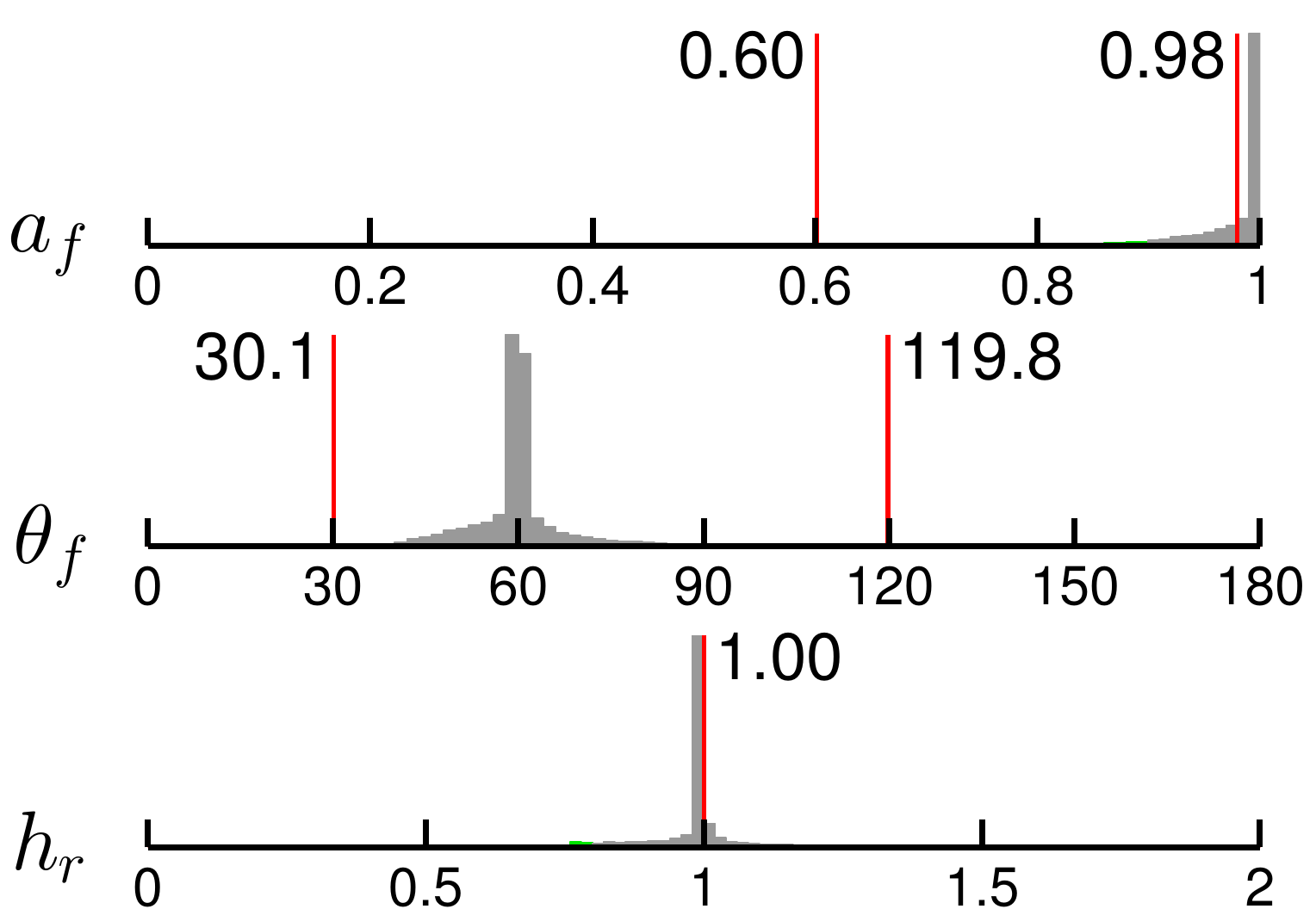}

\\ \hline

\parbox[b][1em][b]{.300\textwidth}{\center (\textsc{cgal--dr}): $|\TS|=16,268$, $\mathrm{t}=1.31\mathrm{s}$} &
\parbox[b][1em][b]{.300\textwidth}{\center (\textsc{cgal--dr}): $|\TS|=23,404$, $\mathrm{t}=1.83\mathrm{s}$} &
\parbox[b][1em][b]{.300\textwidth}{\center (\textsc{cgal--dr}): $|\TS|=19,110$, $\mathrm{t}=1.47\mathrm{s}$} \\

\begin{minipage}[c]{.300\textwidth}
\begin{center}
\includegraphics[height=5.50cm]{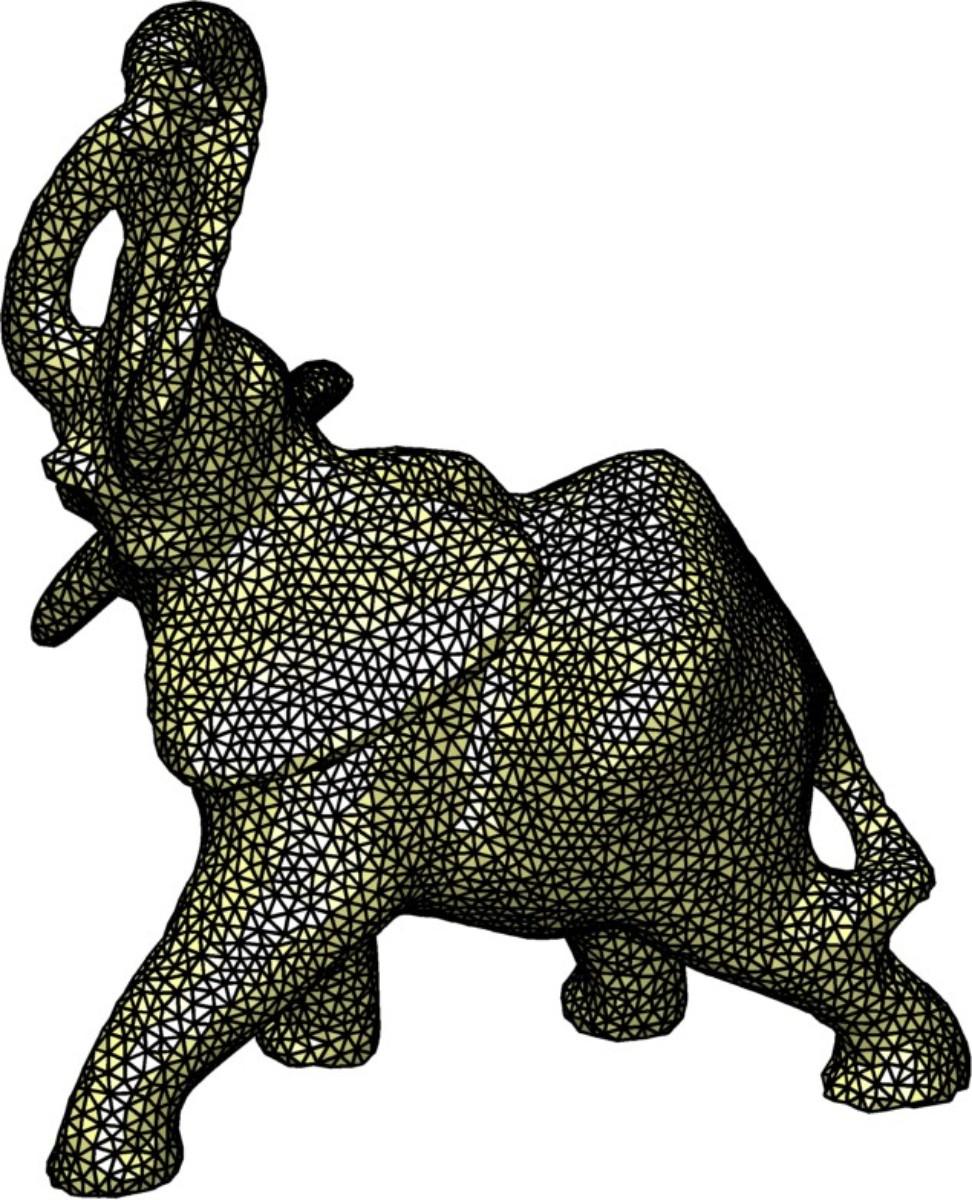} 
\end{center}
\end{minipage} &
\begin{minipage}[c]{.300\textwidth}
\begin{center}
\includegraphics[height=5.50cm]{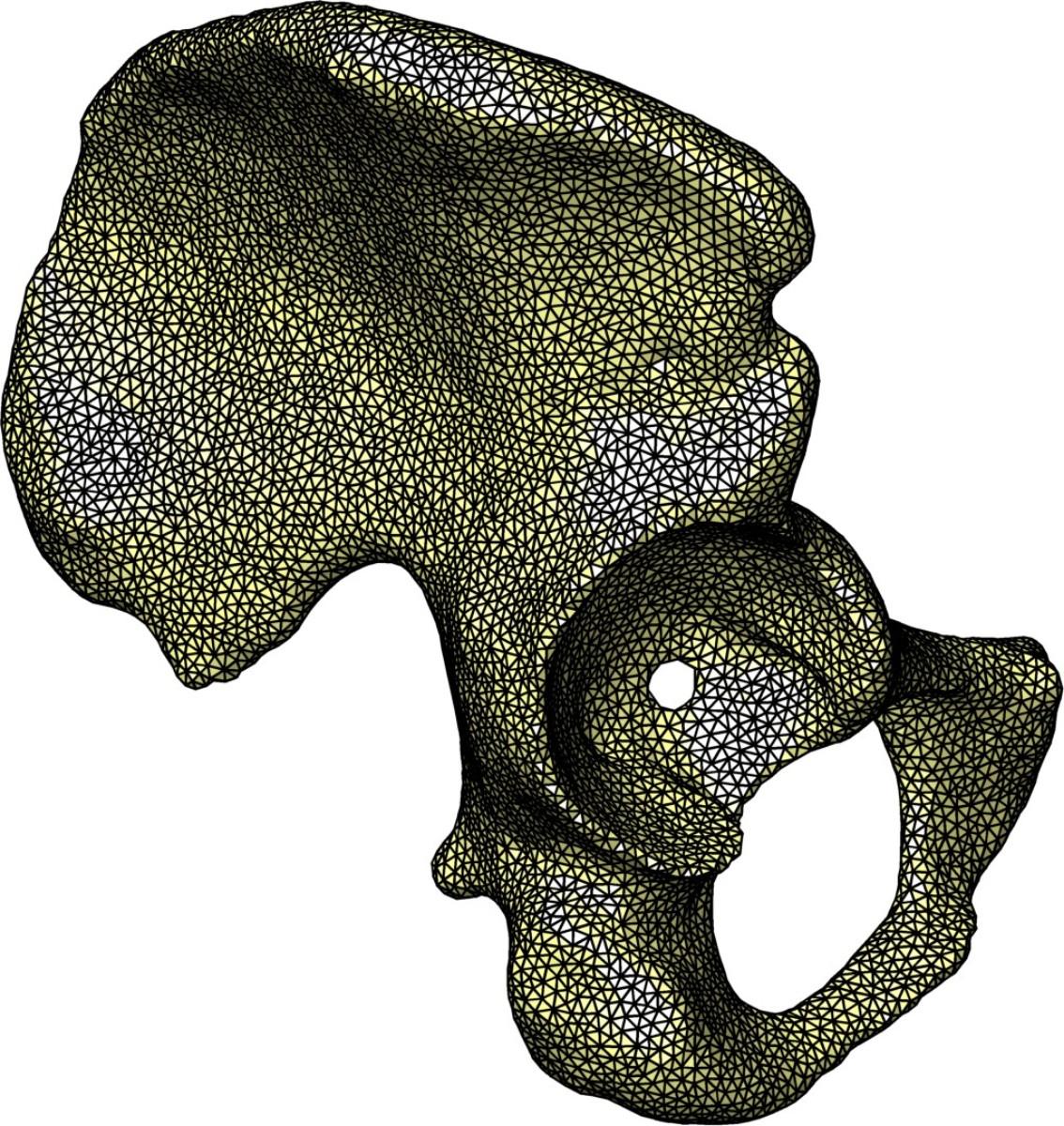} 
\end{center}
\end{minipage} &
\begin{minipage}[c]{.300\textwidth}
\begin{center}
\includegraphics[height=5.50cm]{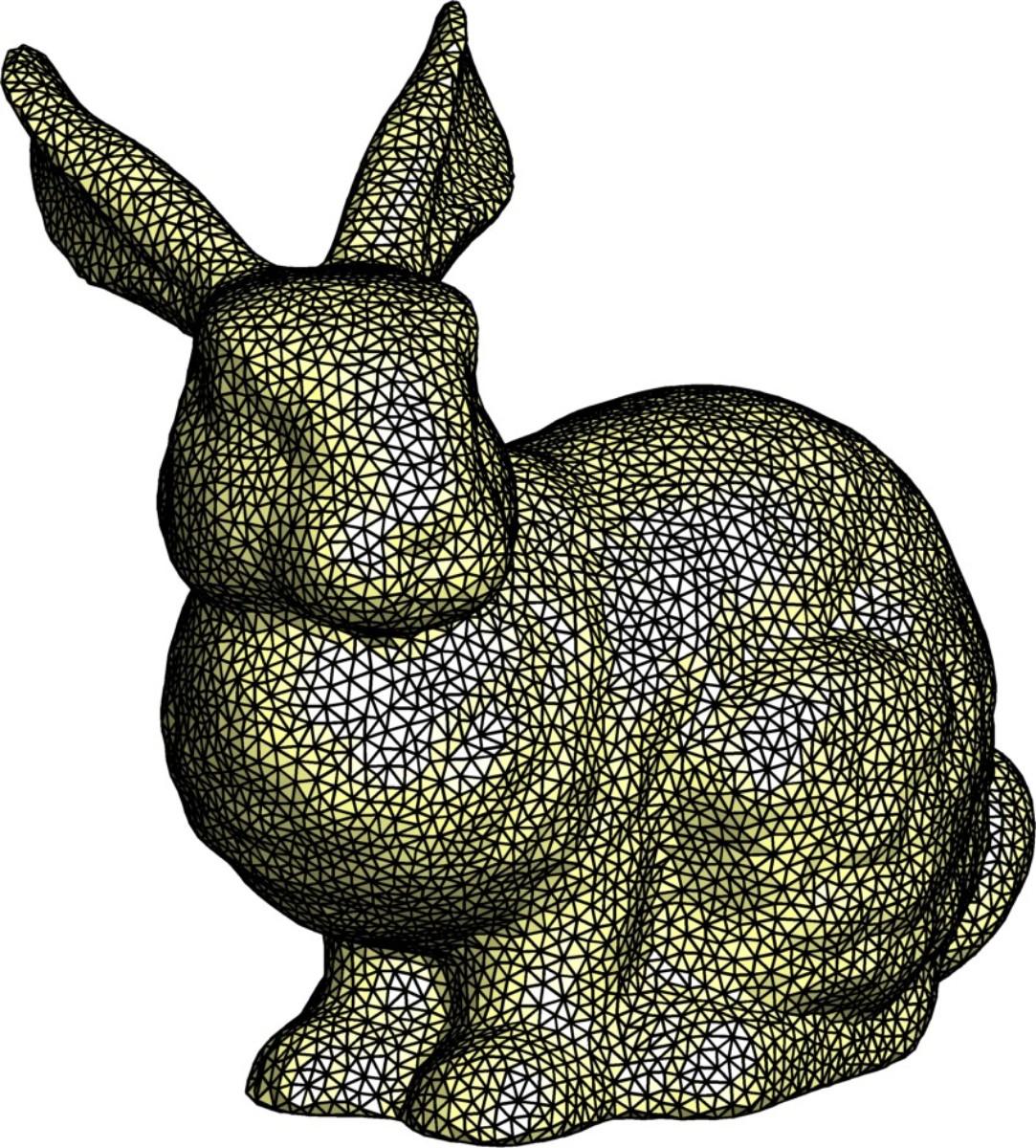} 
\end{center}
\end{minipage}
\rule{0pt}{\tablestrutsize}\rule[-\tablestrutsize]{0pt}{\tablestrutsize} \\

\includegraphics[width=5.25cm]{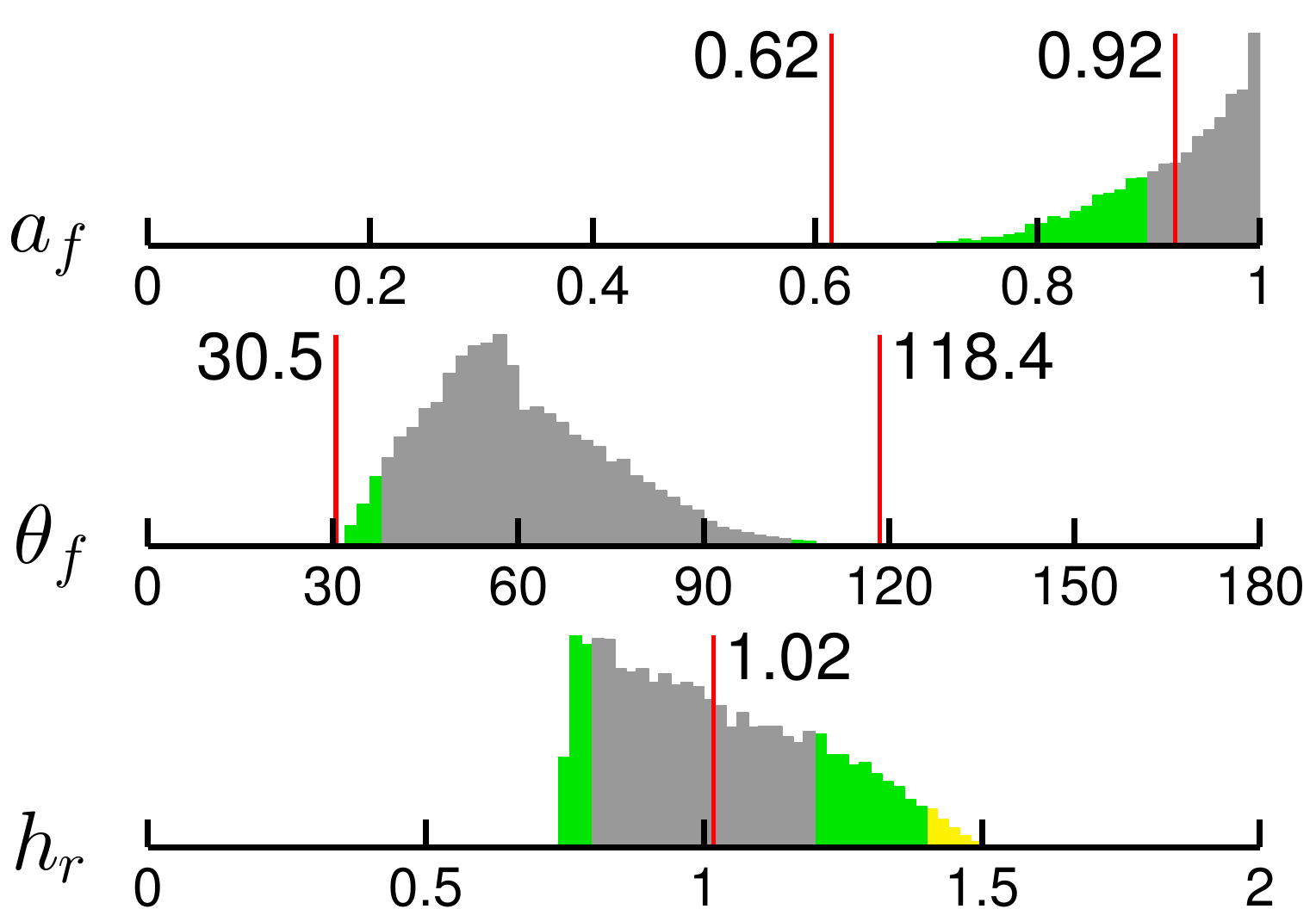} &
\includegraphics[width=5.25cm]{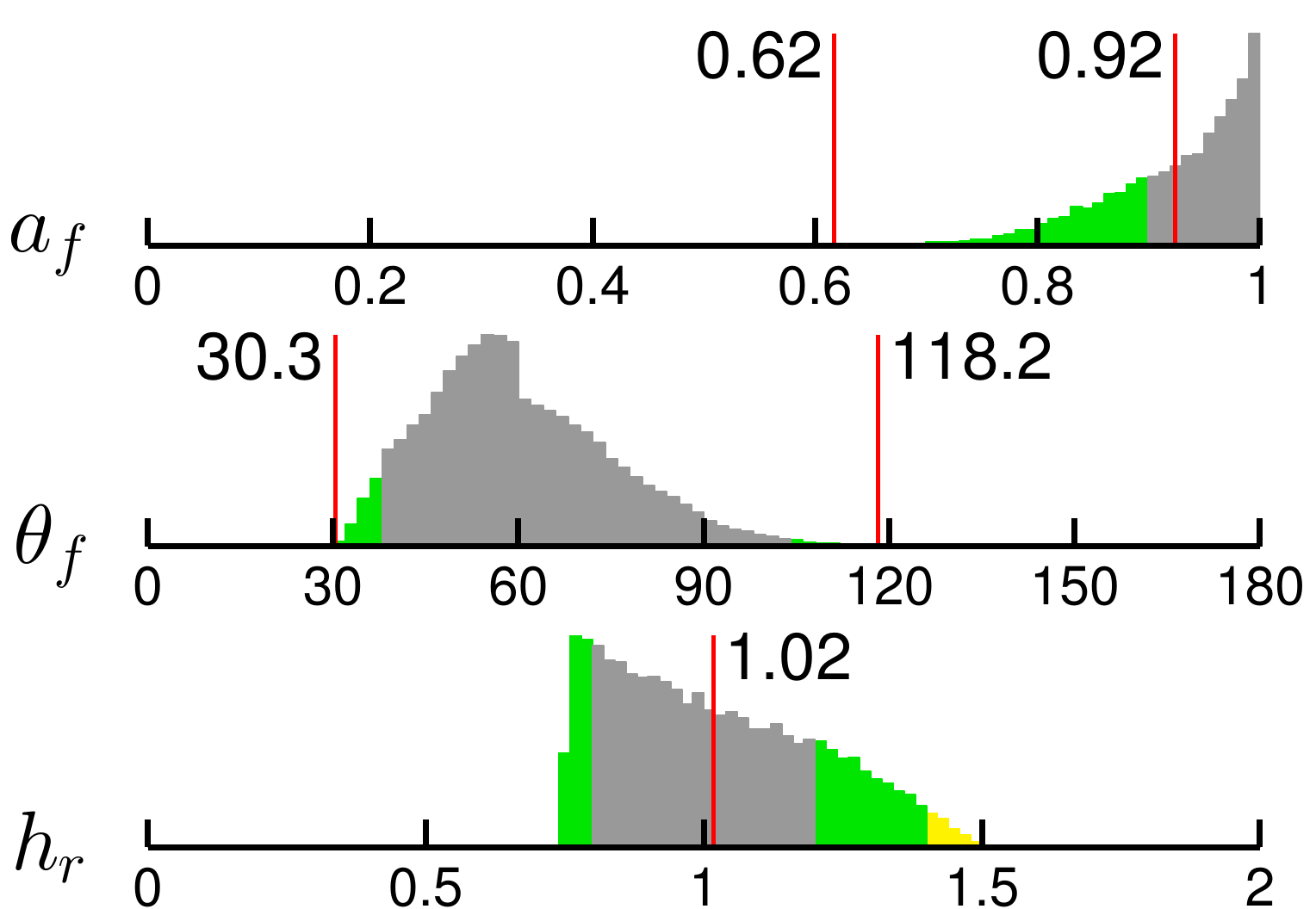} &
\includegraphics[width=5.25cm]{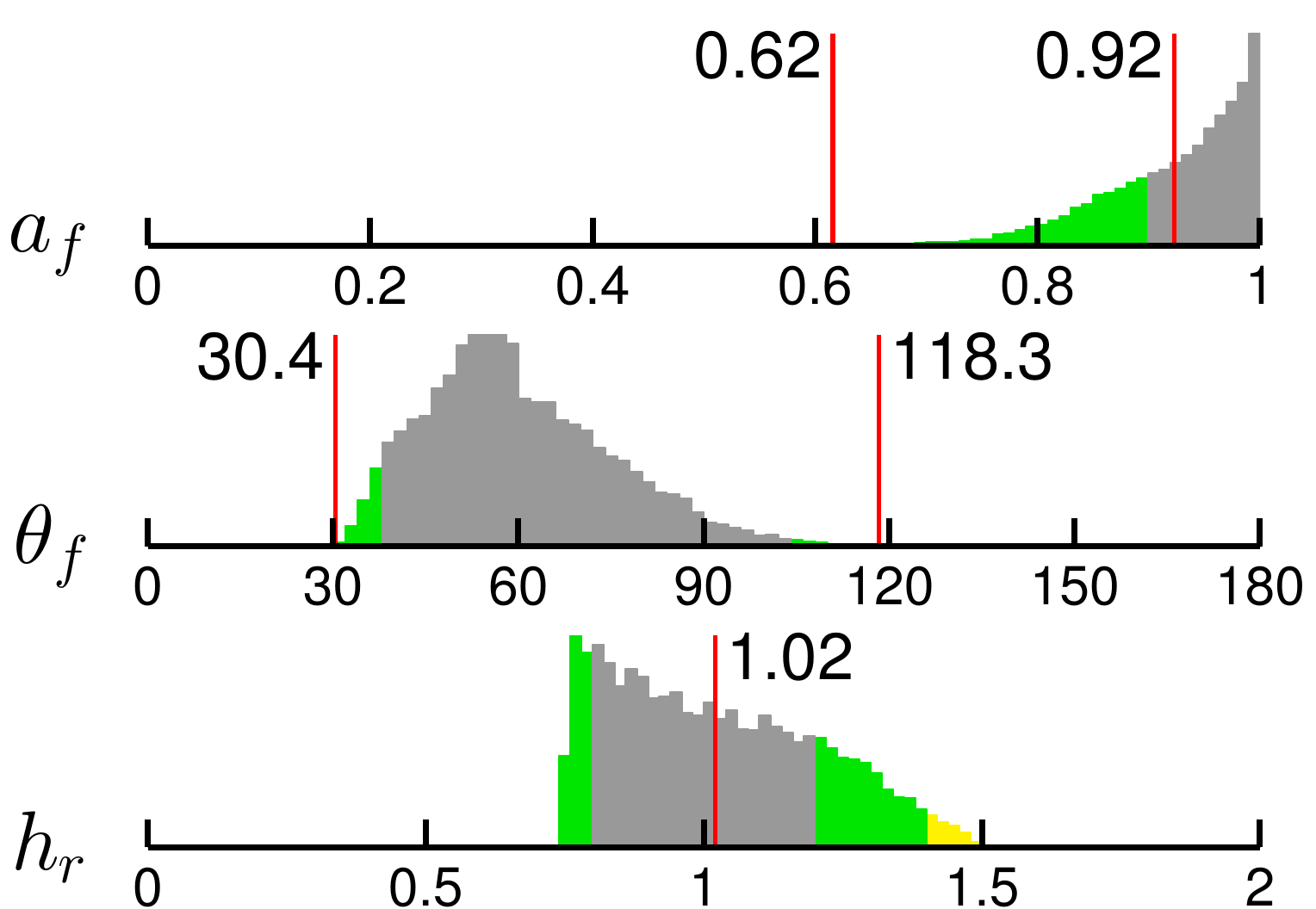}

\end{tabu}}
\end{figure*}

The performance of the Delaunay-refinement and Frontal-Delaunay surface meshing algorithms presented in Sections~\ref{section_delaunay_refinement} and \ref{section_frontal_delaunay} was investigated experimentally, with both techniques used to mesh a series of benchmark problems. Both the Frontal-Delaunay and Delaunay-refinement algorithms were implemented, allowing the performance and output of the two techniques to be compared side-by-side. Due to similarities in the overall algorithmic structure, a common code-base was used, with the algorithms differing only in the type of Steiner vertices inserted, as per the discussions outlined in Section~\ref{section_frontal_delaunay}. Care was taken to ensure that both methods were implemented in a consistent fashion, allowing unbiased comparisons to be made between the algorithms without needing to account for systemic differences arising from particular implementation and/or design choices. Both algorithms were implemented in \cpp and compiled as 64-bit executables. Both the Frontal-Delaunay and Delaunay-refinement algorithms have been implemented as part of the \textsc{jigsaw} meshing package, currently available by request from the first author. These implementations are referred to as \textsc{jgsw--fd} and \textsc{jgsw--dr} throughout the following discussions, with the suffixes \textsc{--fd} and \textsc{--dr} denoting `Frontal-Delaunay' and `Delaunay-refinement', respectively.

In order to provide additional performance information, the well-known \textsc{cgalmesh} implementation \cite{jamin2013cgalmesh} was also included in a subset of the meshing comparisons. The \textsc{cgalmesh} algorithm was sourced from version 4.6 of the \textsc{cgal} package \cite{cgal:pt-t3-15a,cgal:pt-tds3-15a} and was compiled as a 64-bit library. The \textsc{cgalmesh} algorithm is referred to as \textsc{cgal--dr} throughout the following discussions, with the suffix \textsc{--dr} denoting `Delaunay-refinement'. All tests were run using a single core of an Intel i7 processor. Visualisation and post-processing was completed using \textsc{matlab}.

\subsection{Preliminaries}

The various Delaunay-refinement and Frontal-Delaunay algorithms were used to mesh a series of eight benchmark problems, including (i) a set of three test-cases presented in Figure~\ref{figure_surf_cgal_uniform}, in which the performance of the new \textsc{jgsw--fd} algorithm and the existing \textsc{cgal--dr} implementation was compared, (ii) three additional test-cases shown in Figure~\ref{figure_surf_fdvsdr_uniform}, designed to contrast the performance of the \textsc{jgsw--fd} and \textsc{jgsw--dr} algorithms, and (iii) two detailed comparative studies presented in Figures~\ref{figure_surf_fdvsdr_bimba} and~\ref{figure_surf_fdvsdr_kiss} designed to assess the impact of non-uniform mesh-sizing constraints on the performance of the \textsc{jgsw--fd} and \textsc{jgsw--dr} algorithms. 

In all test cases, a constant radius-edge ratio threshold, $\bar{\rho} = 1$ was specified, corresponding to $\theta_{\text{min}}\geq 30^\circ$. Additionally, non-uniform surface discretisation constraints were enforced, setting $\bar{\epsilon}(\mathbf{x})=\beta \bar{h}\left(\mathbf{x}\right)$, with $\beta=\nicefrac{1}{4}$. For all test problems, detailed statistics on element quality are presented, including  histograms of element \textit{area-length} $a\left(f\right)$, \textit{plane-angle} $\theta\left(f\right)$ and \textit{relative-length} $\reledge$ metrics. Histograms further highlight the minimum and mean area-length metrics, the worst-case plane angle bounds, $\theta_{\text{min}}$ and $\theta_{\text{max}}$ and the mean relative edge-length. 

Overall, both Delaunay-refinement algorithms (\textsc{cgal--dr}, \textsc{jgsw--dr}) and the new Frontal-Delaunay technique (\textsc{jgsw--fd}) were found to successfully mesh the full set of benchmark problems, satisfying all constraints. These results demonstrate that the new Frontal-Delaunay technique is capable of generating high-quality surface triangulations, satisfying constraints on element shape-quality $\rho\left(f\right)$, element size $h(\mathbf{x}_{f})$ and surface discretisation error $\epsilon\left(\mathbf{x}_f\right)$, consistent with conventional Delaunay-refinement approaches.

\begin{figure*}[p]
\centering
\caption{Triangulations for the \textsc{ifp2}, \textsc{femur} and \textsc{rocker} test problems, showing output for the \textsc{jgsw--fd} (upper) and \textsc{jgsw--dr} (lower) algorithms. Meshes were built using: (i) uniform mesh size functions $\bar{h}\left(\mathbf{x}\right)=\alpha$, with $\alpha\in\mathbb{R}^{+}$, (ii) tight constraints on element shape-quality, such that $\theta_{\text{min}}\geq 30^\circ$, and (iii) non-uniform surface discretisation functions $\bar{\epsilon} = \left(\nicefrac{1}{4}\right)\, \bar{h}\left(\mathbf{x}\right)$. Element counts $|\TS|$, and algorithm run-times $(\mathrm{t})$ are included for each case. Normalised histograms of element area-length ratio $a\left(f\right)$, plane-angle $\theta\left(f\right)$ and relative-length $\reledge$ are also illustrated.}

\label{figure_surf_fdvsdr_uniform}

{
\footnotesize
\tabulinesep=1pt

\medskip

\begin{tabu} {c|c|c}

\parbox[b][1em][b]{.300\textwidth}{\center (\textsc{jgsw--fd}): $|\TS|=87,083$, $\mathrm{t}=6.24\mathrm{s}$} &
\parbox[b][1em][b]{.300\textwidth}{\center (\textsc{jgsw--fd}): $|\TS|= 7,854$, $\mathrm{t}=0.48\mathrm{s}$} &
\parbox[b][1em][b]{.300\textwidth}{\center (\textsc{jgsw--fd}): $|\TS|=13,023$, $\mathrm{t}=0.72\mathrm{s}$} \\

\begin{minipage}[c]{.300\textwidth}
\begin{center}
\includegraphics[width=5.50cm]{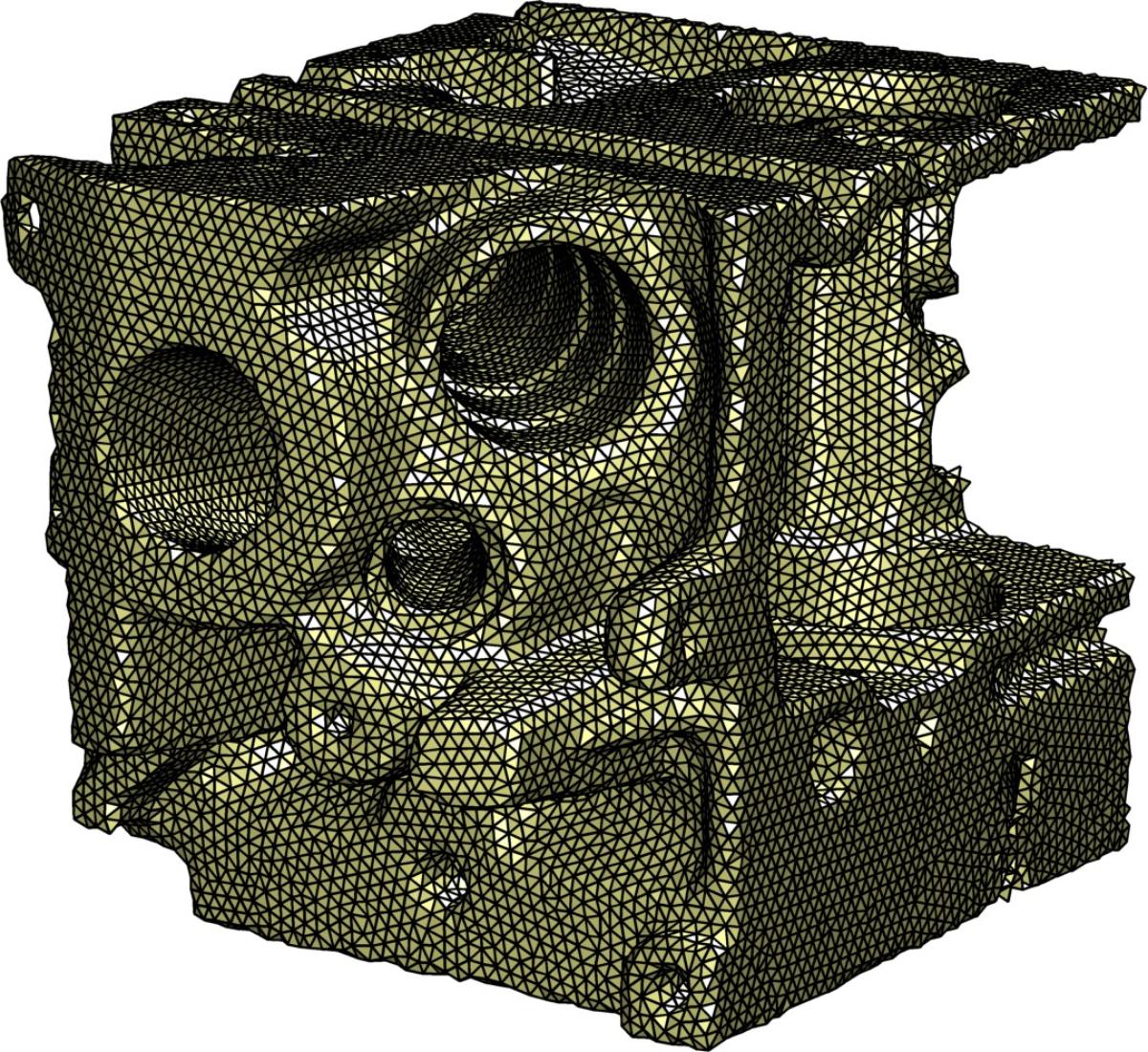} 
\end{center}
\end{minipage} &
\begin{minipage}[c]{.300\textwidth}
\begin{center}
\includegraphics[height=5.50cm]{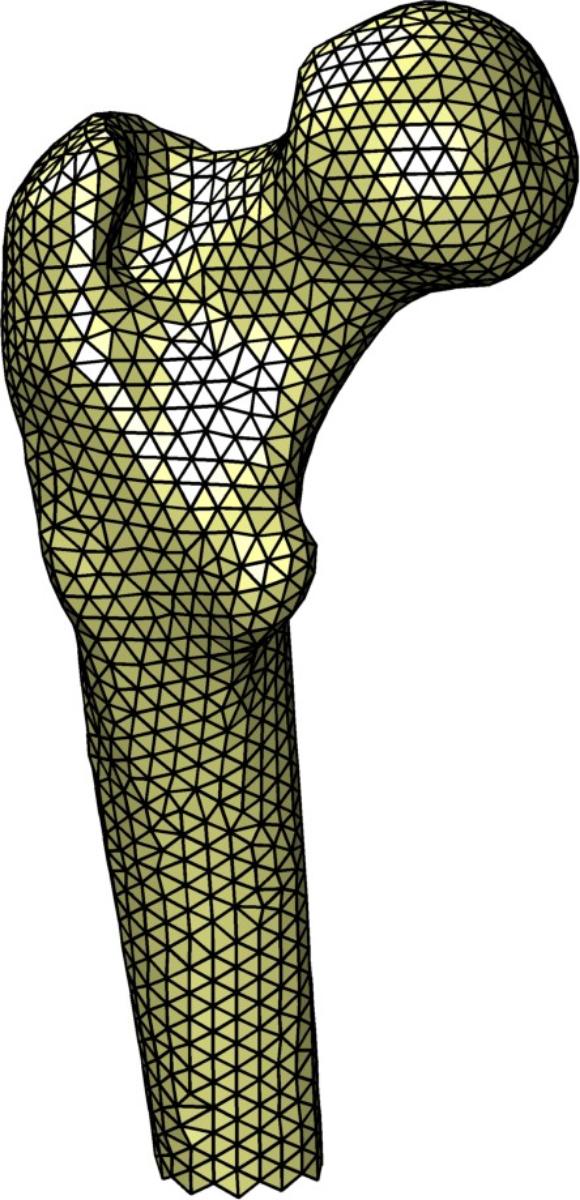} 
\end{center}
\end{minipage} &
\begin{minipage}[c]{.300\textwidth}
\begin{center}
\includegraphics[width=5.50cm]{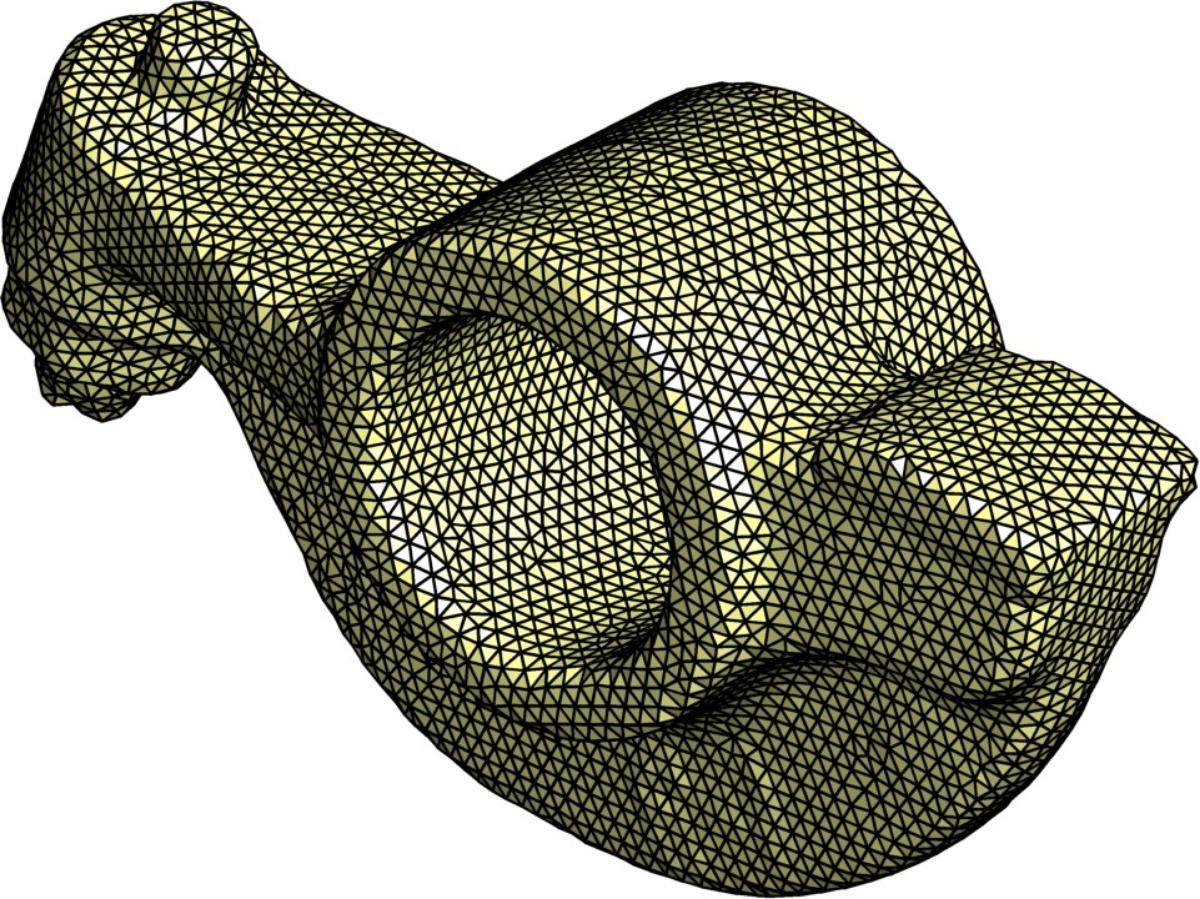} 
\end{center}
\end{minipage}
\rule{0pt}{\tablestrutsize}\rule[-\tablestrutsize]{0pt}{\tablestrutsize} \\

\includegraphics[width=5.25cm]{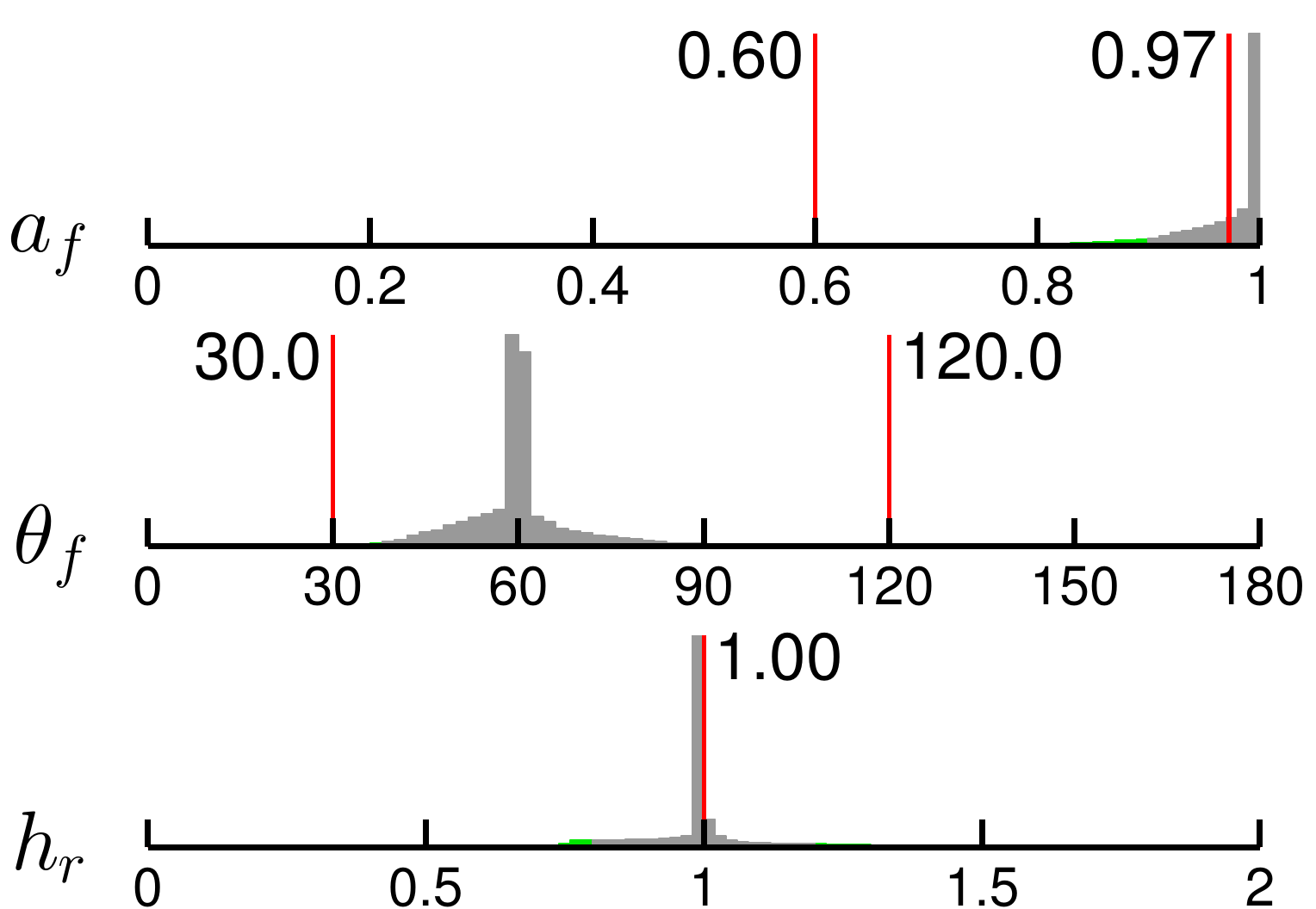} &
\includegraphics[width=5.25cm]{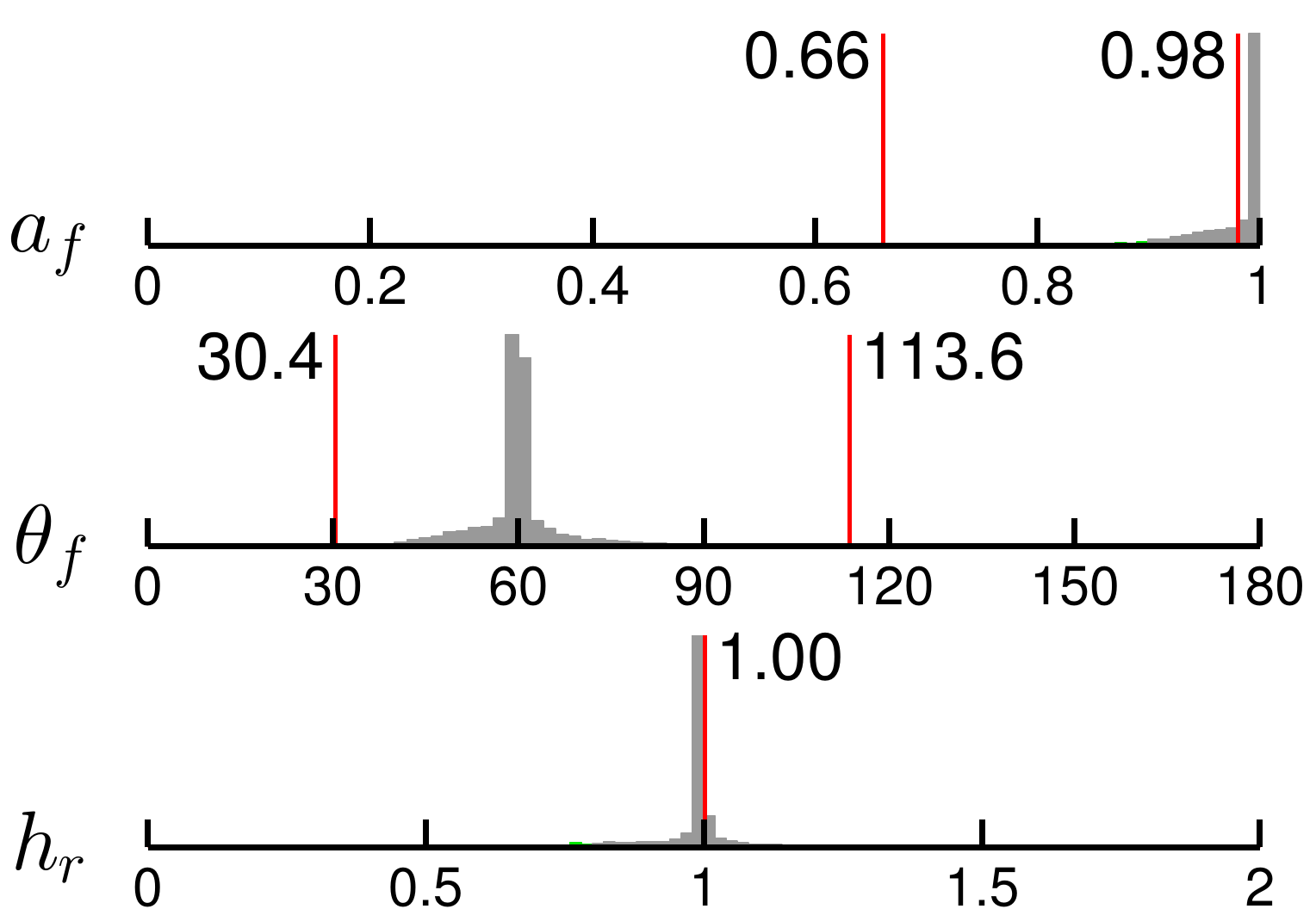} &
\includegraphics[width=5.25cm]{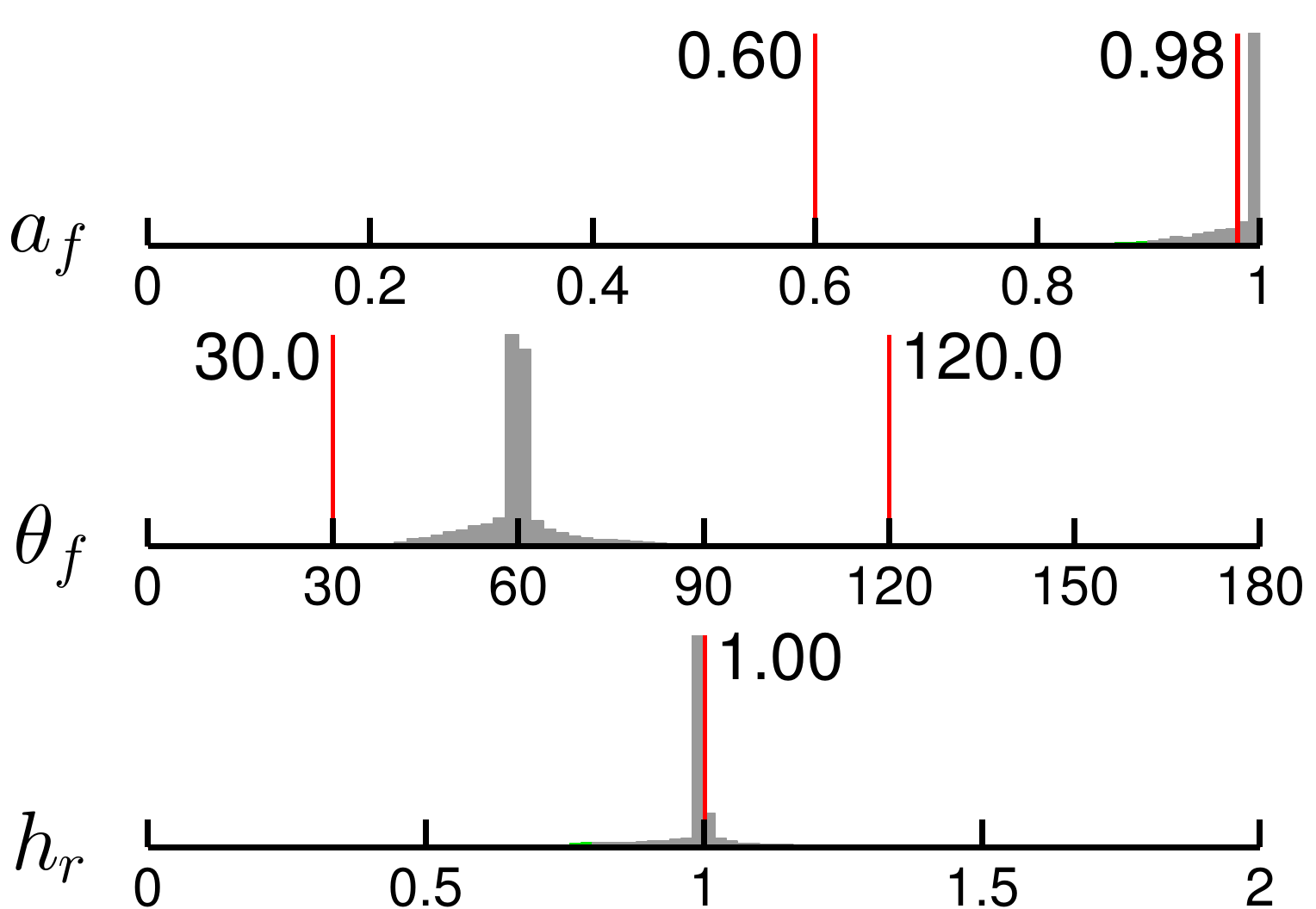}

\\ \hline

\parbox[b][1em][b]{.300\textwidth}{\center (\textsc{jgsw--dr}): $|\TS|=91,706$, $\mathrm{t}=4.20\mathrm{s}$} &
\parbox[b][1em][b]{.300\textwidth}{\center (\textsc{jgsw--dr}): $|\TS|= 8,412$, $\mathrm{t}=0.33\mathrm{s}$} &
\parbox[b][1em][b]{.300\textwidth}{\center (\textsc{jgsw--dr}): $|\TS|=13,644$, $\mathrm{t}=0.50\mathrm{s}$} \\

\begin{minipage}[c]{.300\textwidth}
\begin{center}
\includegraphics[width=5.50cm]{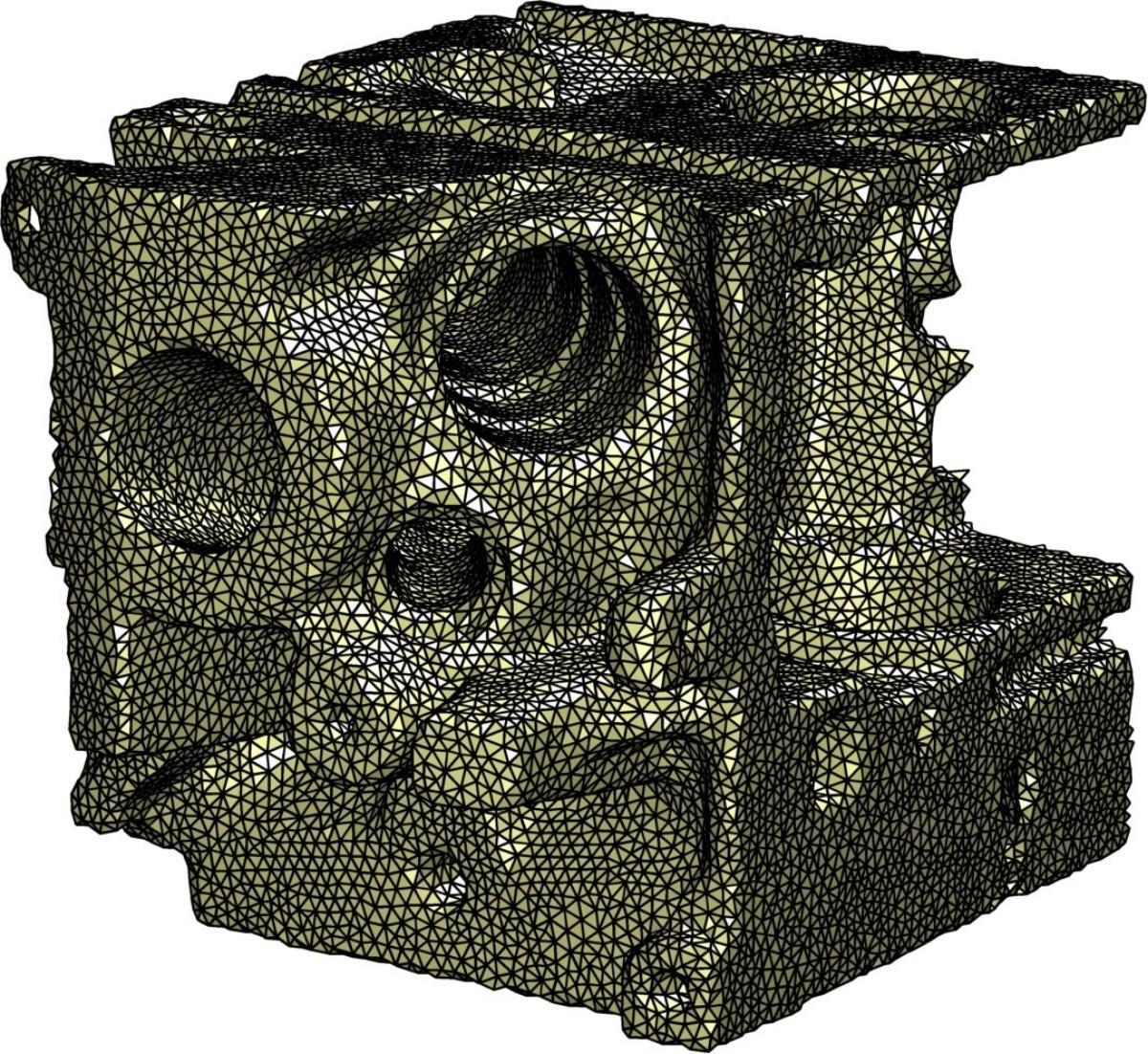} 
\end{center}
\end{minipage} &
\begin{minipage}[c]{.300\textwidth}
\begin{center}
\includegraphics[height=5.50cm]{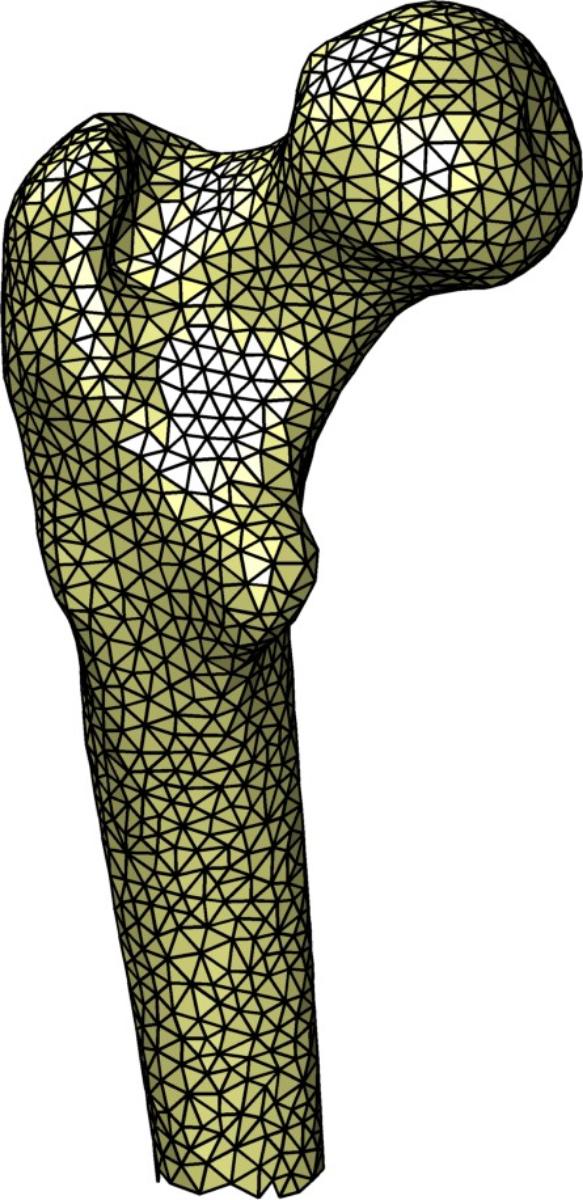} 
\end{center}
\end{minipage} &
\begin{minipage}[c]{.300\textwidth}
\begin{center}
\includegraphics[width=5.50cm]{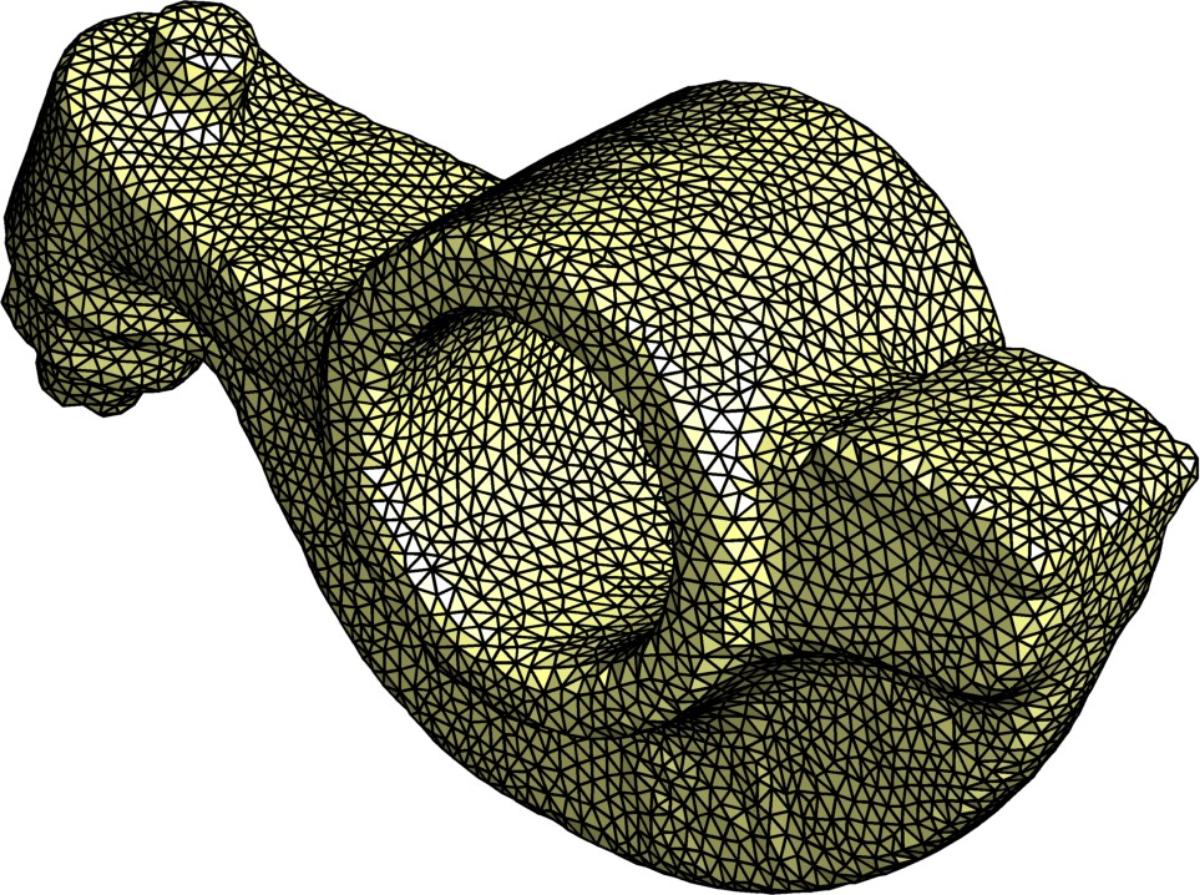} 
\end{center}
\end{minipage}
\rule{0pt}{\tablestrutsize}\rule[-\tablestrutsize]{0pt}{\tablestrutsize} \\

\includegraphics[width=5.25cm]{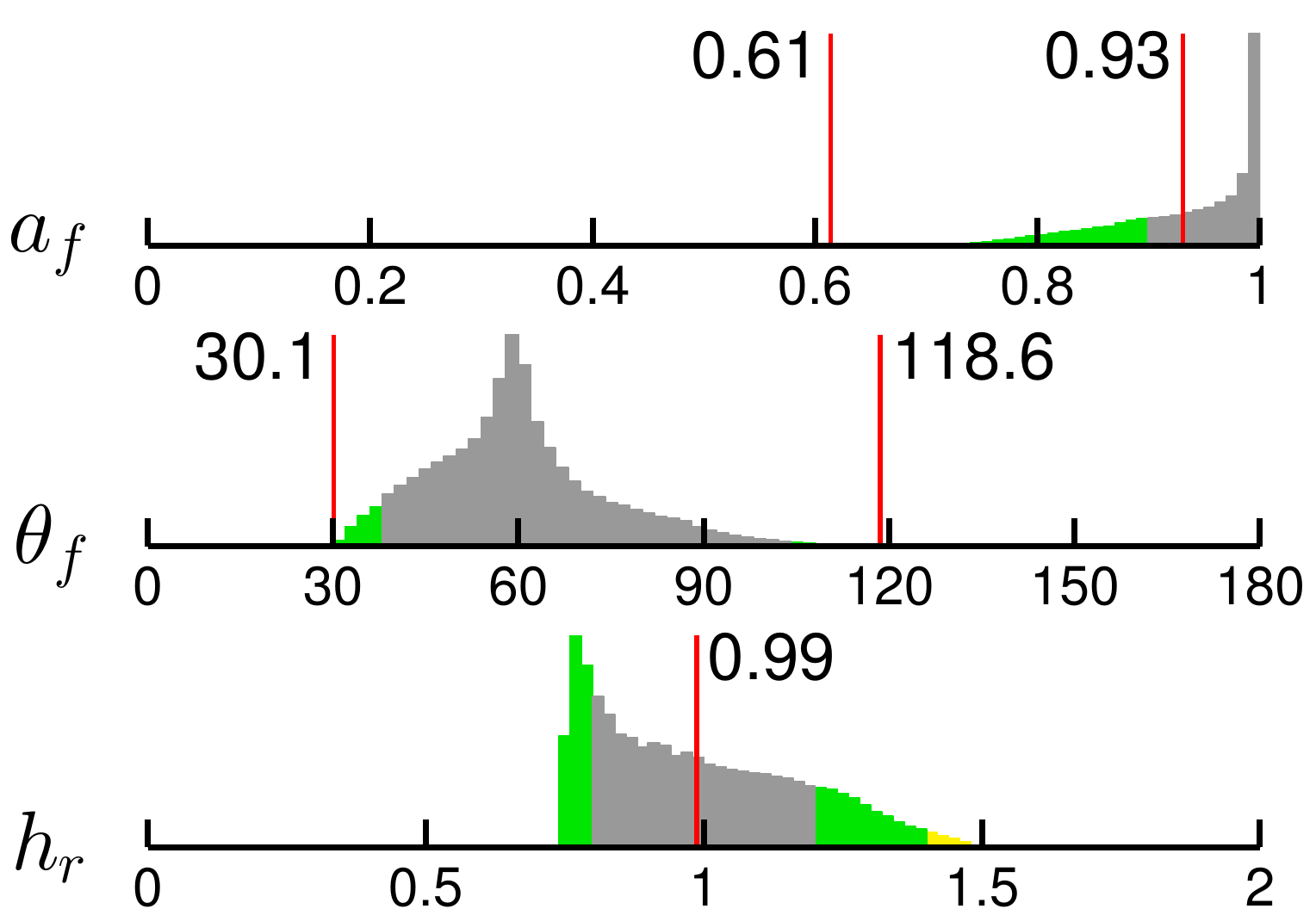} &
\includegraphics[width=5.25cm]{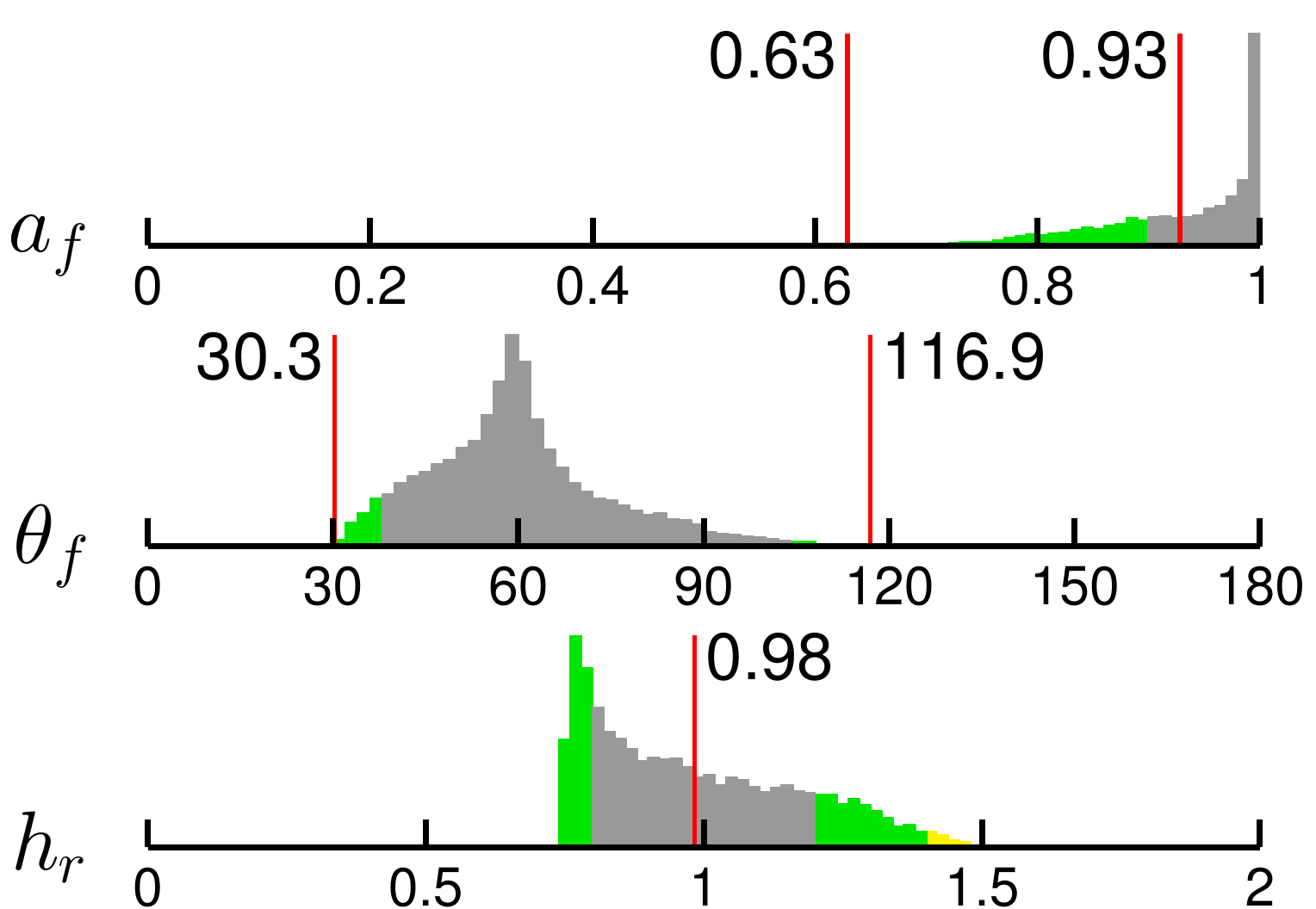} &
\includegraphics[width=5.25cm]{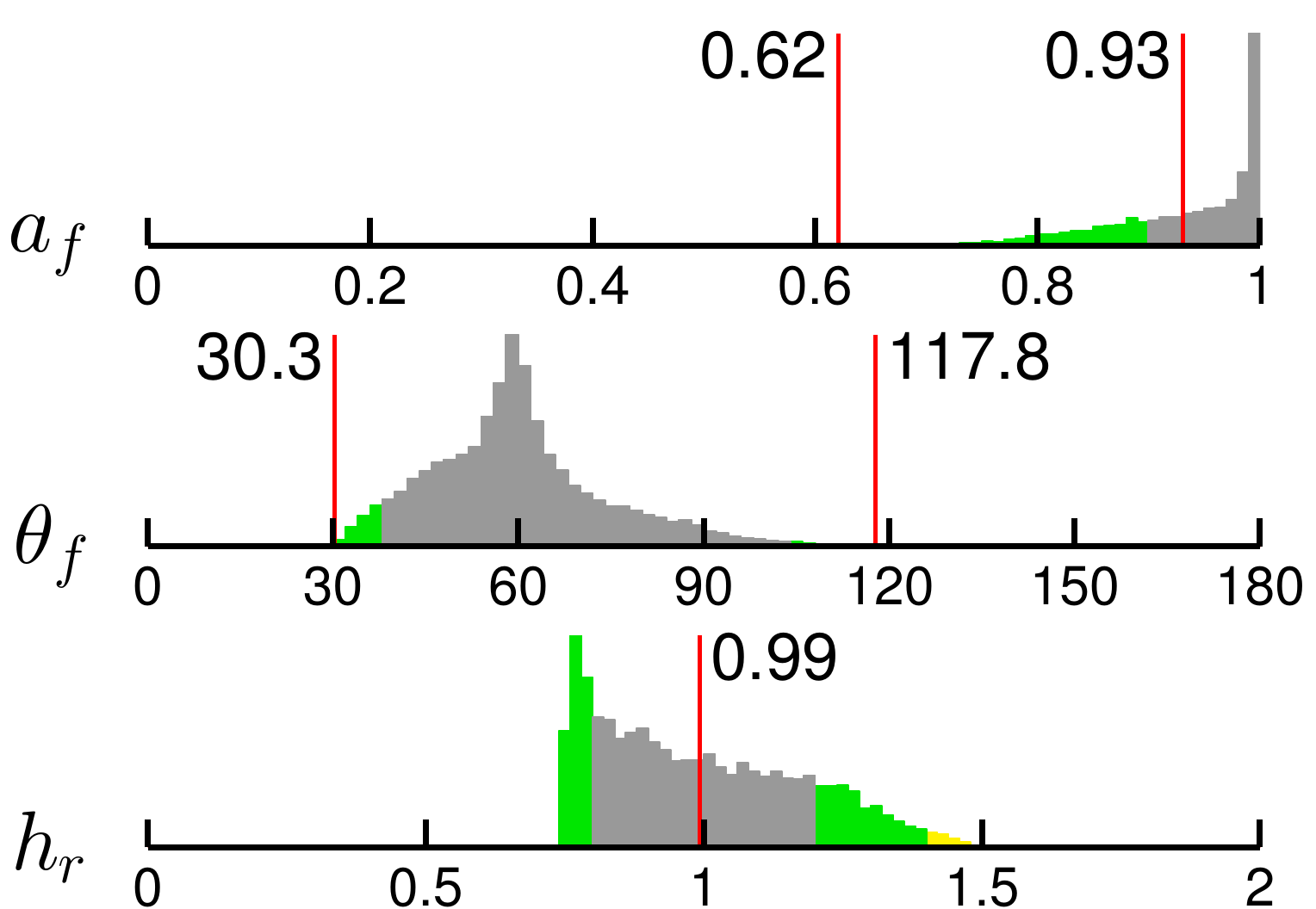}

\end{tabu}}
\end{figure*}

\subsection{Mesh Quality Metrics}

Before moving on to a detailed analysis of results, a number of mesh quality metrics are first introduced.

\begin{definition}
Given a $2$-simplex $f$, the \textit{plane-angle} between any two adjacent edges is given by:
\begin{gather}
\cos(\theta_{i,j}) = \frac{\mathbf{e}_{i}\cdot\mathbf{e}_{j}}{\|\mathbf{e}_{i}\| \|\mathbf{e}_{j}\|}
\end{gather}
for all adjacent pairs $\mathbf{e}_{i}$,$\mathbf{e}_{j}\in f$.
\end{definition}

Clearly, high-quality surface triangulations include a majority of element plane-angles clustered about $60^\circ$.

\begin{definition}
Given a $2$-simplex $f$, its \textit{area-length} ratio, $a\left(f\right)$, is given by:
\begin{gather}
a\left(f\right) = \frac{4\sqrt{3}}{3}\frac{A_{f}}{\|\mathbf{e}_{\text{rms}}\|^{2}},
\end{gather} 
where $A_{f}$ is the signed-area of $f$ and $\|\mathbf{e}_{\text{rms}}\|$ is the root-mean-square edge length.
\end{definition}

The area-length ratio is a robust, scalar measure of element shape-quality, and is typically normalised to achieve a score of $+1$ for \textit{ideal} elements. The area-length ratio decreases with increasing distortion, achieving a score of ${+0}$ for degenerate elements and ${-1}$ for fully inverted elements. 

\begin{definition}
Given a $2$-simplex $f$, the \textit{relative-length} of its edges is given by:
\begin{gather}
h_{r}(e_{i}) = \frac{\|\mathbf{e}_{i}\|}{\bar{h}\left(\mathbf{x}_{e}\right)}
\end{gather}
where $\|\mathbf{e}_{i}\|$ is the length of the $i$-th edge and $\bar{h}(\mathbf{x}_{e})$ is the mesh size function sampled at the edge midpoint.
\end{definition}

The \textit{relative-length} distribution $\reledge$ is a measure of size-function conformance, expressing the ratio of actual-to-desired edge length for all edges $e\in\DelS{X}$. A value of $\reledge=1$ indicates perfect conformance.

\begin{definition}
Given a distribution of element-wise plane angles $\theta(f_{i})$, the Mean Absolute Deviation $\MAD$ is given by:
\begin{gather}
\MAD = \frac{1}{n}\sum_{j=1}^{n}|\bar{\theta}(f_{i}) - \theta_{j}|
\end{gather}
where $\bar{\theta}(f_{i})$ denotes the mean value of the underlying distribution $\theta(f_{i})$.
\end{definition}

$\MAD$ is a measure of the \textit{spread} of the distribution of element angles, with smaller values indicative of distributions that are more tightly clustered about the mean value. In this study, $\MAD$ is used as a measure of \textit{average} element quality, with smaller values associated with higher quality meshes.

\subsection{A Comparison of \textsc{jgsw--fd} and \textsc{cgal--dr}}
\label{section_fdvscgal}

The performance of the \textsc{jgsw--fd} and \textsc{cgal--dr} algorithms was assessed using a set of three test-cases. The \textsc{elephant}, \textsc{hip} and \textsc{bunny} geometries were meshed using uniform mesh-size constraints, with a small constant value, $\bar{h}\left(\mathbf{x}\right)=\alpha$, imposed globally, where $\alpha$ was chosen to be approximately 2\% of the mean bounding-box dimension associated with each model. Due to differences in the way that mesh-size constraints are interpreted by the two meshing packages, the mesh-size value selected for the \textsc{cgal--dr} algorithm was reduced by a factor of $\nicefrac{4}{3}$. This reduction accounts for the fact that in \textsc{cgal--dr}, mesh-size constraints are imposed with respect to the element circumradii, whereas the \textsc{jgsw--fd} algorithm treats the mesh-size function in terms of element edge length. Both algorithms were found to produce meshes incorporating consistent mean edge-lengths based on these modified mesh-size values.

Overall, the results shown in Figure~\ref{figure_surf_cgal_uniform} demonstrate that both the \textsc{jgsw--fd} and \textsc{cgal--dr} algorithms generate high-quality surface meshes for all test cases -- satisfying the required element shape-quality, mesh-size and surface discretisation error thresholds. Focusing on the distribution of element shape-quality explicitly, it is clear that the \textsc{jgsw--fd} algorithm achieves significantly better results, producing meshes with higher mean area-length ratios in all cases. In fact, with mean area-length ratios of $\bar{a}(f)\simeq 0.98$ (compared to $\bar{a}(f)\simeq 0.92$ for \textsc{cgal--dr}), it is clear that the vast majority of elements generated by the \textsc{jgsw--fd} algorithm are of near-ideal shape. This improvement in mean mesh-quality is also evident in the distribution of element plane-angles, with the \textsc{jgsw--fd} approach generating histograms clustered more tightly about $60^\circ$. These results are confirmed by an analysis of the mean-absolute-deviation in angle, $\MAD$ -- included in Table~\ref{table_mad}, showing that the \textit{spread} of the plane-angle distribution is reduced by a factor of approximately three through use of the \textsc{jgsw--fd} algorithm. Finally, analysis of the relative-length distributions shows that meshes generated using the \textsc{jgsw--fd} algorithm tightly conform to the imposed mesh size constraints, with a tight clustering of $\reledge$ about $1$. In contrast, output generated using the \textsc{cgal--dr} scheme is seen to incorporate significant sizing `error', typified by broad distributions of relative-length, straddling $\reledge\simeq 1$.

In terms of computational expense, meshes generated by the \textsc{jgsw--fd} and \textsc{cgal--dr} algorithms are shown to be very similar in size. Additionally, despite the additional work required to support the off-centre refinement scheme, the \textsc{jgsw--fd} algorithm appears to be slightly more efficient than \textsc{cgal--dr}, though the difference in run-time performance is not significant.

\subsection{A Comparison of \textsc{jgsw--fd} and \textsc{jgsw--dr}}
\label{section_fdvsdr}

Consistent with the analysis presented in Section~\ref{section_fdvscgal}, the performance of the \textsc{jgsw--fd} and \textsc{jgsw--dr} algorithms was next assessed. While \textsc{cgal--dr} and \textsc{jgsw--dr} are representative of the same class of meshing \textit{algorithms}, the benchmarks included in this section allow for an unbiased comparison of the Frontal-Delaunay and Delaunay-refinement approaches, with both the \textsc{jgsw--fd} and \textsc{jgsw--dr} algorithms benefiting from the same set of implementation design and optimisation decisions. 

Consistent with the previous analysis, the performance of the \textsc{jgsw--fd} and \textsc{jgsw--dr} algorithms was assessed using a set of three test-cases. The \textsc{ifp2}, \textsc{femur} and \textsc{rocker} geometries shown in Figure~\ref{figure_surf_fdvsdr_uniform} were meshed using uniform mesh-size constraints, with a small constant value, $\bar{h}\left(\mathbf{x}\right)=\alpha$, imposed globally, where $\alpha$ was chosen to be approximately 2\% of the mean bounding-box dimension associated with each model. An analysis of the distributions of element area-length, plane-angle and relative edge-length measures confirms much the behaviour noted in Section~\ref{section_fdvscgal} -- in terms of overall mesh-quality and size-conformance, the \textsc{jgsw--fd} algorithm is a clear winner.

An analysis of computational expense leads to different conclusions, with the \textsc{jgsw--dr} algorithm seen to produce meshes that are slightly and yet consistently larger than those generated using \textsc{jgsw--fd}. The \textsc{jgsw--dr} scheme also appears to be somewhat faster than the \textsc{jgsw--fd} algorithm, typically requiring only approximately 70\% of the run-time imposed by \textsc{jgsw--fd}. Such behaviour is not unexpected, with the additional work associated with the off-centre point-placement scheme employed in \textsc{jgsw--fd} expected to result in lower overall performance. The search for additional efficiency in the \textsc{jgsw--fd} algorithm is an interesting topic for future research.

\begin{figure*}[p]
\centering
\caption{Size-driven meshing for the \textsc{bimba} problem, showing output for the \textsc{jgsw--fd} (upper) and \textsc{jgsw--dr} (lower) algorithms. Meshes were built using: (i) graded, $g$-Lipschitz smooth mesh size functions $\bar{h}\left(\mathbf{x}\right)$, with $g_{i}\in\left\{\nicefrac{3}{10},\nicefrac{2}{10},\nicefrac{1}{10}\right\}$ from left to right, (ii) tight constraints on element shape-quality, such that $\theta_{\text{min}}\geq 30^\circ$, and (iii) non-uniform surface discretisation functions $\bar{\epsilon} = \left(\nicefrac{1}{4}\right)\, \bar{h}\left(\mathbf{x}\right)$. Element counts $|\TS|$, and algorithm run-times $(\mathrm{t})$ are included for each case. Normalised histograms of element area-length ratio $a\left(f\right)$, plane-angle $\theta\left(f\right)$ and relative-length $\reledge$ are illustrated.}

\label{figure_surf_fdvsdr_bimba}

{
\footnotesize
\tabulinesep=1pt

\medskip

\begin{tabu} {c|c|c}

\parbox[b][1em][b]{.300\textwidth}{\center (\textsc{jgsw--fd}): $|\TS|=41,404$, $\mathrm{t}=3.10\mathrm{s}$} &
\parbox[b][1em][b]{.300\textwidth}{\center (\textsc{jgsw--fd}): $|\TS|=48,860$, $\mathrm{t}=3.70\mathrm{s}$} &
\parbox[b][1em][b]{.300\textwidth}{\center (\textsc{jgsw--fd}): $|\TS|=69,422$, $\mathrm{t}=5.35\mathrm{s}$} \\

\begin{minipage}[c]{.300\textwidth}
\begin{center}
\includegraphics[width=5.50cm]{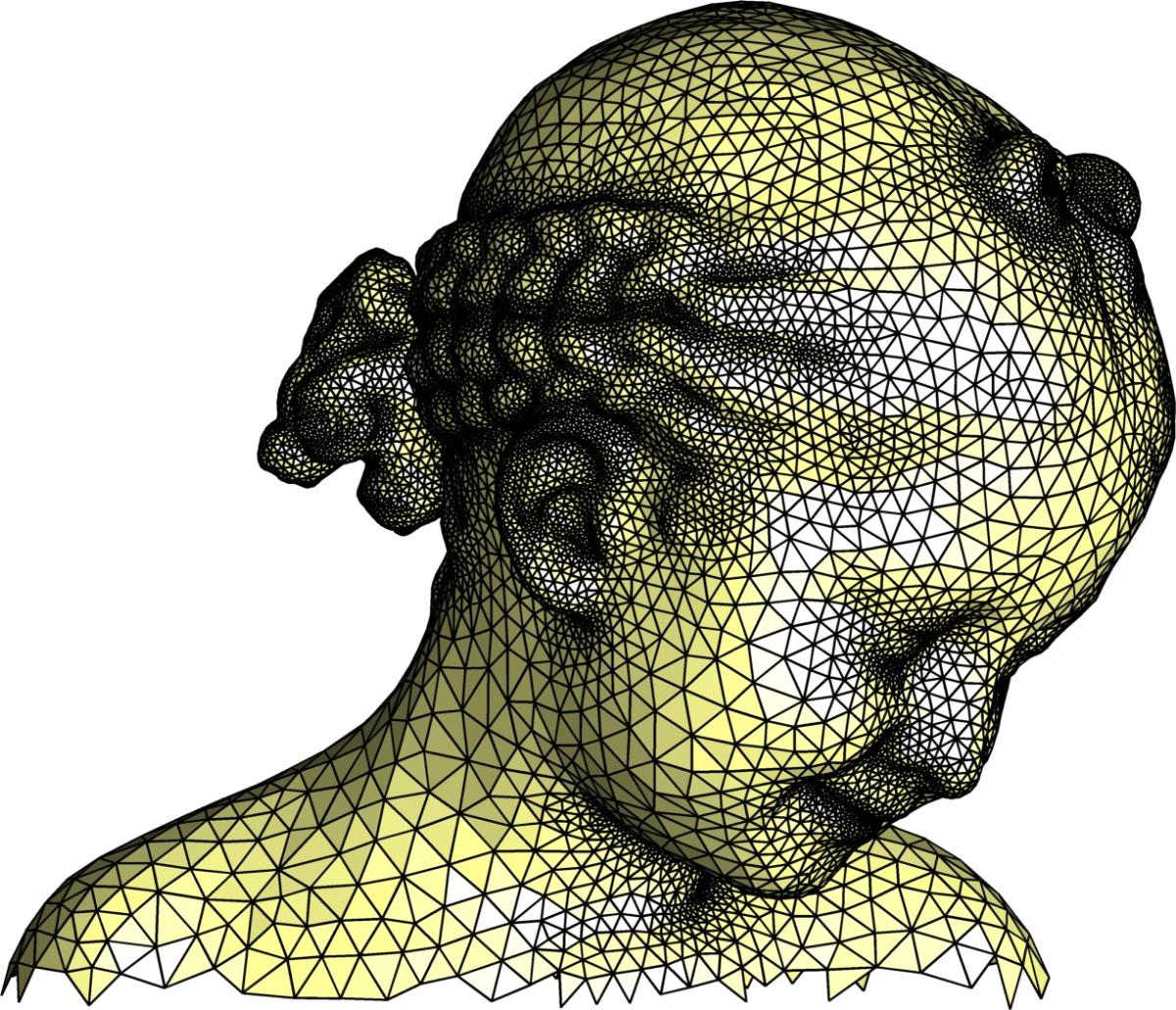} 
\end{center}
\end{minipage} &
\begin{minipage}[c]{.300\textwidth}
\begin{center}
\includegraphics[width=5.50cm]{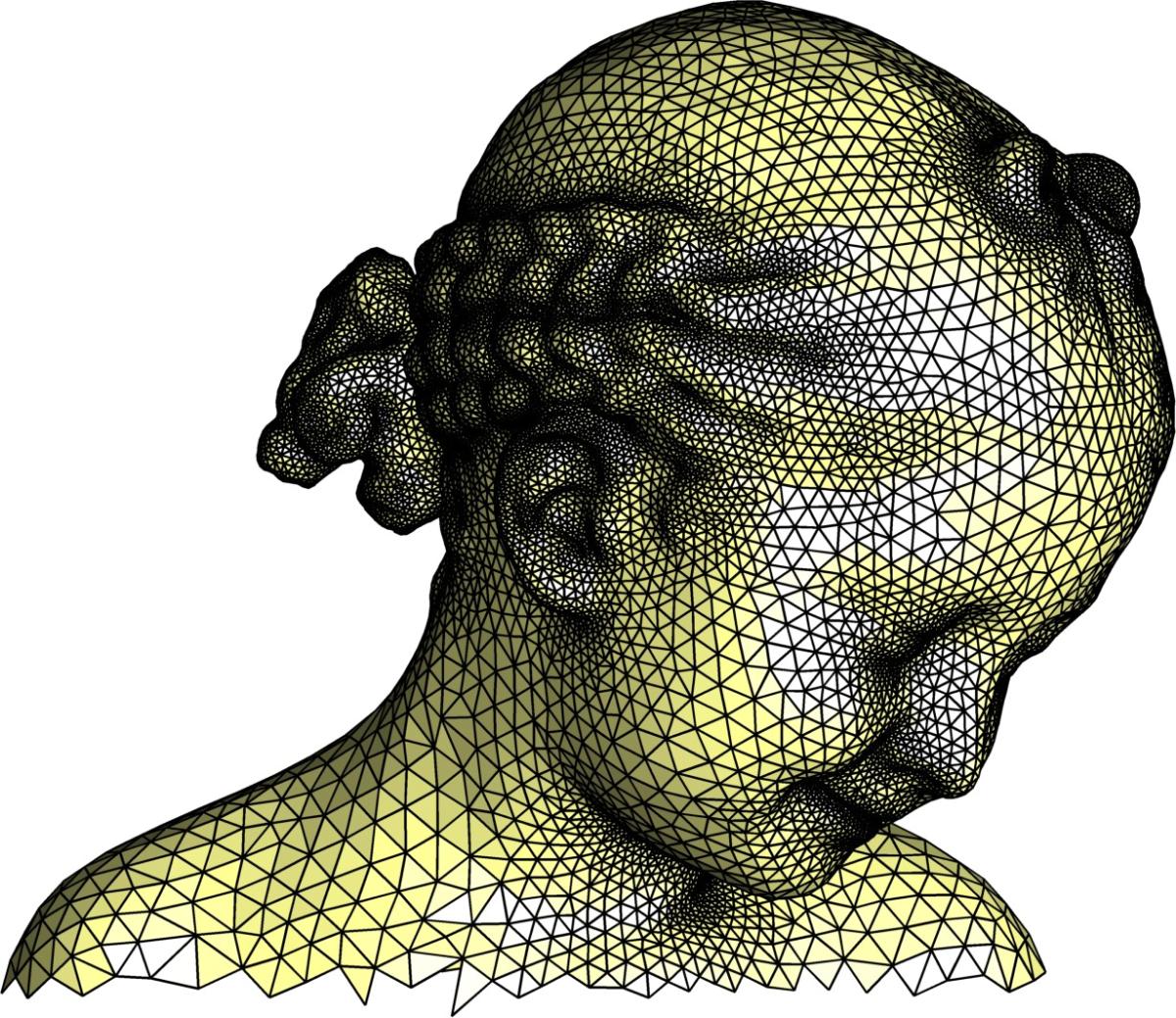}
\end{center}
\end{minipage} &
\begin{minipage}[c]{.300\textwidth}
\begin{center}
\includegraphics[width=5.50cm]{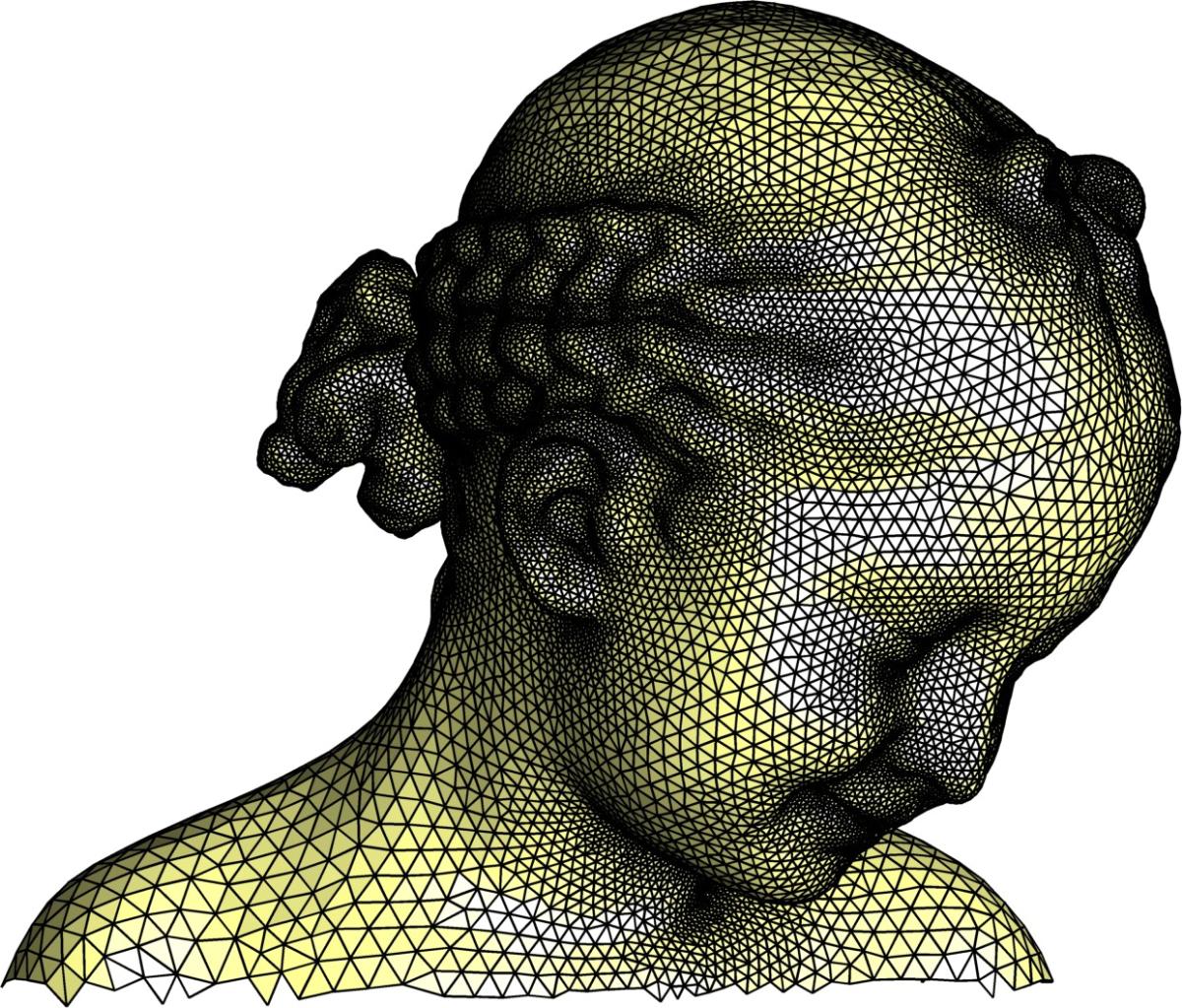}
\end{center}
\end{minipage}
\rule{0pt}{\tablestrutsize}\rule[-\tablestrutsize]{0pt}{\tablestrutsize} \\

\includegraphics[width=5.25cm]{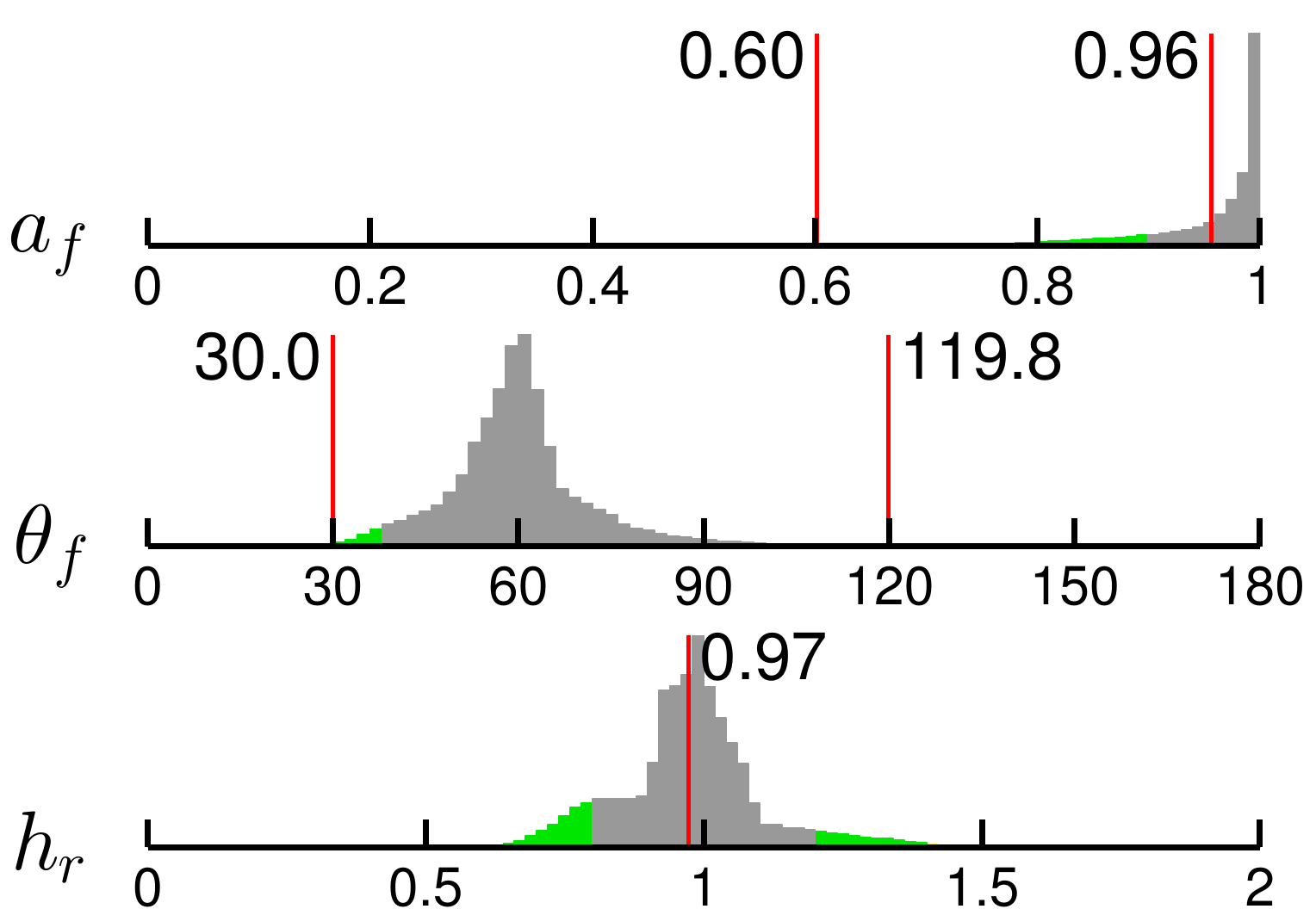} &
\includegraphics[width=5.25cm]{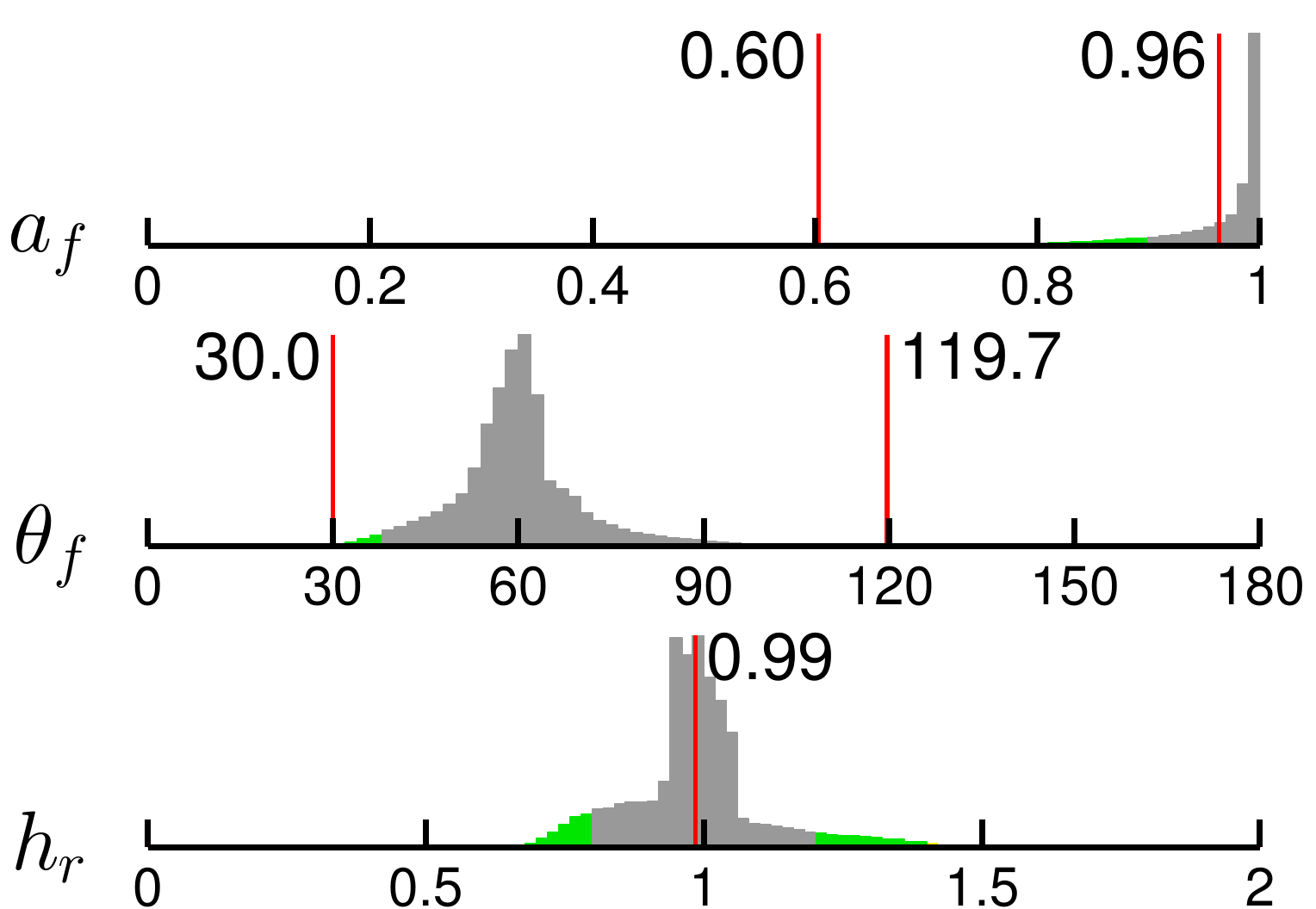} &
\includegraphics[width=5.25cm]{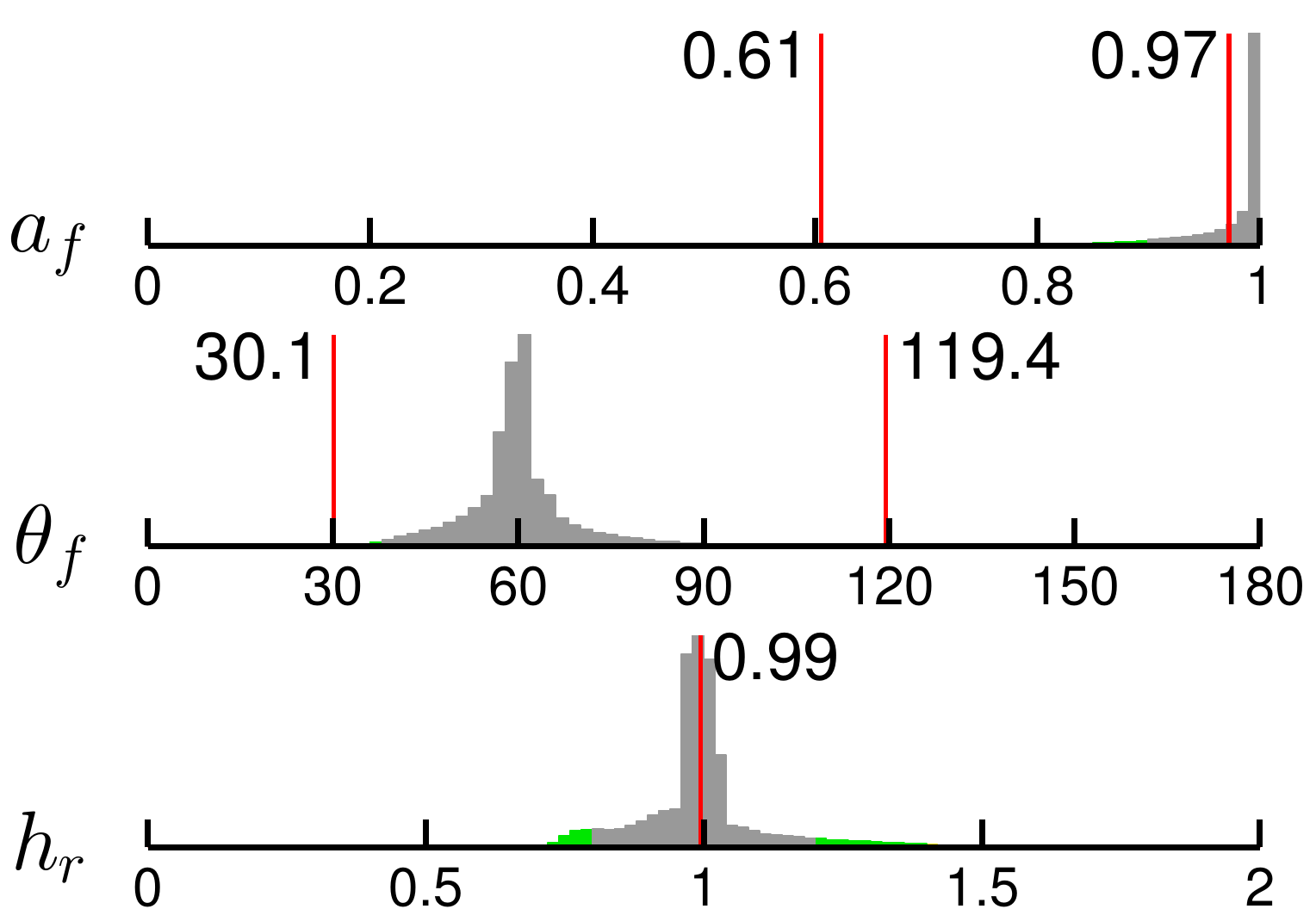}

\\ \hline

\parbox[b][1em][b]{.300\textwidth}{\center (\textsc{jgsw--dr}): $|\TS|=44,544$, $\mathrm{t}=1.97\mathrm{s}$} &
\parbox[b][1em][b]{.300\textwidth}{\center (\textsc{jgsw--dr}): $|\TS|=52,082$, $\mathrm{t}=2.31\mathrm{s}$} &
\parbox[b][1em][b]{.300\textwidth}{\center (\textsc{jgsw--dr}): $|\TS|=73,088$, $\mathrm{t}=3.25\mathrm{s}$} \\

\begin{minipage}[c]{.300\textwidth}
\begin{center}
\includegraphics[width=5.50cm]{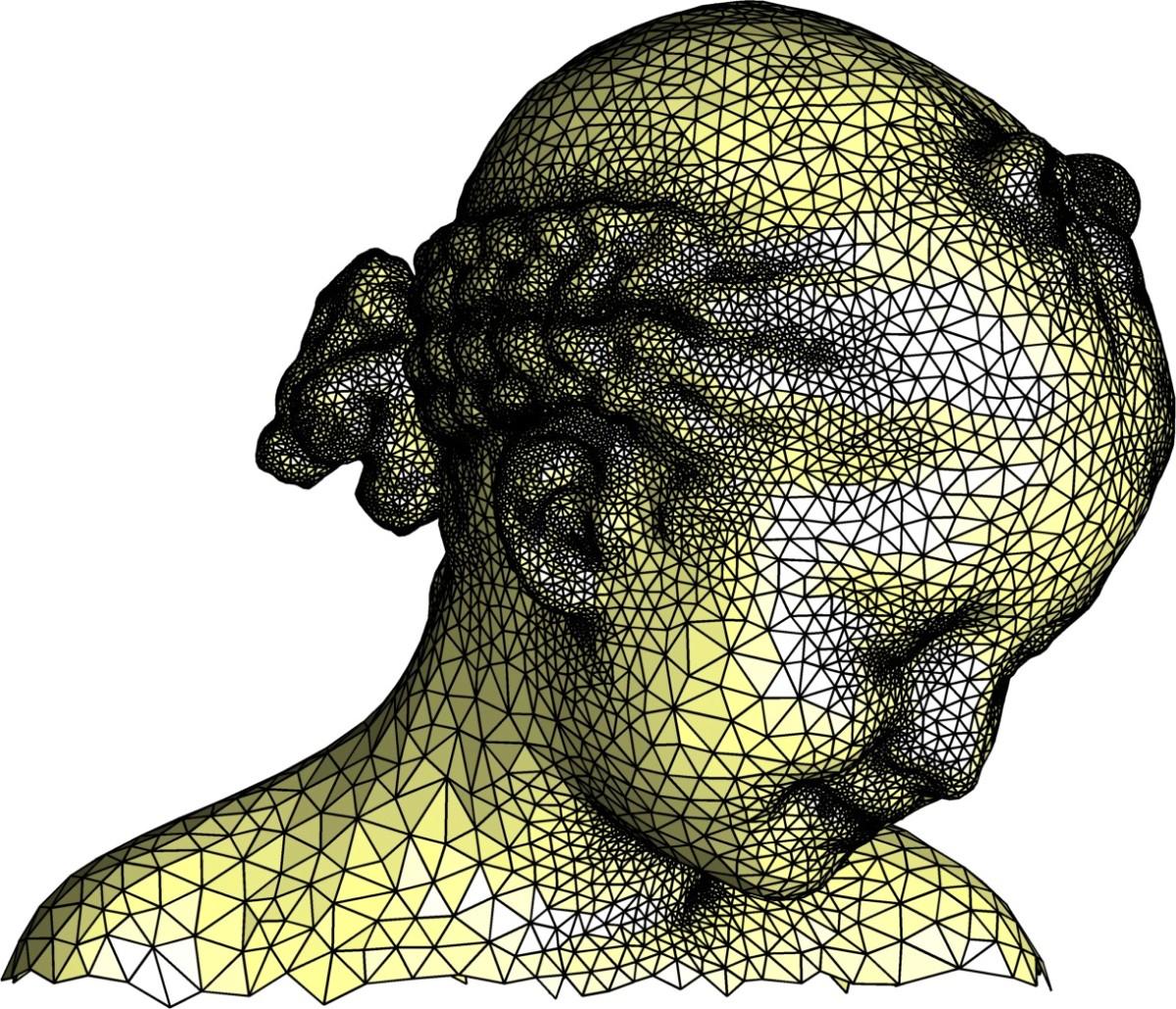} 
\end{center}
\end{minipage} &
\begin{minipage}[c]{.300\textwidth}
\begin{center}
\includegraphics[width=5.50cm]{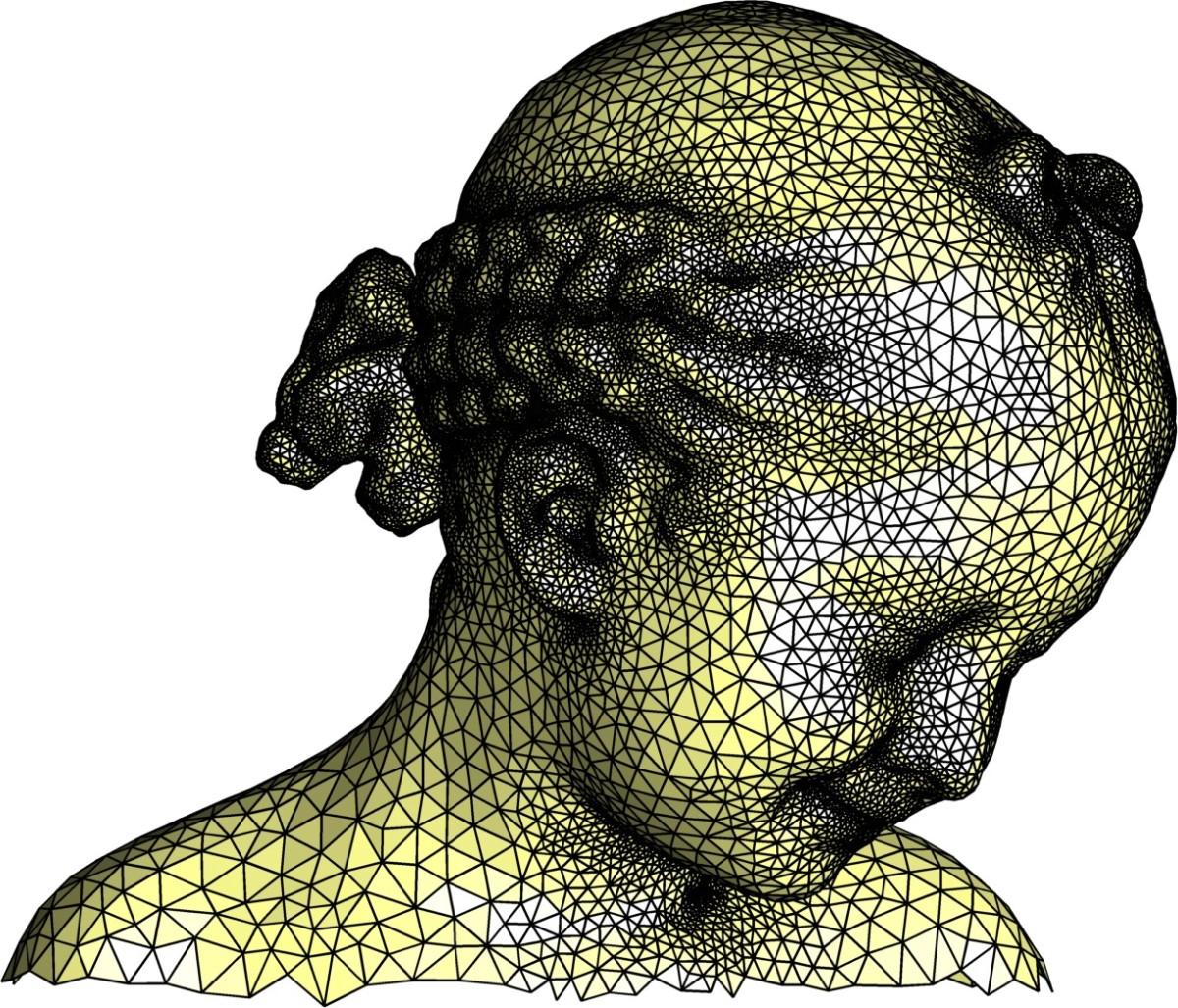}
\end{center}
\end{minipage} &
\begin{minipage}[c]{.300\textwidth}
\begin{center}
\includegraphics[width=5.50cm]{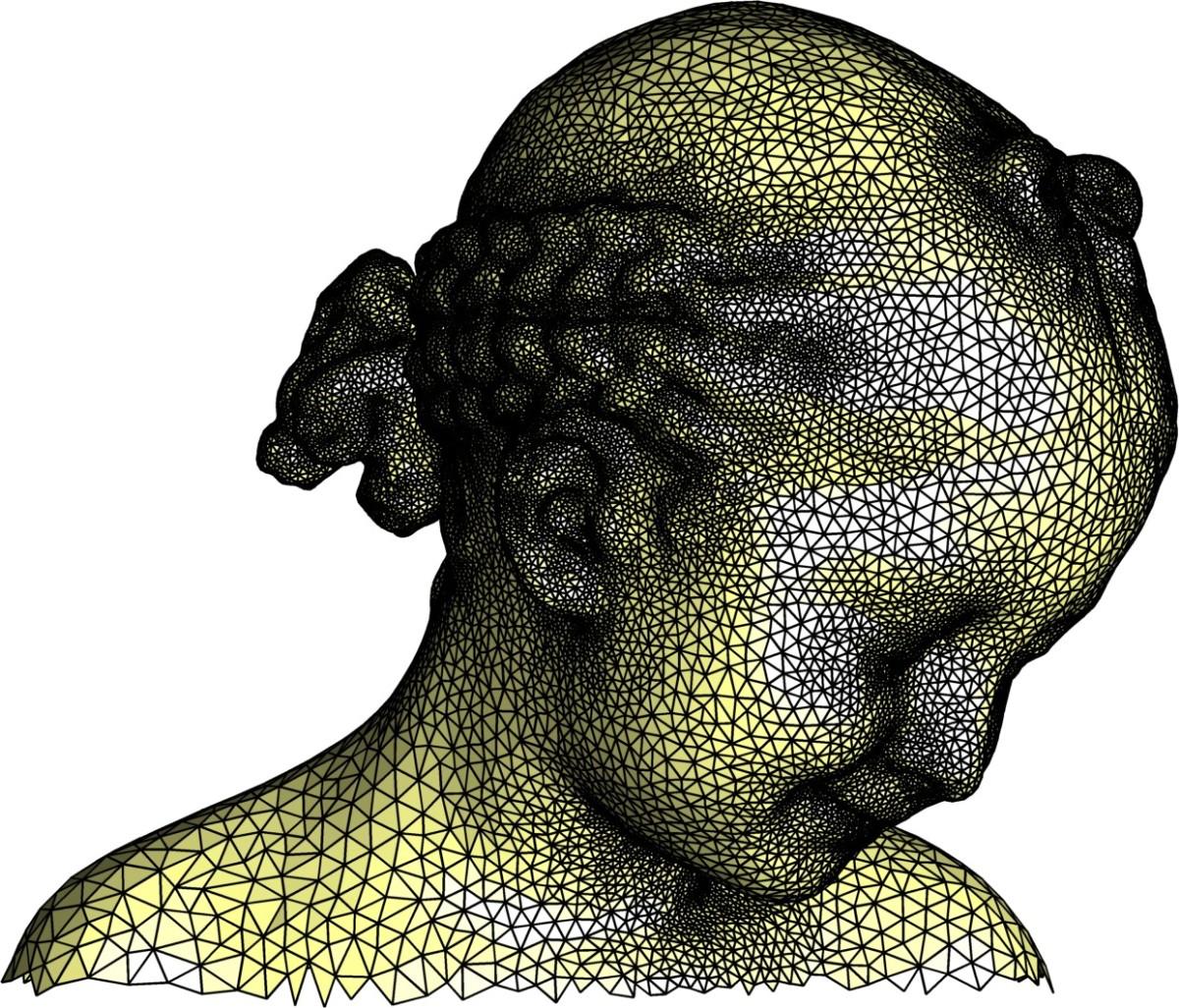}
\end{center}
\end{minipage}
\rule{0pt}{\tablestrutsize}\rule[-\tablestrutsize]{0pt}{\tablestrutsize} \\

\includegraphics[width=5.25cm]{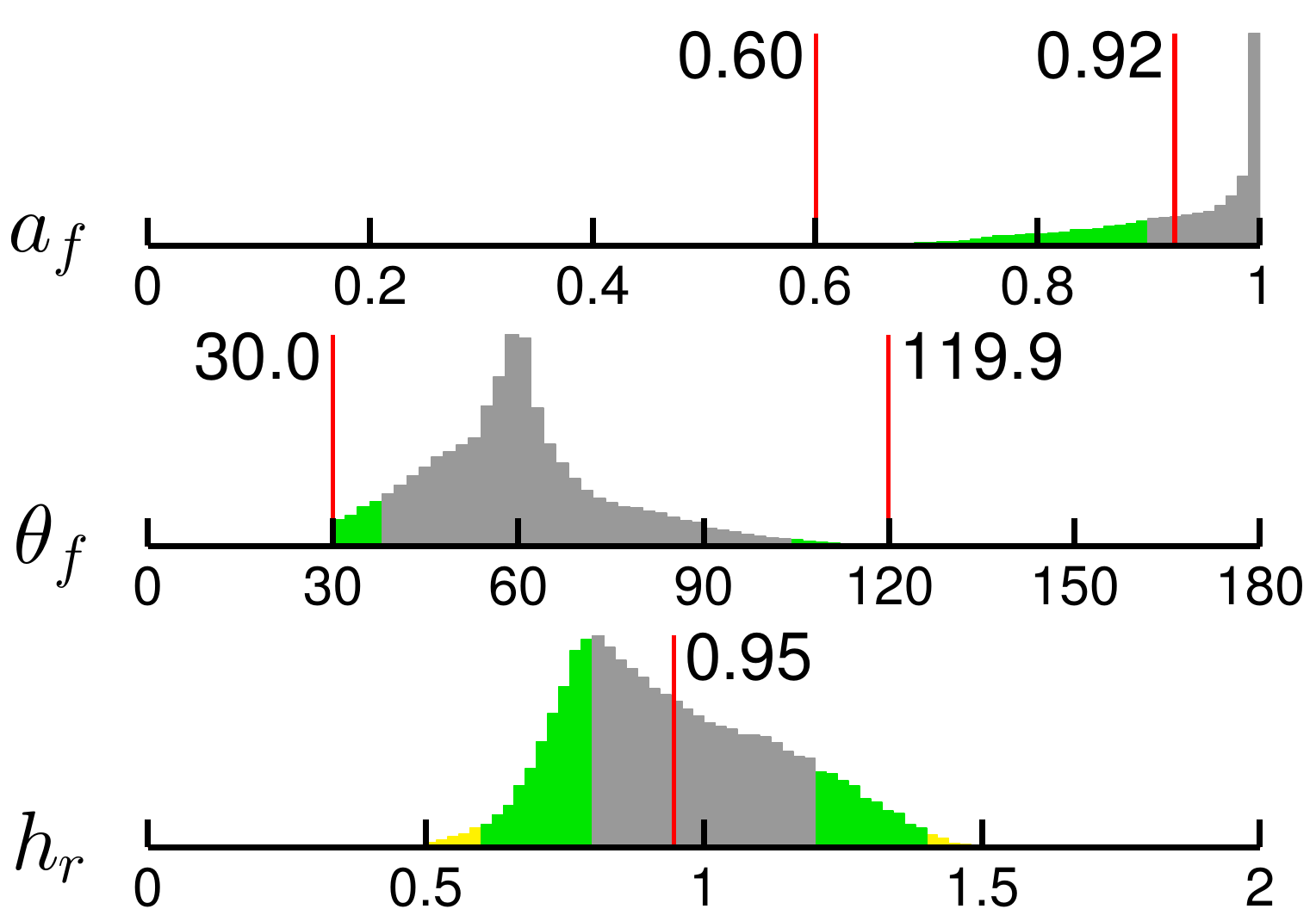} &
\includegraphics[width=5.25cm]{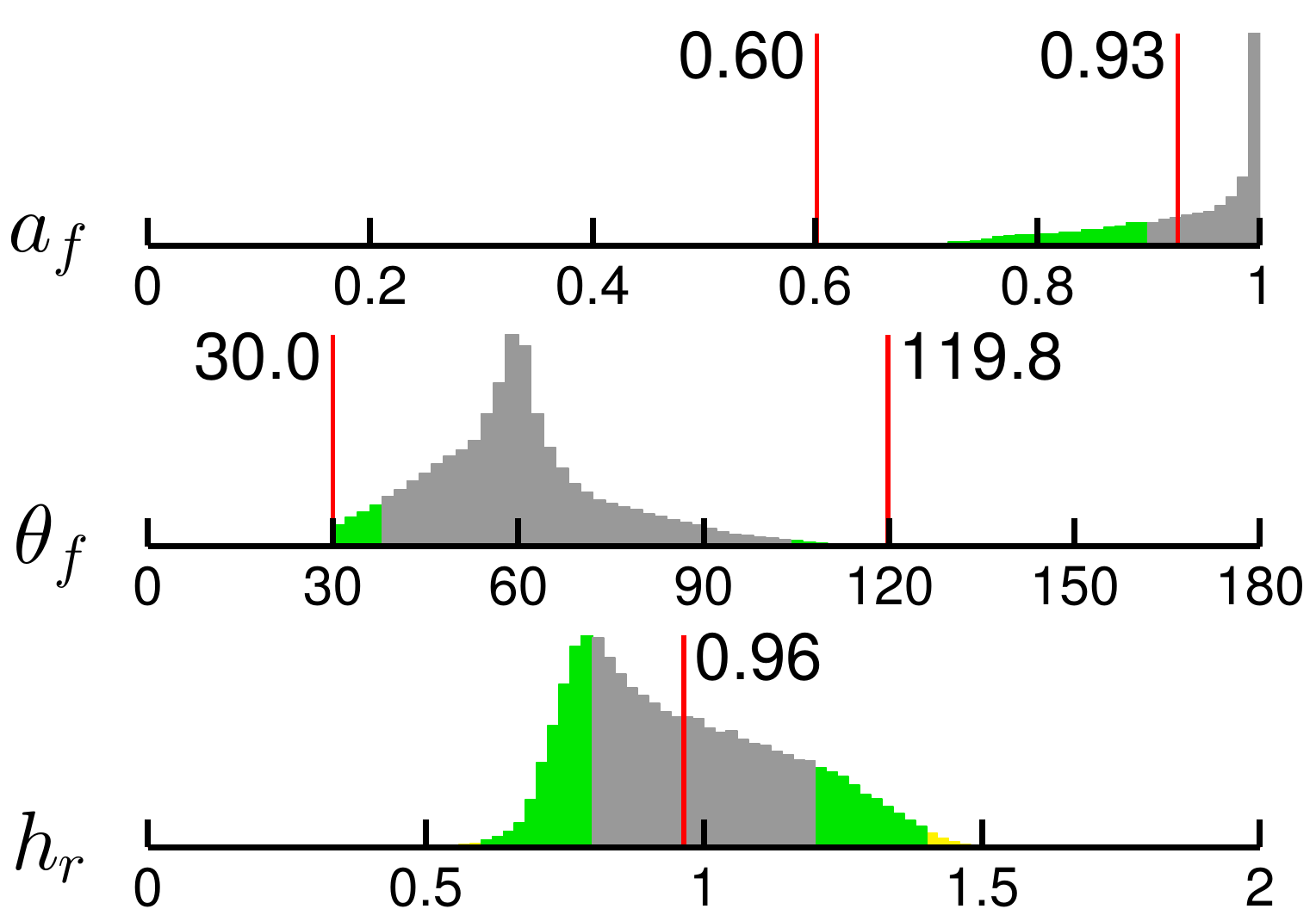} &
\includegraphics[width=5.25cm]{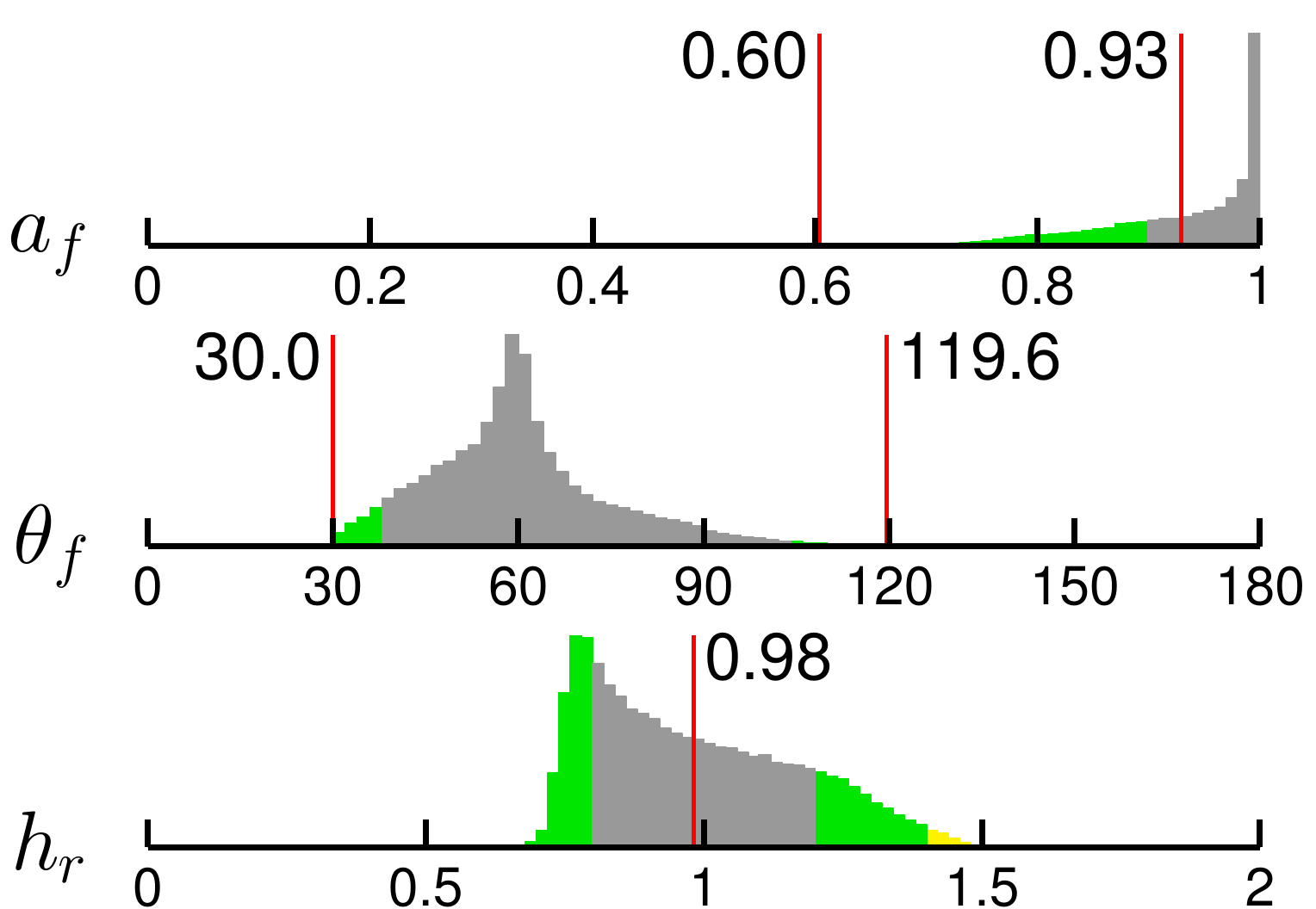}

\end{tabu}}
\end{figure*}

\begin{figure*}[p]
\centering
\caption{Size-driven meshing for the \textsc{kiss} problem, showing output for the \textsc{jgsw--fd} (upper) and \textsc{jgsw--dr} (lower) algorithms. Meshes were built using: (i) graded, $g$-Lipschitz smooth mesh size functions $\bar{h}\left(\mathbf{x}\right)$, with $g_{i}\in\left\{\nicefrac{3}{10},\nicefrac{2}{10},\nicefrac{1}{10}\right\}$ from left to right, (ii) tight constraints on element shape-quality, such that $\theta_{\text{min}}\geq 30^\circ$, and (iii) non-uniform surface discretisation functions $\bar{\epsilon} = \left(\nicefrac{1}{4}\right)\, \bar{h}\left(\mathbf{x}\right)$. Element counts $|\TS|$, and algorithm run-times $(\mathrm{t})$ are included for each case. Normalised histograms of element area-length ratio $a\left(f\right)$, plane-angle $\theta\left(f\right)$ and relative-length $\reledge$ are illustrated.}

\label{figure_surf_fdvsdr_kiss}

{
\footnotesize
\tabulinesep=1pt

\medskip

\begin{tabu} {c|c|c}

\parbox[b][1em][b]{.300\textwidth}{\center (\textsc{jgsw--fd}): $|\TS|=49,278$, $\mathrm{t}=3.20\mathrm{s}$} &
\parbox[b][1em][b]{.300\textwidth}{\center (\textsc{jgsw--fd}): $|\TS|=60,814$, $\mathrm{t}=3.98\mathrm{s}$} &
\parbox[b][1em][b]{.300\textwidth}{\center (\textsc{jgsw--fd}): $|\TS|=91,406$, $\mathrm{t}=6.10\mathrm{s}$} \\

\begin{minipage}[c]{.300\textwidth}
\begin{center}
\includegraphics[height=5.50cm]{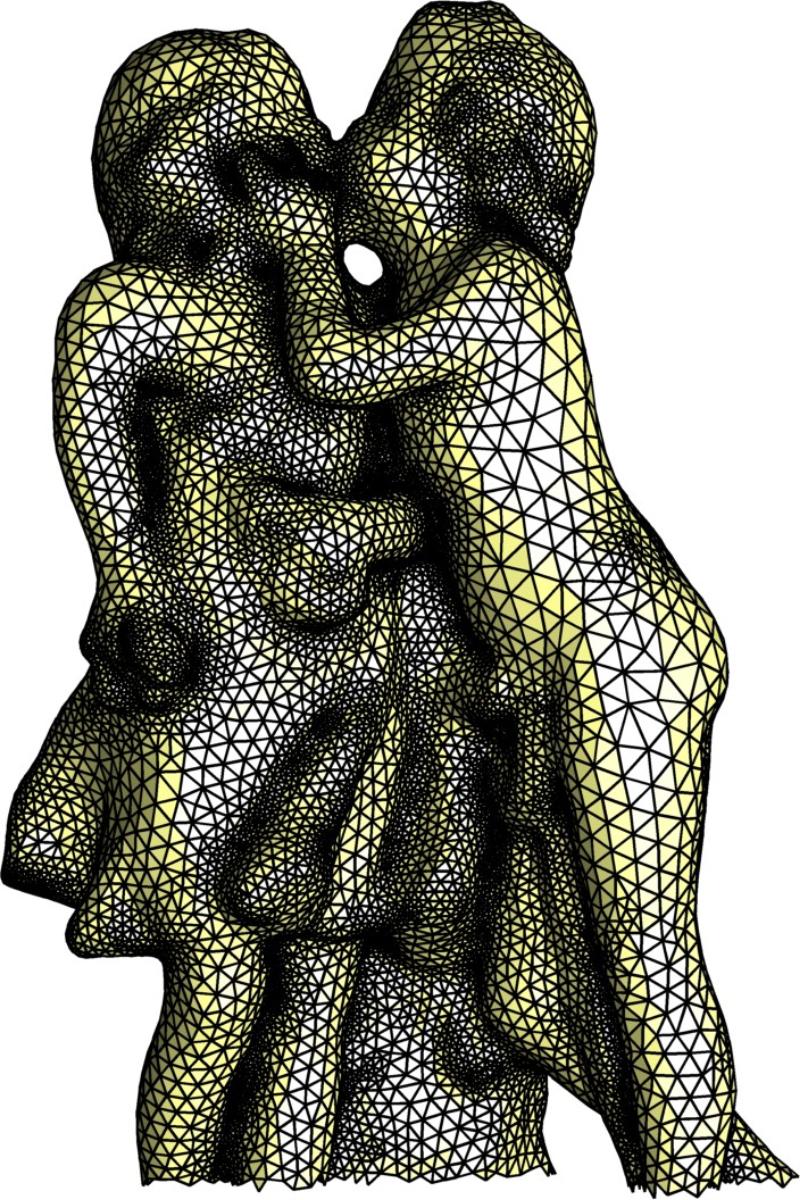} 
\end{center}
\end{minipage} &
\begin{minipage}[c]{.300\textwidth}
\begin{center}
\includegraphics[height=5.50cm]{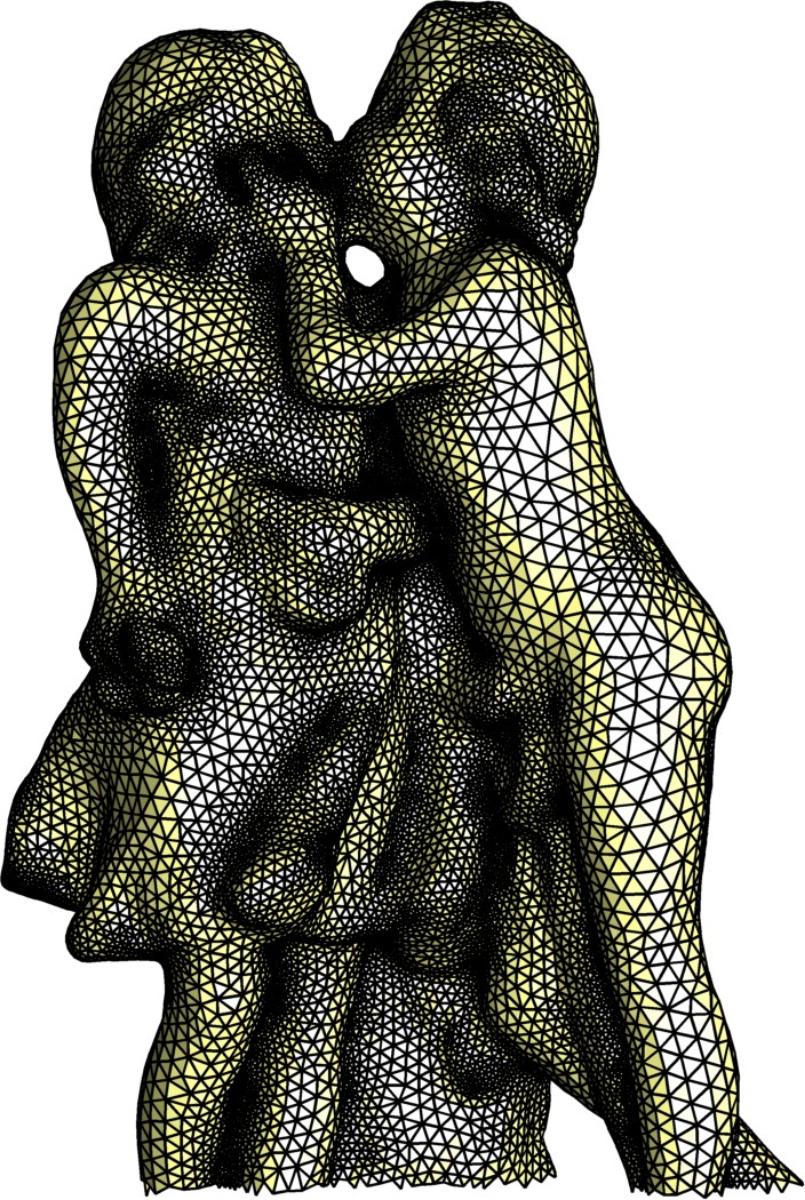}
\end{center}
\end{minipage} &
\begin{minipage}[c]{.300\textwidth}
\begin{center}
\includegraphics[height=5.50cm]{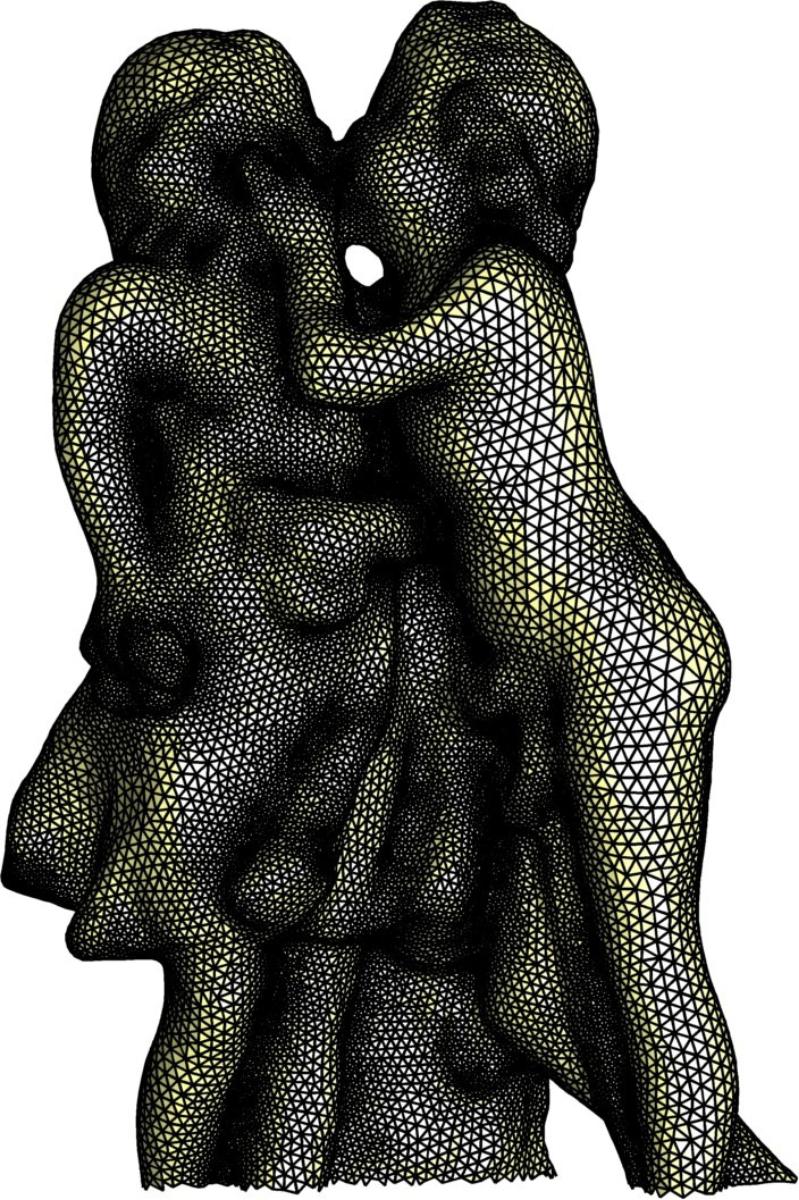}
\end{center}
\end{minipage}
\rule{0pt}{\tablestrutsize}\rule[-\tablestrutsize]{0pt}{\tablestrutsize} \\

\includegraphics[width=5.25cm]{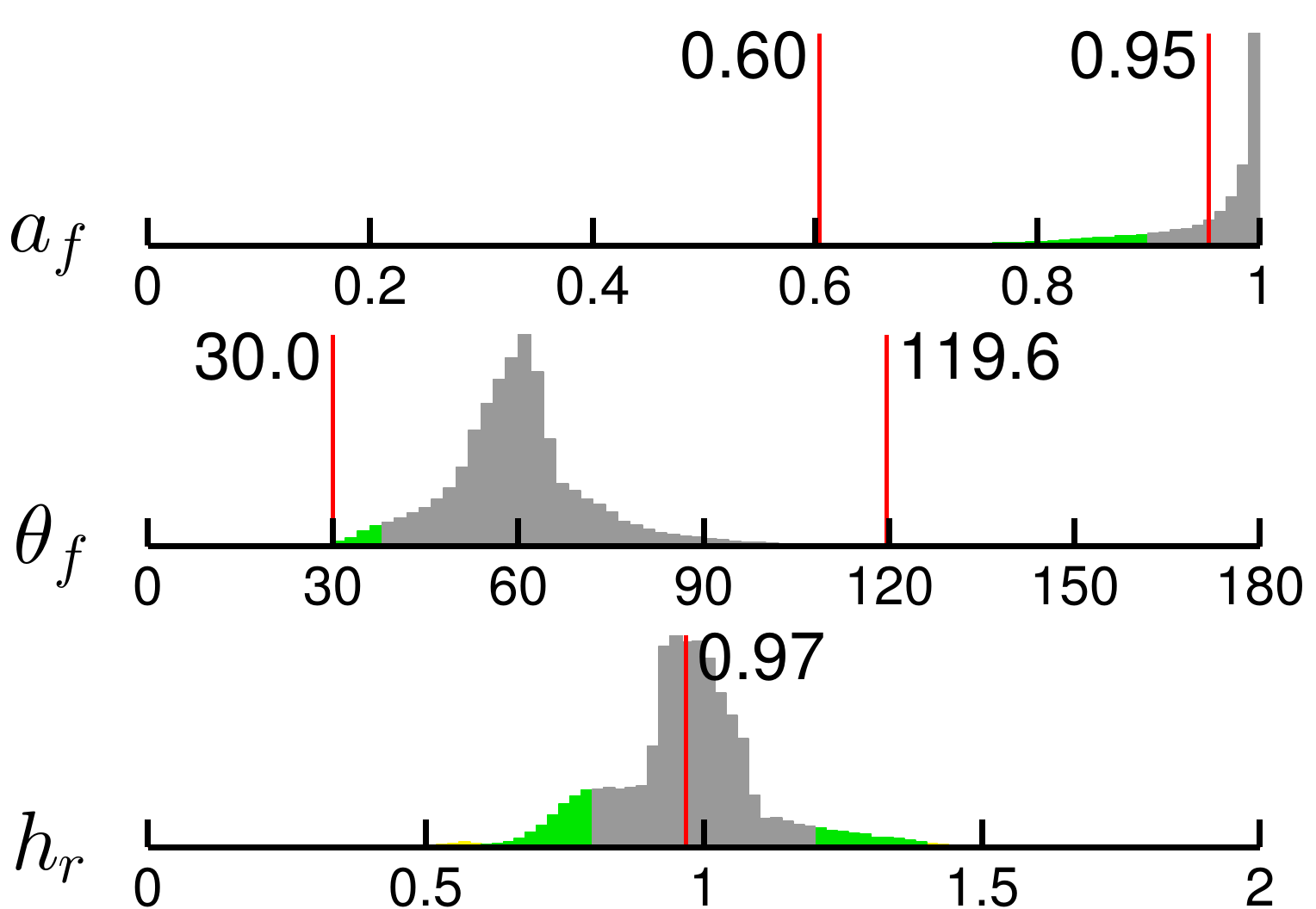} &
\includegraphics[width=5.25cm]{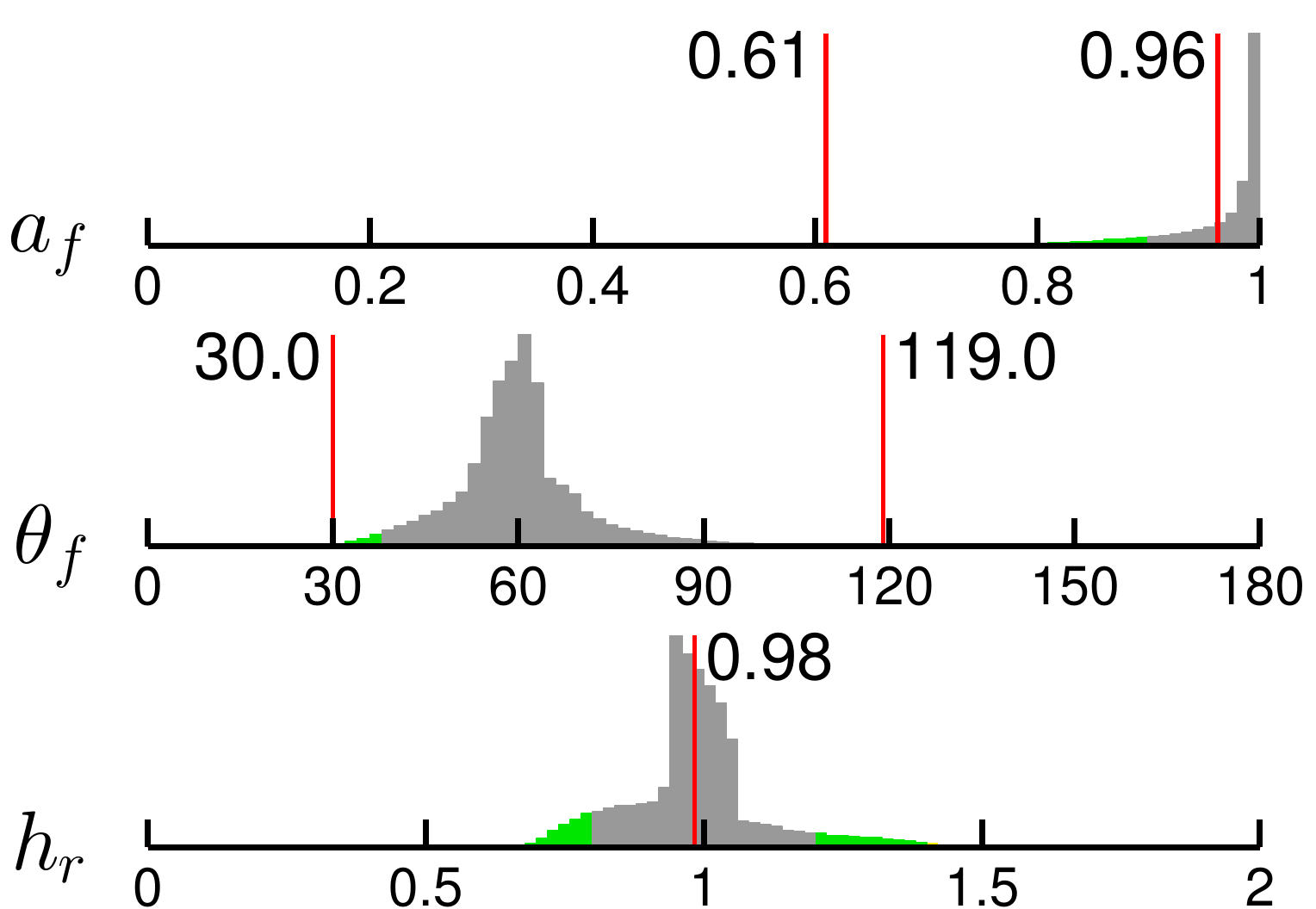} &
\includegraphics[width=5.25cm]{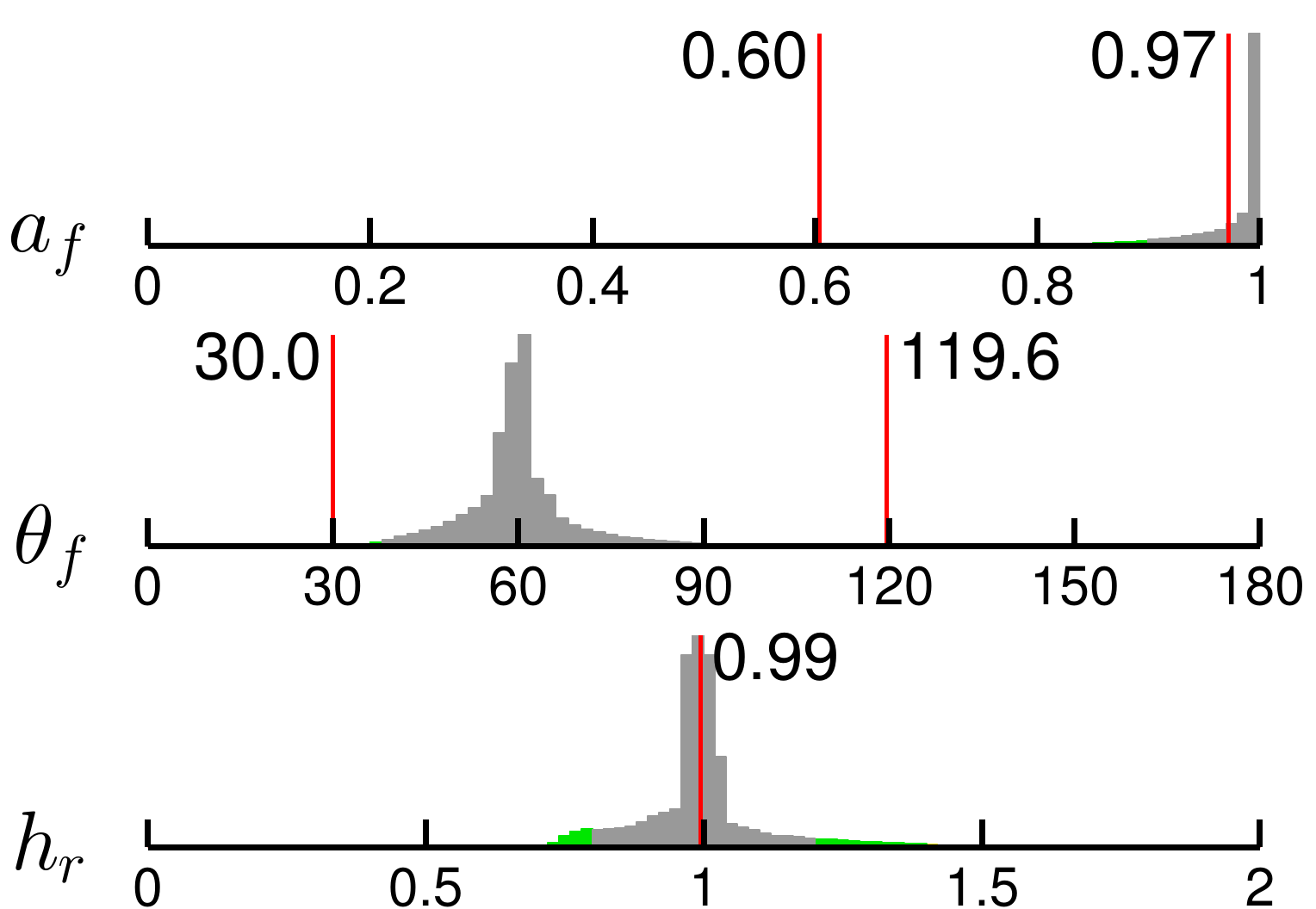}

\\ \hline

\parbox[b][1em][b]{.300\textwidth}{\center (\textsc{jgsw--dr}): $|\TS|=52,696$, $\mathrm{t}=2.19\mathrm{s}$} &
\parbox[b][1em][b]{.300\textwidth}{\center (\textsc{jgsw--dr}): $|\TS|=64,460$, $\mathrm{t}=2.71\mathrm{s}$} &
\parbox[b][1em][b]{.300\textwidth}{\center (\textsc{jgsw--dr}): $|\TS|=96,408$, $\mathrm{t}=4.01\mathrm{s}$} \\

\begin{minipage}[c]{.300\textwidth}
\begin{center}
\includegraphics[height=5.50cm]{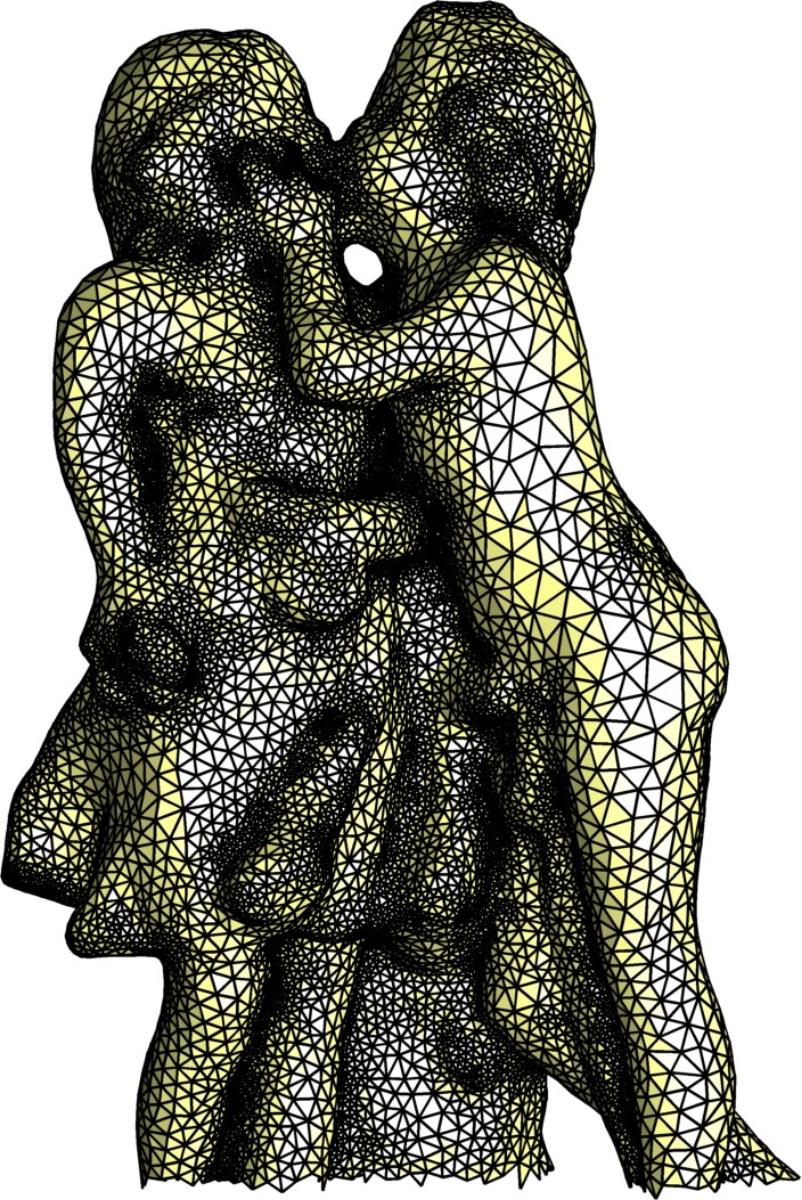} 
\end{center}
\end{minipage} &
\begin{minipage}[c]{.300\textwidth}
\begin{center}
\includegraphics[height=5.50cm]{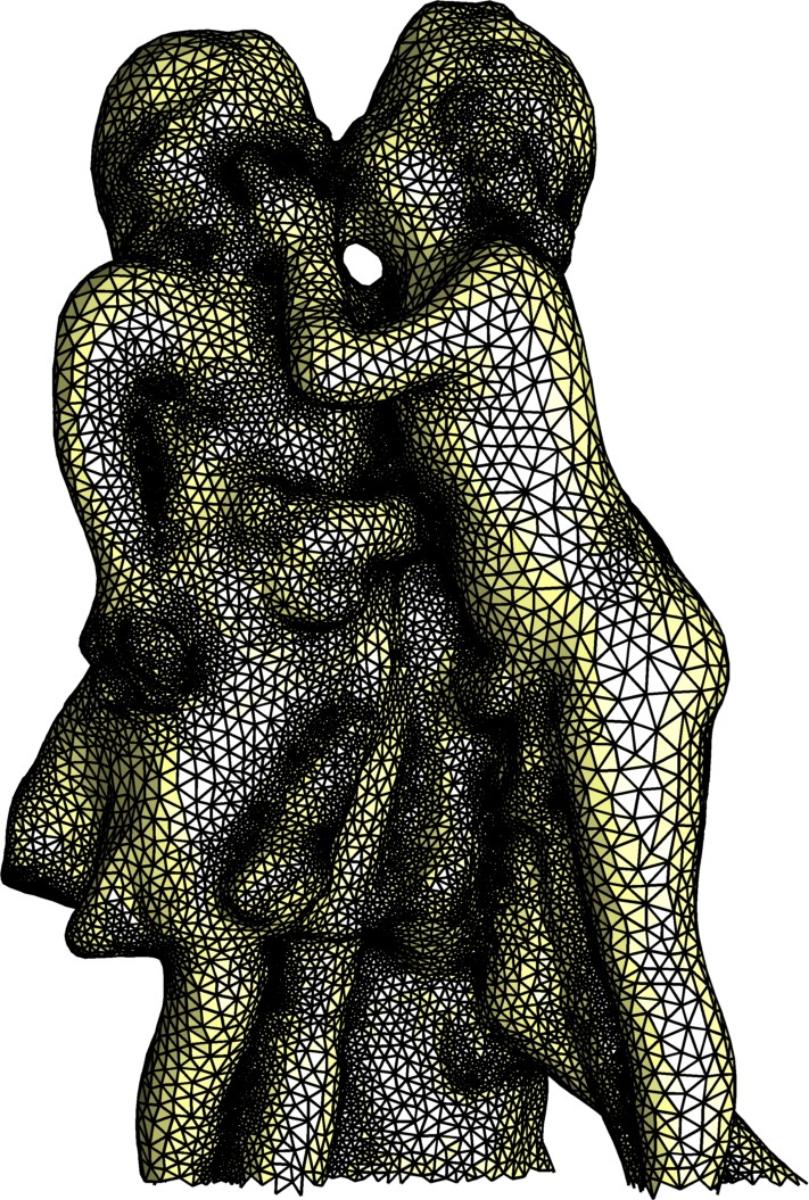}
\end{center}
\end{minipage} &
\begin{minipage}[c]{.300\textwidth}
\begin{center}
\includegraphics[height=5.50cm]{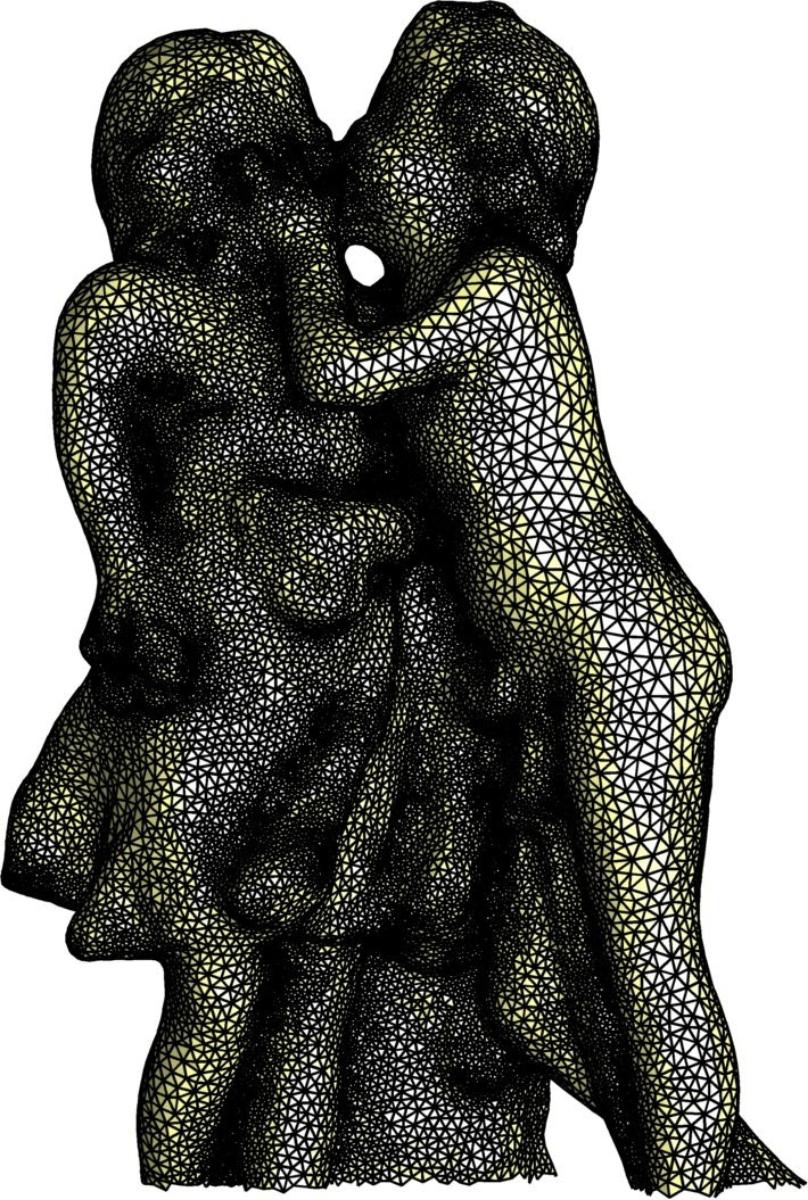}
\end{center}
\end{minipage}
\rule{0pt}{\tablestrutsize}\rule[-\tablestrutsize]{0pt}{\tablestrutsize} \\

\includegraphics[width=5.25cm]{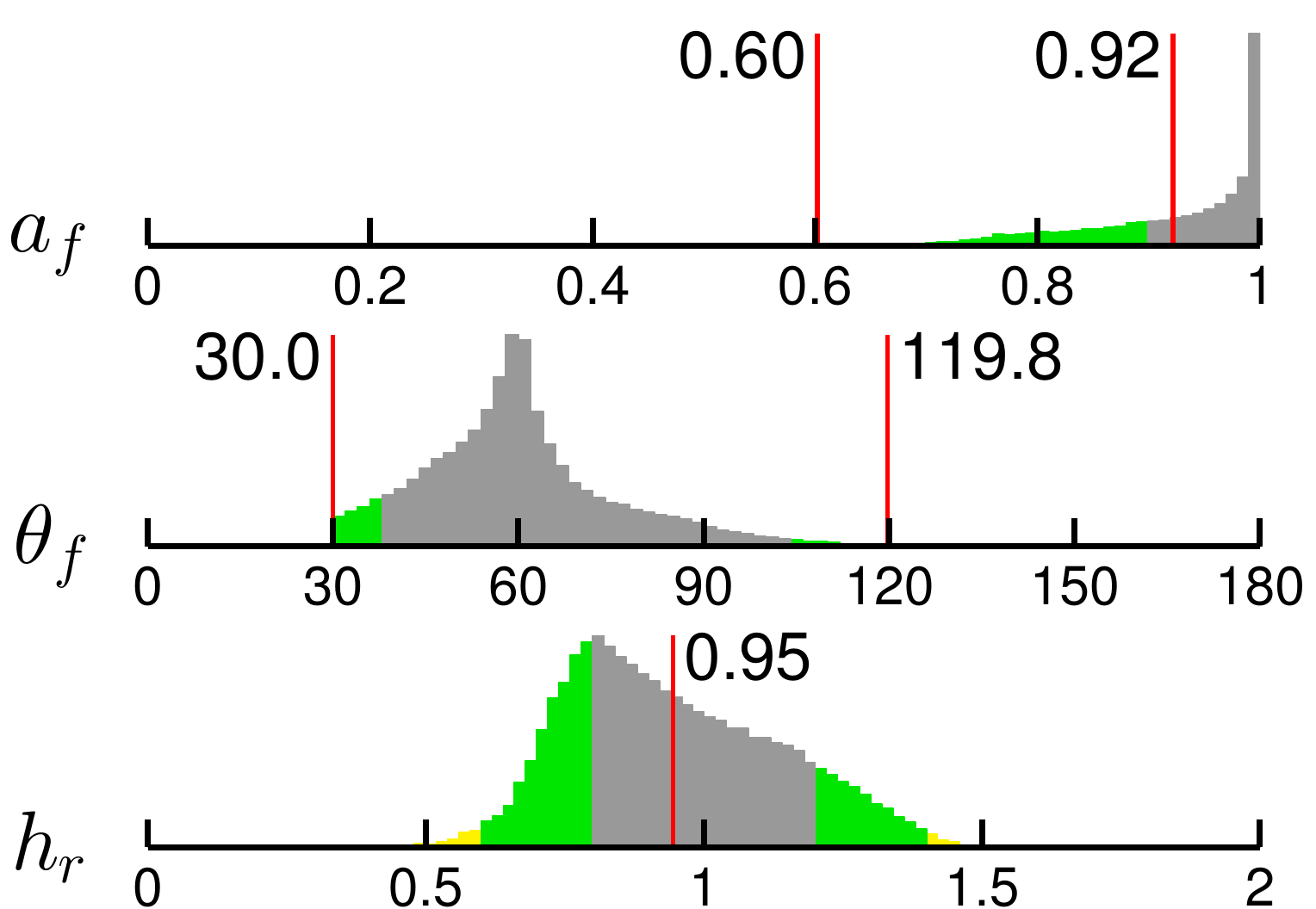} &
\includegraphics[width=5.25cm]{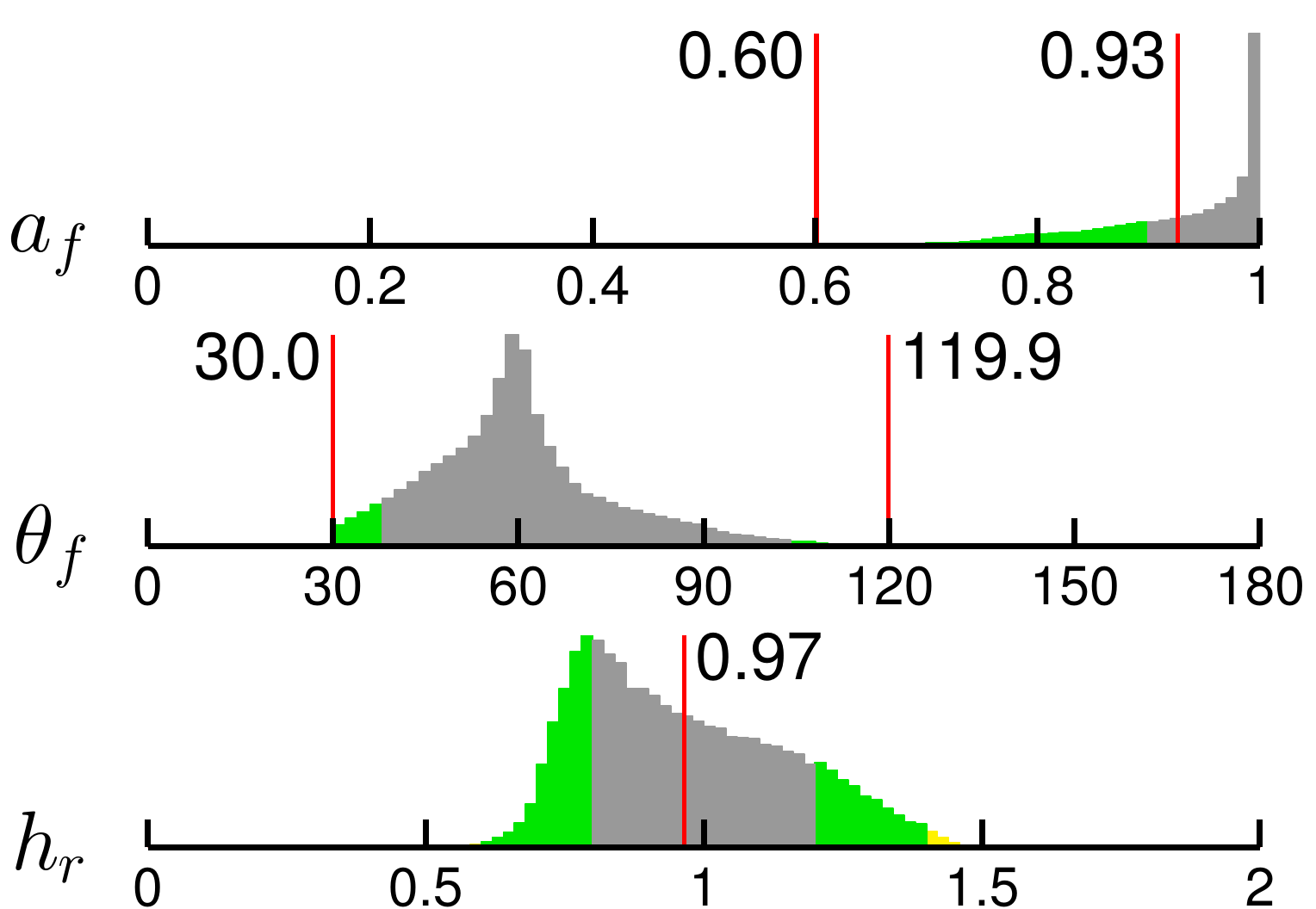} &
\includegraphics[width=5.25cm]{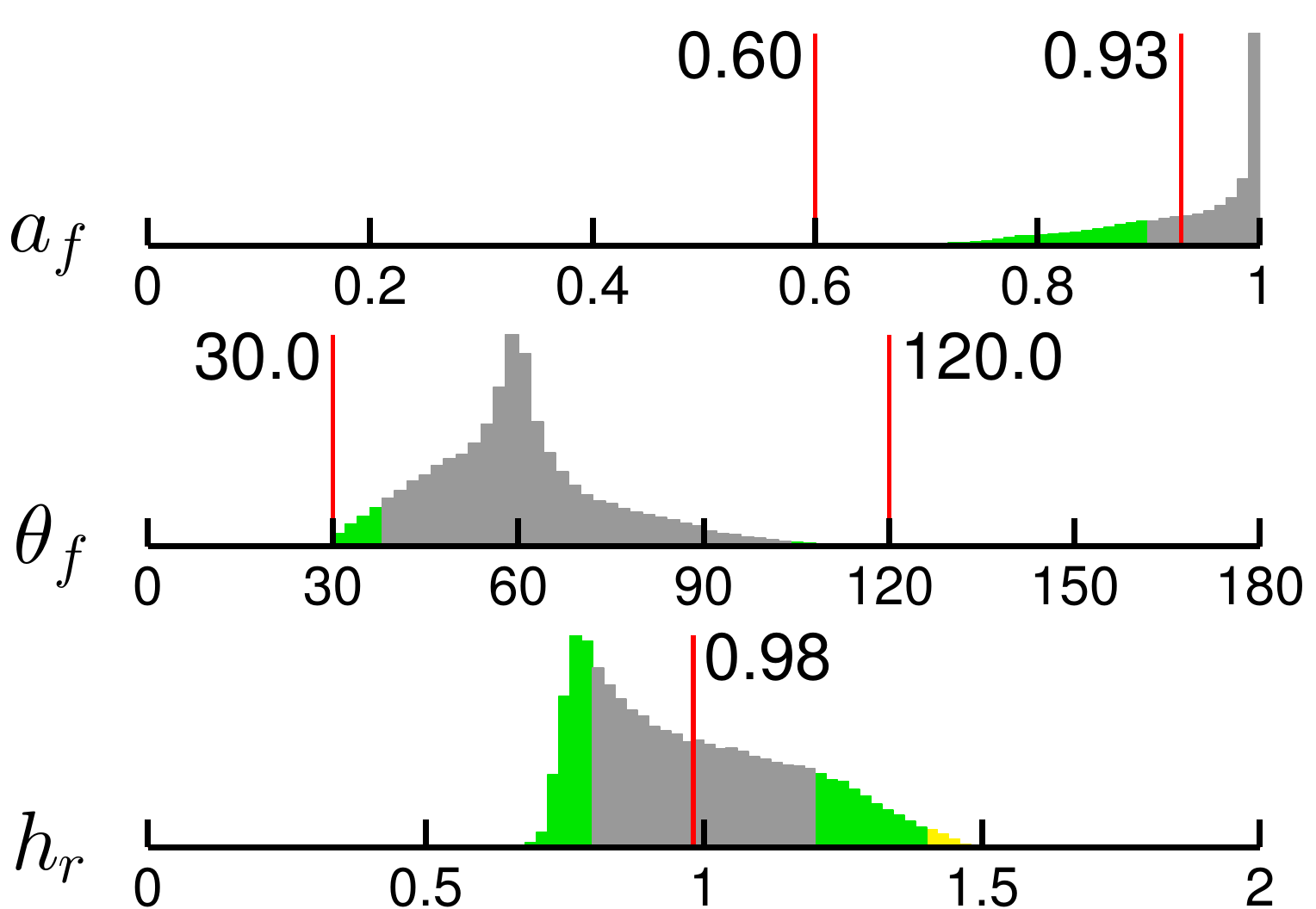}

\end{tabu}}
\end{figure*}

\begin{table*}[t]
\centering
\caption{A comparison of the mean absolute deviation in the element angle distributions $\MAD$ for the various combinations of refinement algorithms and benchmark problems considered in this study. Smaller values correspond to improved mean mesh quality. \textsc{cgal-dr} denotes the Delaunay-refinement algorithm included in the \textsc{cgal} package, and \textsc{jgsw-dr}/\textsc{jgsw-fd} denote the Delaunay-refinement and Frontal-Delaunay algorithms implemented in the current study. The rightmost column indicates the reduction in absolute deviation achieved using the Frontal-Delaunay algorithm.}

\label{table_mad}

{
\tabulinesep=1pt

\smallskip

\begin{tabu} {|c|c||c|c|c|c|}

\hline
& $g$ & \textsc{cgal-dr} & \textsc{jgsw-dr} & \textsc{jgsw-fd} & \textsc{factor} \\
\hline
\textsc{elephant} & -- & $11.96^\circ$ & -- & $4.35^\circ$ & 2.8 \\
\textsc{hip}      & -- & $11.90^\circ$ & -- & $3.89^\circ$ & 3.1 \\
\textsc{bunny}    & -- & $12.02^\circ$ & -- & $4.10^\circ$ & 2.9 \\
\hline
\textsc{ifp2}     & -- & -- & $10.72^\circ$ & $5.14^\circ$ & 2.1 \\
\textsc{femur}    & -- & -- & $10.93^\circ$ & $3.85^\circ$ & 2.8 \\
\textsc{rocker}   & -- & -- & $10.69^\circ$ & $3.92^\circ$ & 2.7 \\
\hline
\textsc{bimba} & $\nicefrac{3}{10}$ & -- & $11.25^\circ$ & $7.83^\circ$ & 1.4 \\
               & $\nicefrac{2}{10}$ & -- & $10.99^\circ$ & $7.04^\circ$ & 1.6 \\
               & $\nicefrac{1}{10}$ & -- & $10.74^\circ$ & $5.68^\circ$ & 1.9 \\
\hline
\textsc{kiss}  & $\nicefrac{3}{10}$ & -- & $11.37^\circ$ & $8.15^\circ$ & 1.4 \\
               & $\nicefrac{2}{10}$ & -- & $10.98^\circ$ & $7.16^\circ$ & 1.5 \\
               & $\nicefrac{1}{10}$ & -- & $10.74^\circ$ & $5.74^\circ$ & 1.9 \\
\hline

\end{tabu}}
\end{table*}

\subsection{Non-uniform Mesh-Size Constraints}

A sequence of meshes for the \textsc{bimba} and \textsc{kiss} problems are presented in Figures~\ref{figure_surf_fdvsdr_bimba} and~\ref{figure_surf_fdvsdr_kiss}, examining the impact of non-uniform mesh-size constraints on algorithm performance. A set of meshes were generated for both test cases using low, medium and high resolution settings, where the associated mesh size functions $\bar{h}\left(\mathbf{x}\right)$ were built using increasingly stringent gradient limits, such that $g_{i}\in\left\{\nicefrac{3}{10},\nicefrac{2}{10},\nicefrac{1}{10}\right\}$ respectively. Analysis of Figures~\ref{figure_surf_fdvsdr_bimba} and \ref{figure_surf_fdvsdr_kiss} show that both the \textsc{jgsw--fd} and \textsc{jgsw--dr} algorithms generate high-quality meshes for all test cases -- satisfying the required element quality, size and surface error thresholds. Analysis of the area-length and plane-angle distributions again shows that the \textsc{jgsw--fd} algorithm consistently outperforms the \textsc{jgsw--dr} scheme, generating meshes with higher mean area-length ratios in all cases. Furthermore, it is evident that the quality of meshes generated using the Frontal-Delaunay algorithm improves with increased uniformity, via $g\rightarrow 0$, as indicated by the increase in mean $a\left(f\right)$ and the narrowing of the $\theta\left(f\right)$ distribution about $60.0^\circ$. Results for the mean absolute deviation in the element angle distribution, $\MAD$ -- included in Table~\ref{table_mad}, shows that use of the \textsc{jgsw--fd} algorithm results in a reduction in the spread of the element angle distributions by a factor ranging from approximately 1.5 to 2.0.

In contrast, similar analysis shows that output generated via the \textsc{jgsw--dr} scheme to be essentially independent of sizing uniformity, with distributions of $a\left(f\right)$ and $\theta\left(f\right)$ showing little variation with $\bar{h}\left(\mathbf{x}\right)$. Visually, the enhanced quality of the meshes generated using the Frontal-Delaunay algorithm is again evident, with all test cases showing a marked increase in both smoothness and sub-structure. Given the tight bounds on radius-edge ratios in these tests, the Delaunay-refinement algorithm appears to achieve a maximum mean area-length ratio $\overline{a}\left(f\right)\simeq 0.92$. The Frontal-Delaunay algorithm is seen to generate consistently better output, with mean area-length metrics reaching $\bar{a}\left(f\right)\simeq 0.97$ at high resolution settings. 

Consistent with previous results, analysis of the distribution of relative-length ratios show that the \textsc{jgsw--fd} algorithm tightly conforms to the imposed non-uniform sizing constraints, with $\reledge$ tightly clustered about $1$. Furthermore, like the element shape-quality results discussed previously, it can be seen that this distribution improves as $g\rightarrow 0$, as indicated by the narrowing of the distribution about $\reledge=1$. In contrast, the Delaunay-refinement scheme is again shown to be consistently associated with significant sizing error, as illustrated by broad distributions of relative-length straddling $\reledge\simeq 1$. Such behaviour appears to be largely independent of the characteristics of the imposed mesh size constraints.

\section{Conclusions}
\label{section_conclusions}

A new Frontal-Delaunay algorithm has been developed to triangulate smooth surfaces embedded in $\mathbb{R}^{3}$. The new algorithm is based on the so-called \textit{restricted} Delaunay paradigm, in which a surface mesh $\DelS{X}$, conforming to an underlying surface description $\Sigma$, is constructed as a subset of a full-dimensional Delaunay tessellation $\Del{X}$. The new restricted Frontal-Delaunay algorithm is based on a generalisation of the work of Rebay \cite{Rebay93FrontalDelaunay} and \"Ung\"or \cite{Ungor09OffCenters} for planar problems, in which generalised `off-centre' Steiner vertices are inserted along edges in the Voronoi diagram. This work has been extended to support surface meshing operations through the development of a new point-placement strategy that positions vertices on facets in the underlying Voronoi complex. This new scheme allows for the insertion of both size- and shape-optimal Steiner vertices, leading to a hybrid approach that combines many of the advantages of conventional advancing-front and Delaunay-refinement techniques. A series of comparative experimental studies confirm the effectiveness of this new approach in practice, demonstrating that an improvement in mesh-quality is typically achieved when compared to conventional Delaunay-refinement schemes. Importantly, it has also been demonstrated that the new Frontal-Delaunay algorithm satisfies the same set of constraints as conventional restricted Delaunay-refinement approaches, adhering to limits on element radius-edge ratios, edge length and surface discretisation error. Results show that the new algorithm is an effective hybridisation of existing mesh generation techniques, combining the high element quality and mesh sub-structure of advancing-front techniques with the theoretical guarantees of Delaunay-refinement schemes. It is expected that applications that place a premium on mesh-quality, including problems in computational fluid dynamics and/or structural analysis, may benefit from the new Frontal-Delaunay technique. Future work should focus on support for an extended class of surface definitions, including domains containing sharp features.

\section*{Acknowledgements} This work was carried out at the University of Sydney with the support of an Australian Postgraduate Award. The authors also wish to thank the anonymous reviewers for their helpful comments and feedback.

\section*{References}

\bibliographystyle{elsarticle-num}
\bibliography{references}

\end{document}